\documentclass[12pt]{amsart}      
\usepackage{amssymb}
\usepackage{eucal}
\usepackage{amsmath}
\usepackage{amscd}
\usepackage[dvips]{color}
\usepackage{multicol}
\usepackage[all]{xy}           
\usepackage{graphicx}
\usepackage{color}
\usepackage{colordvi}
\usepackage{xspace}
\usepackage{axodraw}
\usepackage{txfonts}

\usepackage{hyperref}

\begin{document}  

\newcommand{\nc}{\newcommand}
\newcommand{\delete}[1]{}
\nc{\dfootnote}[1]{{}}          
\nc{\ffootnote}[1]{\dfootnote{#1}}
\nc{\mfootnote}[1]{\footnote{#1}} 
\nc{\todo}[1]{\tred{To do:} #1} \delete{
\nc{\mlabel}[1]{\label{#1}}  
\nc{\mcite}[1]{\cite{#1}}  
\nc{\mref}[1]{\ref{#1}}  
}

\nc{\mlabel}[1]{\label{#1}}  
\nc{\mcite}[1]{\cite{#1}}  
\nc{\mref}[1]{\ref{#1}}  
\nc{\mbibitem}[1]{\bibitem{#1}} 

\delete{
\nc{\mlabel}[1]{\label{#1}  
{\hfill \hspace{1cm}{\bf{{\ }\hfill(#1)}}}}
\nc{\mcite}[1]{\cite{#1}{{\bf{{\ }(#1)}}}}  
\nc{\mref}[1]{\ref{#1}{{\bf{{\ }(#1)}}}}  
\nc{\mbibitem}[1]{\bibitem[\bf #1]{#1}} 
}
\newtheorem{theorem}{Theorem}[section]
\newtheorem{prop}[theorem]{Proposition}
\newtheorem{defn}[theorem]{Definition}
\newtheorem{lemma}[theorem]{Lemma}
\newtheorem{coro}[theorem]{Corollary}
\newtheorem{prop-def}{Proposition-Definition}[section]
\newtheorem{claim}{Claim}[section]
\newtheorem{remark}[theorem]{Remark}
\newtheorem{propprop}{Proposed Proposition}[section]
\newtheorem{conjecture}{Conjecture}
\newtheorem{exam}[theorem]{Example}
\newtheorem{assumption}[theorem]{Assumption}
\newtheorem{condition}[theorem]{Condition}

\renewcommand{\labelenumi}{{\rm(\roman{enumi})}}
\renewcommand{\theenumi}{\roman{enumi}}

\nc{\tred}[1]{\textcolor{red}{#1}}
\nc{\tblue}[1]{\textcolor{blue}{#1}}
\nc{\tgreen}[1]{\textcolor{green}{#1}}
\nc{\tpurple}[1]{\textcolor{purple}{#1}}
\nc{\btred}[1]{\textcolor{red}{\bf #1}}
\nc{\btblue}[1]{\textcolor{blue}{\bf #1}}
\nc{\btgreen}[1]{\textcolor{green}{\bf #1}}
\nc{\btpurple}[1]{\textcolor{purple}{\bf #1}}

\nc{\li}[1]{\textcolor{red}{Li:#1}}
\nc{\cm}[1]{\textcolor{blue}{Chengming: #1}}
\nc{\xiang}[1]{\textcolor{green}{Xiang: #1}}

\nc{\adec}{\check{;}} \nc{\aop}{\alpha}
\nc{\dftimes}{\widetilde{\otimes}} \nc{\dfl}{\succ} \nc{\dfr}{\prec}
\nc{\dfc}{\circ} \nc{\dfb}{\bullet} \nc{\dft}{\star}
\nc{\dfcf}{{\mathbf k}} \nc{\apr}{\ast} \nc{\spr}{\cdot}
\nc{\twopr}{\circ} \nc{\sempr}{\ast} \nc{\bwt}{{mass}\xspace}
\nc{\bwts}{{masses}\xspace} \nc{\bop}{{extention}\xspace}
\nc{\ewt}{{mass}\xspace} \nc{\ewts}{{masses}\xspace}
\nc{\tto}{{extended}\xspace} \nc{\Tto}{{Extended}\xspace}
\nc{\tte}{{extended}\xspace} \nc{\gyb}{{generalized}\xspace}
\nc{\Gyb}{{Generalized}\xspace} \nc{\ECYBE}{{ECYBE}\xspace}
\nc{\GAYBE}{{GAYBE}\xspace}
\nc{\triple}{{triple Lie algebra}\xspace}
\nc{\triples}{{triple Lie algebras}\xspace}
\nc{\Triple}{{Triple Lie algebra}\xspace}
\nc{\Triples}{{Triple Lie algebras}\xspace}
\nc{\disp}[1]{\displaystyle{#1}}
\nc{\bin}[2]{ (_{\stackrel{\scs{#1}}{\scs{#2}}})}  
\nc{\binc}[2]{ \left (\!\! \begin{array}{c} \scs{#1}\\
    \scs{#2} \end{array}\!\! \right )}  
\nc{\bincc}[2]{  \left ( {\scs{#1} \atop
    \vspace{-.5cm}\scs{#2}} \right )}  
\nc{\sarray}[2]{\begin{array}{c}#1 \vspace{.1cm}\\ \hline
    \vspace{-.35cm} \\ #2 \end{array}}
\nc{\bs}{\bar{S}} \nc{\dcup}{\stackrel{\bullet}{\cup}}
\nc{\dbigcup}{\stackrel{\bullet}{\bigcup}} \nc{\etree}{\big |}
\nc{\la}{\longrightarrow} \nc{\fe}{\'{e}} \nc{\rar}{\rightarrow}
\nc{\dar}{\downarrow} \nc{\dap}[1]{\downarrow
\rlap{$\scriptstyle{#1}$}} \nc{\uap}[1]{\uparrow
\rlap{$\scriptstyle{#1}$}} \nc{\defeq}{\stackrel{\rm def}{=}}
\nc{\dis}[1]{\displaystyle{#1}} \nc{\dotcup}{\,
\displaystyle{\bigcup^\bullet}\ } \nc{\sdotcup}{\tiny{
\displaystyle{\bigcup^\bullet}\ }} \nc{\hcm}{\ \hat{,}\ }
\nc{\hcirc}{\hat{\circ}} \nc{\hts}{\hat{\shpr}}
\nc{\lts}{\stackrel{\leftarrow}{\shpr}}
\nc{\rts}{\stackrel{\rightarrow}{\shpr}} \nc{\lleft}{[}
\nc{\lright}{]} \nc{\uni}[1]{\tilde{#1}} \nc{\wor}[1]{\check{#1}}
\nc{\free}[1]{\bar{#1}} \nc{\den}[1]{\check{#1}} \nc{\lrpa}{\wr}
\nc{\curlyl}{\left \{ \begin{array}{c} {} \\ {} \end{array}
    \right .  \!\!\!\!\!\!\!}
\nc{\curlyr}{ \!\!\!\!\!\!\!
    \left . \begin{array}{c} {} \\ {} \end{array}
    \right \} }
\nc{\leaf}{\ell}       
\nc{\longmid}{\left | \begin{array}{c} {} \\ {} \end{array}
    \right . \!\!\!\!\!\!\!}
\nc{\ot}{\otimes} \nc{\sot}{{\scriptstyle{\ot}}}
\nc{\otm}{\overline{\ot}} \nc{\ora}[1]{\stackrel{#1}{\rar}}
\nc{\ola}[1]{\stackrel{#1}{\la}}
\nc{\pltree}{\calt^\pl} \nc{\epltree}{\calt^{\pl,\NC}}
\nc{\rbpltree}{\calt^r} \nc{\scs}[1]{\scriptstyle{#1}}
\nc{\mrm}[1]{{\rm #1}}
\nc{\dirlim}{\displaystyle{\lim_{\longrightarrow}}\,}
\nc{\invlim}{\displaystyle{\lim_{\longleftarrow}}\,}
\nc{\mvp}{\vspace{0.5cm}} \nc{\svp}{\vspace{2cm}}
\nc{\vp}{\vspace{8cm}} \nc{\proofbegin}{\noindent{\bf Proof: }}
\nc{\proofend}{$\blacksquare$ \vspace{0.5cm}}
\nc{\freerbpl}{{F^{\mathrm RBPL}}}
\nc{\sha}{{\mbox{\cyr X}}}  
\nc{\ncsha}{{\mbox{\cyr X}^{\mathrm NC}}} \nc{\ncshao}{{\mbox{\cyr
X}^{\mathrm NC,\,0}}}
\nc{\shpr}{\diamond}    
\nc{\shprm}{\overline{\diamond}}    
\nc{\shpro}{\diamond^0}    
\nc{\shprr}{\diamond^r}     
\nc{\shpra}{\overline{\diamond}^r} \nc{\shpru}{\check{\diamond}}
\nc{\catpr}{\diamond_l} \nc{\rcatpr}{\diamond_r}
\nc{\lapr}{\diamond_a} \nc{\sqcupm}{\ot} \nc{\lepr}{\diamond_e}
\nc{\vep}{\varepsilon} \nc{\labs}{\mid\!} \nc{\rabs}{\!\mid}
\nc{\hsha}{\widehat{\sha}} \nc{\lsha}{\stackrel{\leftarrow}{\sha}}
\nc{\rsha}{\stackrel{\rightarrow}{\sha}} \nc{\lc}{\lfloor}
\nc{\rc}{\rfloor} \nc{\tpr}{\sqcup} \nc{\nctpr}{\vee}
\nc{\plpr}{\star} \nc{\rbplpr}{\bar{\plpr}} \nc{\sqmon}[1]{\langle
#1\rangle} \nc{\forest}{\calf} \nc{\ass}[1]{\alpha({#1})}
\nc{\altx}{\Lambda_X} \nc{\vecT}{\vec{T}} \nc{\onetree}{\bullet}
\nc{\Ao}{\check{A}} \nc{\seta}{\underline{\Ao}}
\nc{\deltaa}{\overline{\delta}} \nc{\trho}{\tilde{\rho}}

\nc{\mmbox}[1]{\mbox{\ #1\ }} \nc{\ann}{\mrm{ann}}
\nc{\Aut}{\mrm{Aut}} \nc{\can}{\mrm{can}} \nc{\twoalg}{{two-sided
algebra}\xspace} \nc{\colim}{\mrm{colim}} \nc{\Cont}{\mrm{Cont}}
\nc{\rchar}{\mrm{char}} \nc{\cok}{\mrm{coker}} \nc{\dtf}{{R-{\rm
tf}}} \nc{\dtor}{{R-{\rm tor}}}
\renewcommand{\det}{\mrm{det}}
\nc{\depth}{{\mrm d}} \nc{\Div}{{\mrm Div}} \nc{\End}{\mrm{End}}
\nc{\Ext}{\mrm{Ext}} \nc{\Fil}{\mrm{Fil}} \nc{\Frob}{\mrm{Frob}}
\nc{\Gal}{\mrm{Gal}} \nc{\GL}{\mrm{GL}} \nc{\Hom}{\mrm{Hom}}
\nc{\hsr}{\mrm{H}} \nc{\hpol}{\mrm{HP}} \nc{\id}{\mrm{id}}
\nc{\im}{\mrm{im}} \nc{\incl}{\mrm{incl}} \nc{\length}{\mrm{length}}
\nc{\LR}{\mrm{LR}} \nc{\mchar}{\rm char} \nc{\NC}{\mrm{NC}}
\nc{\mpart}{\mrm{part}} \nc{\pl}{\mrm{PL}} \nc{\ql}{{\QQ_\ell}}
\nc{\qp}{{\QQ_p}} \nc{\rank}{\mrm{rank}} \nc{\rba}{\rm{RBA }}
\nc{\rbas}{\rm{RBAs }} \nc{\rbpl}{\mrm{RBPL}} \nc{\rbw}{\rm{RBW }}
\nc{\rbws}{\rm{RBWs }} \nc{\rcot}{\mrm{cot}}
\nc{\rest}{\rm{controlled}\xspace} \nc{\rdef}{\mrm{def}}
\nc{\rdiv}{{\rm div}} \nc{\rtf}{{\rm tf}} \nc{\rtor}{{\rm tor}}
\nc{\res}{\mrm{res}} \nc{\SL}{\mrm{SL}} \nc{\Spec}{\mrm{Spec}}
\nc{\tor}{\mrm{tor}} \nc{\Tr}{\mrm{Tr}} \nc{\mtr}{\mrm{sk}}

\nc{\ab}{\mathbf{Ab}} \nc{\Alg}{\mathbf{Alg}}
\nc{\Algo}{\mathbf{Alg}^0} \nc{\Bax}{\mathbf{Bax}}
\nc{\Baxo}{\mathbf{Bax}^0} \nc{\RB}{\mathbf{RB}}
\nc{\RBo}{\mathbf{RB}^0} \nc{\BRB}{\mathbf{RB}}
\nc{\Dend}{\mathbf{DD}} \nc{\bfk}{{\bf k}} \nc{\bfone}{{\bf 1}}
\nc{\base}[1]{{a_{#1}}} \nc{\detail}{\marginpar{\bf More detail}
    \noindent{\bf Need more detail!}
    \svp}
\nc{\Diff}{\mathbf{Diff}} \nc{\gap}{\marginpar{\bf
Incomplete}\noindent{\bf Incomplete!!}
    \svp}
\nc{\FMod}{\mathbf{FMod}} \nc{\mset}{\mathbf{MSet}}
\nc{\rb}{\mathrm{RB}} \nc{\Int}{\mathbf{Int}}
\nc{\Mon}{\mathbf{Mon}}
\nc{\remarks}{\noindent{\bf Remarks: }}
\nc{\OS}{\mathbf{OS}} 
\nc{\Rep}{\mathbf{Rep}} \nc{\Rings}{\mathbf{Rings}}
\nc{\Sets}{\mathbf{Sets}} \nc{\DT}{\mathbf{DT}}

\nc{\BA}{{\mathbb A}} \nc{\CC}{{\mathbb C}} \nc{\DD}{{\mathbb D}}
\nc{\EE}{{\mathbb E}} \nc{\FF}{{\mathbb F}} \nc{\GG}{{\mathbb G}}
\nc{\HH}{{\mathbb H}} \nc{\LL}{{\mathbb L}} \nc{\NN}{{\mathbb N}}
\nc{\QQ}{{\mathbb Q}} \nc{\RR}{{\mathbb R}} \nc{\TT}{{\mathbb T}}
\nc{\VV}{{\mathbb V}} \nc{\ZZ}{{\mathbb Z}}


\nc{\calao}{{\mathcal A}} \nc{\cala}{{\mathcal A}}
\nc{\calc}{{\mathcal C}} \nc{\cald}{{\mathcal D}}
\nc{\cale}{{\mathcal E}} \nc{\calf}{{\mathcal F}}
\nc{\calfr}{{{\mathcal F}^{\,r}}} \nc{\calfo}{{\mathcal F}^0}
\nc{\calfro}{{\mathcal F}^{\,r,0}} \nc{\oF}{\overline{F}}
\nc{\calg}{{\mathcal G}} \nc{\calh}{{\mathcal H}}
\nc{\cali}{{\mathcal I}} \nc{\calj}{{\mathcal J}}
\nc{\call}{{\mathcal L}} \nc{\calm}{{\mathcal M}}
\nc{\caln}{{\mathcal N}} \nc{\calo}{{\mathcal O}}
\nc{\calp}{{\mathcal P}} \nc{\calr}{{\mathcal R}}
\nc{\calt}{{\mathcal T}} \nc{\caltr}{{\mathcal T}^{\,r}}
\nc{\calu}{{\mathcal U}} \nc{\calv}{{\mathcal V}}
\nc{\calw}{{\mathcal W}} \nc{\calx}{{\mathcal X}}
\nc{\CA}{\mathcal{A}}

\nc{\fraka}{{\mathfrak a}} \nc{\frakB}{{\mathfrak B}}
\nc{\frakb}{{\mathfrak b}} \nc{\frakd}{{\mathfrak d}}
\nc{\oD}{\overline{D}} \nc{\frakF}{{\mathfrak F}}
\nc{\frakg}{{\mathfrak g}} \nc{\frakk}{{\mathfrak k}}
\nc{\frakm}{{\mathfrak m}} \nc{\frakn}{{\mathfrak n}}
\nc{\frakM}{{\mathfrak M}} \nc{\frakMo}{{\mathfrak M}^0}
\nc{\frakp}{{\mathfrak p}} \nc{\frakS}{{\mathfrak S}}
\nc{\frakSo}{{\mathfrak S}^0} \nc{\fraks}{{\mathfrak s}}
\nc{\os}{\overline{\fraks}} \nc{\frakT}{{\mathfrak T}}
\nc{\oT}{\overline{T}}
\nc{\frakX}{{\mathfrak X}} \nc{\frakXo}{{\mathfrak X}^0}
\nc{\frakx}{{\mathbf x}}
\nc{\frakTx}{\frakT}      
\nc{\frakTa}{\frakT^a}        
\nc{\frakTxo}{\frakTx^0}   
\nc{\caltao}{\calt^{a,0}}   
\nc{\ox}{\overline{\frakx}} \nc{\fraky}{{\mathfrak y}}
\nc{\frakz}{{\mathfrak z}} \nc{\oX}{\overline{X}}

\font\cyr=wncyr10

\nc{\redtext}[1]{\textcolor{red}{#1}}


\title[Nonabelian generalized Lax pairs]{Nonabelian generalized Lax pairs, the classical Yang-Baxter equation and
PostLie algebras}

\author{Xiang Ni}
\address{Chern Institute of Mathematics \& LPMC, Nankai
University, Tianjin 300071, P.R.
China}\email{xiangn$_-$math@yahoo.cn}

\author{Chengming Bai}
\address{Chern Institute of Mathematics \& LPMC, Nankai University, Tianjin 300071, P.R. China}
         \email{baicm@nankai.edu.cn}

\author{Li Guo}
\address{Department of Mathematics and Computer Science,
         Rutgers University,
         Newark, NJ 07102}
\email{liguo@newark.rutgers.edu}



\begin{abstract}
We generalize the classical study of (generalized) Lax pairs and the related $\calo$-operators and the (modified) classical Yang-Baxter equation by introducing the concepts of nonabelian generalized Lax pairs, \tto $\calo$-operators and the \tte classical Yang-Baxter equation. We study in this context the nonabelian generalized $r$-matrix ansatz and the related double Lie algebra structures. Relationship between \tto $\calo$-operators and the \tte classical Yang-Baxter equation is established, especially for self-dual Lie algebras. This relationship allows us to obtain explicit description of the Manin triples for a new class of Lie bialgebras. Furthermore, we show that a natural structure of PostLie algebra is behind $\calo$-operators and fits in a setup of \triple that produces self-dual nonabelian generalized Lax pairs.
\end{abstract}

\subjclass[2000]{}

\keywords{Lax pair, Lie algebra, Lie bialgebra, $\calo$-operator, classical Yang-Baxter equation, PostLie algebra}

\maketitle

\tableofcontents

\setcounter{section}{0} {\ } \vspace{-1cm}

\section{Introduction}

This paper is devoted to a systematic study of the integrable
Hamiltonian systems and the related (generalized) classical
Yang-Baxter equation (CYBE) in a broad context that generalizes or extends the
studies of Bordemann~\cite{Bo}, Hodge and Yakimov~\mcite{HY},
Kosmann-Schwarzbach and Magri~\mcite{KoM}, and
Semonov-Tian-Shansky ~\cite{Se}.
\smallskip

Since their introduction by Lax in 1968, Lax pairs have become
important in giving conservation laws in an integrable system. In
connection with $r$-matrices satisfying the classical Yang-Baxter
equation (CYBE), Poisson commuting conservation laws could be
constructed. Main contributors in this direction include
Adler~\mcite{Ad}, Babelon and Viallet~\mcite{BV1,BV2}, Belavin and Drinfeld~\mcite{BeD,BD,Dr}, Faddeev~\mcite{Fa}, Kostant~\mcite{Kos}, Reyman and Semonov-Tian-Shansky~\mcite{RS2,Se}, Sklyanin~\mcite{Sk1,Sk2} and Symes~\mcite{Sy1, Sy2}.

In~\mcite{Bo} Bordemann introduced the notions of generalized Lax pairs
and generalized $r$-matrix ansatz. He achieved this through replacing the well-known Lax
equation~\mcite{La}
$$
\frac{dL}{dt}=[L,M]$$ by
\begin{equation}
\frac{dL}{dt} = -\rho(M)L, \mlabel{eq:glax}
\end{equation}
where $\rho$ is any representation of a Lie algebra $\frakg$ in a
representation space $V$, $M$ is a $\frakg$-valued function on the
phase space and $L$ is a $V$-valued function on the phase space,
reducing to the Lax equation when $V$ is taken to be $\frakg$ and
$\rho$ is taken to be the adjoint representation. In this
generality, the correct framework to extend the classical
$r$-matrices is through their operator forms, later called
$\calo$-operators by Kupershmidt~\mcite{Ku}.

The classical Yang-Baxter equation, through its operator
form and tensor form, plays a central role in relating several areas
in mathematics.
For the most part, the
operator form is more convenient in application to integrable
systems. For example, the modified classical Yang-Baxter equation is
solely defined in the operator form. Nevertheless, the tensor form of
the CYBE is the classical limit of the quantum Yang-Baxter equation, and its
solutions give rise to important concepts such as (coboundary)
Lie bialgebras. Thus it is desirable to work with both forms of the
CYBE.

In the present paper, we keep both forms of the CYBE in mind while
we generalize the previous works. For the operator form, we further
generalize the work of Bordemann and Kupershmidt by introducing the concepts of
an {\bf $\calo$-operator of weight $\lambda$} (for a constant $\lambda$) and an {\bf \tto $\calo$-operator}. This is motivated by our attempt to
extend generalized Lax pairs of Bordemann to {\bf nonabelian
generalized Lax pairs}, by still considering Eq.~(\mref{eq:glax})
but replacing the representation space $V$ by any Lie algebra $\fraka$ and the representation $\rho$ by any Lie
algebra homomorphism from $\frakg$ to ${\rm Der} (\fraka)$
consisting of derivations of $\fraka$. The setting of Bordemann is
recovered when $\fraka$ is taken to be an abelian Lie algebra. We
extend the generalized $r$-matrix ansatz of Bordemann to the
nonabelian context and show that \tto $\calo$-operators ensure the
consistency of a Lie structure on $\fraka^*$ defined by the
$r$-matrices. For the tensor form, we introduce the concept of the {\bf \tto
classical Yang-Baxter equation} and establish their relationship with
\tto $\calo$-operators as in the case of (the tensor form and
operator form) of the CYBE. We further extend the well-known work of
Drinfeld on quasitriangular Lie bialgebras from the CYBE to what we
dubbed {\bf type II quasitriangular Lie bialgebras} from a case of the
\tto classical Yang-Baxter equation, called the {\bf type II CYBE}. The
corresponding Drinfeld's doubles and Manin triples are studied
carefully as in the classical case by Hodge and Yakimov~\mcite{HY},
for their importance in the classification of the Poisson
homogeneous spaces and symplectic leaves of the corresponding
Poisson-Lie groups ~\cite{Dr1, HY, Se1, Y}.

As it turns out, an $\calo$-operator of weight $\lambda$ is related
to the concept of a PostLie algebra that has recently arisen from the
quite different context of operads~\cite{Va}. More precisely, an
$\calo$-operator, paired with a $\frakg$-Lie algebra, gives a
PostLie algebra. In particular, Baxter Lie algebras and quasitriangular Lie bialgebras give rise to PostLie algebras. Furthermore the well-known relation~\cite{C}
between pre-Lie algebras and dendriform dialgebras, in connection
with the classical relation between associative algebras and Lie
algebras, can be extended to that between PostLie algebras and
dendriform trialgebras.
Quite unexpectedly, this ``digression" of $\calo$-operators to
PostLie algebra is tired up with our primary application of
$\calo$-operators in studying nonabelian generalized Lax pairs: We
introduce the concept of a {\bf \triple} to construct self-dual
nonabelian generalized Lax pairs and show that a natural example of
a \triple is provided by the PostLie algebra from a Rota-Baxter
operator action on a complex simple Lie algebra.
\medskip

We next give a summary of this paper.

We begin our study by introducing the concept of { a nonabelian
generalized Lax pair}. We write down a { ``nonabelian generalized
$r$-matrix ansatz"} to produce Poisson commuting conservation
laws. The idea is to use the Lie-Poisson structure on the
representation space (equipped with a Lie bracket) to twist the
``generalized $r$-matrix ansatz" of Bordemann ~\cite{Bo}. In
geometry, this construction might be understood as ``twisting" a
Hamiltonian system (Poisson bracket) by the Hamiltonian system
(Lie-Poisson bracket) on the dual space of a Lie algebra. The
notions { $\calo$-operator of weight $\lambda$} and {
\tto $\calo$-operator of weight $\lambda$ with extension $\beta$
of \ewt $(\nu,\kappa,\mu)$} (for constants $(\nu,\kappa,\mu)$) appear naturally when we investigate
sufficient conditions for the double Lie algebra structures needed for
the existence of the ansatz.

To generalize the well-known relationship between the operator form
and tensor form of the CYBE, we introduce in
Section~\mref{sec:liebialgebra} the concept of an \tte CYBE and
relate it to \tto $\calo$-operators. Applications
to Lie bialgebras are given. In particular, we study in detail the
structure of the Manin triple of a type II quasitriangular Lie
bialgebra.

In Section~\mref{se:self}, we study the case of self-dual Lie
algebras. The ideal is to use a nondegenerate symmetric and
invariant bilinear form of a self-dual Lie algebra to identify the adjoint
representation and coadjoint representation ~\cite{Se}. Some new
aspects on Lie bialgebras are given along this approach, for
example, new examples of (type II) factorizable quasitriangular Lie
bialgebras are provided.

We show in Section~\mref{sec:postlie} that there naturally exists an
algebraic structure behind an $\calo$-operator of weight $\lambda$,
namely, the PostLie algebra discovered in a study of
operads ~\cite{Va}. We also reveal a relation between PostLie
algebras and { dendriform trialgebras} of Loday and Ronco~\cite{LR}
by a commutative diagram.

In Section~\mref{sec:triple}, we provide a framework of \triples to
construct a class of nonabelian generalized Lax pairs for which the
corresponding $r$-matrix ansatz can be written down explicitly ~\cite{CP}. We
show that PostLie algebras provide natural instances of such
\triples.

Finally in an appendix, we give a geometric explanation of \tto $\calo$-operators.

\smallskip

\noindent {\bf Conventions: }In this paper, the base field is taken
to be $\RR$ of real numbers unless otherwise specified. This is the field from which we take all the
constants and over which we take all the associative and Lie
algebras, vector spaces, linear maps and tensor products, etc. All
Lie algebras, vector spaces and manifolds are assumed to be
finite-dimensional, although many results still hold in
infinite-dimensional case.

\medskip

\noindent {\bf Acknowledgements: } The second author was supported
in part by the National Natural Science Foundation of China (1062
1101), NKBRPC (2006CB805905) and SRFDP (200800550015). The third
author was supported by NSF grant DMS 0505445 and thanks the Chern
Institute of Mathematics at Nankai University for hospitality.

\section{Nonabelian generalized Lax pairs and \tto $\calo$-operators}
\label{sec:lax}
We begin with generalizing the generalized Lax pairs of Bordemann~\mcite{Bo} further to nonabelian generalized Lax pairs. By studying generalized $r$-matrix ansatz and double Lie algebra structures in this context, we are motivated to introducing the concept of an \tto $\calo$-operator, generalizing the work of Bordemann and Kupershmidt~\mcite{Ku} in several directions. The case of adjoint representations is studied separately.

\subsection{Nonabelian generalized Lax pairs}
\mlabel{ss:lax}
We first introduce a suitable replacement of Lie algebra representations in order to extend generalized Lax pairs to the nonabelian context.
\begin{defn}
 {\rm
\begin{enumerate}
\item Let $(\frakg,[\,,\,]_\frakg)$, or simply $\frakg$, denote a Lie algebra $\frakg$ with Lie bracket $[\,,\,]_\frakg$.
\item
For a Lie algebra $\frakb$, let ${\rm Der}_{\RR}\frakb$ denote the Lie algebra of derivations of $\frakb$.
\item
Let $\fraka$ be a Lie algebra. An {\bf $\fraka$-Lie algebra} is a triple $(\frakb,[\,,\,]_\frakb,\pi)$ consisting of a Lie algebra $(\frakb,[\,,\,]_\frakb)$ and a Lie algebra homomorphism $\pi:\fraka \to {\rm Der}_{\RR}\frakb$. To simplify the notation, we also let $(\frakb,\pi)$ or simply $\frakb$ denote $(\frakb,[\,,\,]_\frakb,\pi)$. \item
 Let $\fraka$ be a Lie algebra and let $(\frakg,\pi)$ be an $\fraka$-Lie algebra. Let $a\cdot b$ denote $\pi(a)b$ for $a\in \fraka$ and $b\in \frakg$.
\end{enumerate}
}
\label{de:semidi}
\end{defn}

According
 to~\cite{Kna}, if $(\mathfrak b, \pi)$ is an $\frak{a}$-Lie
 algebra, then there
exists a unique Lie algebra structure on the direct sum
$\mathfrak{g}=\mathfrak{a}\oplus\mathfrak{b}$ of the underlying
vector spaces $\mathfrak{a}$ and $\mathfrak b$ such that
$\mathfrak{a}$ and $\mathfrak b$ are subalgebras and $[x,y]=\pi(x)y$
for $x\in\mathfrak{a}$ and $y\in\mathfrak{b}$. Further,
$\mathfrak{a}$ is a subalgebra and
$\mathfrak{b}$ is an ideal of the Lie algebra $\frakg$.

Let $(P,w)$ be a Poisson manifold with the Poisson bivector
$w\in\bigwedge^2 T(M)$ which induces a Poisson bracket $\{,\}$ on
$C^{\infty}(P)$. A smooth function $f$ on $P$, which is called an
{\bf observable}, determines a {\bf Hamiltonian vector field} $X_f$
by
$$X_f\,g\equiv\{f,g\},\;\; g\in C^{\infty}(P).$$
 If a Hamiltonian
system is modeled by a Poisson manifold $(P,w)$ (the phase space of
the system) and a Hamiltonian $\mathcal{H}\in C^{\infty}(P)$, its
time-evolution is given by the following integral curves of the
Hamiltonian vector field $X_{\mathcal{H}}$ on $P$ corresponding to
$\mathcal{H}$:
$$ X_{\mathcal{H}}(f)\equiv\{\mathcal{H},f\},\;\;\forall f\in
C^{\infty}(P).$$
It follows that
$$\frac{df}{dt}=\{\mathcal{H},f\},$$
in the sense that $(d/dt)(f(m(t)))=\{\calh,f\}(m(t))$ for an integral curve $m(t)$ of $X_\calh$.
As usual, an observable
$f$ is called a {\bf conservation law} or {\bf conserved} if $\{\mathcal{H},f\}=0$. Two conservation laws
$f_1, f_2$ on a Poisson manifold are {\bf in involution} or {\bf Poisson commuting} if
$\{f_1,f_2\}=0$.
Moreover, a Hamiltonian system $(P,w,\mathcal{H})$ is called {\bf completely integrable} if it has the maximum number of conserved observables in involution~\mcite{CP}.

An important procedure to obtain Poisson commuting observables
and completely integrable Hamiltonian systems
is through the concept of {\bf Lax pairs}~\mcite{La} which was generalized by Bordemann~\mcite{Bo} to {\bf generalized Lax pairs}.
We now generalize this further to the following concept.
\begin{defn}
{\rm
\begin{enumerate}
\item
A {\bf nonabelian generalized Lax pair} for a Hamiltonian
system $(P,w,\mathcal{H})$ is a quintuple
$(\mathfrak{g},\rho,\mathfrak{a},L,M)$ satisfying the following
conditions:

\begin{enumerate}
\item
$\frak{g}$ is a (finite-dimensional) Lie algebra;
\item
$(\fraka,\rho)$ is a (finite-dimensional) $\frakg$-Lie algebra with
the Lie algebra homomorphism $\rho:\frakg \to
\text{Der}_\RR(\fraka)$;
\item
$L:P\rightarrow\frak{a}$ is a smooth map,
\item
 $M:P\rightarrow\frak{g}$ is a smooth map such that
\begin{equation}
dL(p)X_\calh(p)=-\rho(M(p))L(p),\quad\forall p\in P.\label{eq:difeq}
\end{equation}
\end{enumerate}
\item
A nonabelian generalized Lax pair
$(\mathfrak{g},\rho,\mathfrak{a},L,M)$ is said to be {\bf self-dual}
if $\frak{a}$ is equipped with a nondegenerate symmetric
bilinear form $\frak{B}:\frak{a}\otimes\frak{a}\rightarrow\mathbb R$
such that
\begin{equation}
\frak{B}([x,y]_{\frak{a}},z)=\frak{B}(x,[y,z]_{\frak{a}}),\quad
\forall x,y,z\in\frak{a},\label{eq:biform}
\end{equation}
\begin{equation}
\frak{B}(\rho(\xi)x,y)+\frak{B}(x,\rho(\xi)y)=0,\quad \forall
\xi\in\frak{g},x,y\in\frak{a}.\label{eq:rhobiform}
\end{equation}
\end{enumerate}
} \label{de:nonlax}
\end{defn}

Note that a bilinear form on a Lie algebra satisfying
Eq.~(\ref{eq:biform}) is called {\bf invariant} and a Lie algebra
endowed with a nondegenerate symmetric invariant bilinear form is
called a {\bf self-dual Lie algebra} ~\cite{FS}.

 By the chain rule,
Eq.~(\ref{eq:difeq}) is equivalent to
\begin{equation}
\frac{dL}{dt}=-\rho(M)L.\label{eq:inteq}
\end{equation}

\begin{remark}
{\rm
\begin{enumerate}
\item
When the Lie bracket on $\frak a$ happens to be trivial, the $\frakg$-Lie algebra $(\fraka,\rho)$ becomes a representation of $\frakg$ and the
nonabelian generalized Lax pair becomes the {\bf generalized Lax pair} in
the sense of Bordemann ~\cite{Bo}.
\item
For $\frak{a}=\frak{g}$ and $\rho={\rm ad}$, Eq.~(\ref{eq:inteq}) is
the usual Lax equation. Moreover, the {Lax pair} can be
realized as a nonabelian generalized Lax pair in two different ways, by either taking $\rho$ to be ${\rm ad}$ and $\frak{a}$ to be the Lie algebra $\frak{g}$, or taking $\rho$ to be ${\rm ad}$ and
$\frak{a}$ to be the underlying vector space of $\frak{g}$ equipped with the trivial Lie bracket.
\end{enumerate} }
\end{remark}

Let $G$ be a connected Lie group whose Lie algebra is $\frak{g}$
such that $\rho$ exponentiates to a representation of $G$ in $V$
which we shall also call $\rho$. We first show that, as in the case of Lax pairs and generalized Lax pairs~\cite{Bo}, nonabelian
generalized Lax pairs also give conservation laws.

\begin{prop}
Let $(\mathfrak{g},\rho,\mathfrak{a},L,M)$ be a nonabelian
generalized Lax pair for a Hamiltonian system $(P,w,\mathcal{H})$.
If $f:\mathfrak{a}\rightarrow\mathbb R$ is a $G$-invariant smooth
function, i.e., $f(\rho(g)x)=f(x)$ for all $g\in G$ and
$x\in\mathfrak{a}$, then $f\circ L$ is a conservation law, i.e.,
$$
\frac{d(f\circ L)}{dt}=\{f\circ L,\mathcal{H}\}=0.$$
\end{prop}\label{pp:conlaw}

\begin{proof}
Since $G$-invariant functions are always constant on each $G$-orbit,
 we have
 $$df(x)\rho(\xi)x=0,\quad \forall \xi\in \frak{g}, x\in\mathfrak{a}.$$
So $$\frac{d}{dt}(f\circ L)=df(L)\frac{d}{dt}L=-df(L)\rho(M)L=0.$$
\end{proof}

Let $\{e_i\}_{1\leq i\leq{\rm dim}\frak{a}}$ be a basis of
$\frak{a}$ and $\{T_A\}_{1\leq A\leq{\rm dim}\frak{g}}$ be a basis
of $\frak{g}$. For any $x=\sum\limits_ix^ie_i\in \frak{a}$ and
$\xi=\sum\limits_A\xi^AT_A\in\frak{g}$, we set
$(\rho(\xi)x)^i=\sum\limits_{A,j}\xi^Ax^j\rho_{Aj}^i$\,. On the other
hand, suppose that the Lie algebra structure on $\frak{a}$ is
given by $[e_i,e_j]_{\frak{a}}=\sum\limits_kc_{ij}^ke_k$. The
Poisson bracket $\{f\circ L,h\circ L\}$ for arbitrary smooth
functions $f,h:\frak{a}\rightarrow\mathbb R$ is
\begin{equation}
\{f\circ L,h\circ L\}=\sum_{i,j}\frac{\partial f}{\partial x^i}\circ
L\frac{\partial h}{\partial x^j}\circ L\{L^i,L^j\}.\label{eq:poibra}
\end{equation}
Now consider smooth maps which we shall call {\bf classical
$r$-matrices} (following~\mcite{Bo})
$$
r_{+},r_{-}:\frak{a}\times P\rightarrow \frak{a}\otimes\frak{g}
$$
and make the following {\bf nonabelian generalized $r$-matrix
ansatz}:
\begin{equation}
\{L^i,L^j\}(p)=-\sum_{A,k}r_{+}^{iA}(L(p),p)\rho_{Ak}^jL^k(p)+\sum_{A,k}
r_{-}^{jA}(L(p),p)\rho_{Ak}^iL^k(p)-\sum_k\theta_i(p)c_{ik}^jL^k(p),
\label{eq:ansatz}
\end{equation}
where $\theta:P\rightarrow\frak{a}$ is a smooth function and
$\theta_i=x^i\circ\theta:P\rightarrow\mathbb R$, $1\leq i\leq{\rm
dim}\frak{a}$.

When $\theta=0$, the third term on the right hand side vanishes and the ansatz is reduced to Bordemann's {\bf generalized $r$-matrix ansatz}~\cite{Bo}. Generalizing the work of Bordemann, we next show that the nonabelian generalized $r$-matrix ansatz gives Poisson commuting conservation laws.

\begin{prop}
Let $(\mathfrak{g},\rho,\mathfrak{a},L,M)$ be a nonabelian
generalized Lax pair for a Hamiltonian system $(P,w,\mathcal{H})$
allowing for classical $r$-matrices that obey
Eq.~$($\ref{eq:ansatz}$)$. Then for two real-valued $G$-invariant
and ${\rm Ad}_{\frak{a}}$-invariant functions $f$ and $h$ on
$\frak{a}$, we have $\{f\circ L,h\circ L\}=0$.\label{pp:pansatz}
\end{prop}

\begin{proof}
Using Eq.~(\ref{eq:poibra}), we get
\begin{eqnarray*}
\{f\circ L,h\circ L\}&=&-\sum_{i,A,j,k}\frac{\partial f}{\partial
x^i}\circ L r_{+}^{iA}\underline{\frac{\partial h}{\partial
x^j}\circ L\rho_{Ak}^jL^k}+\sum_{i,A,j,k}\underline{\frac{\partial
f}{\partial x^i}\circ L\rho_{Ak}^iL^k}\frac{\partial h}{\partial
x^j}\circ Lr_{-}^{jA}\\
&&-\sum_{i,j,k}\underline{\frac{\partial h}{\partial x^j}\circ
Lc_{ik}^jL^k}\frac{\partial f}{\partial x^i}\circ
L\theta_i.\end{eqnarray*} The underlined terms are zero because of
infinitesimal $G$-invariance and ${\rm Ad}_{\frak{a}}$-invariance of
$f$ and $h$.
\end{proof}

As pointed out by Bordemann in ~~\cite{Bo}, for
$\frak{a}=\frak{g}$ (with the trivial bracket) and $\rho={\rm ad}$,
the classical $r$-matrices take values in $\frak{g}\otimes\frak{g}$,
and the above conclusion becomes the classical fact~~\cite{BV1} that arbitrary trace polynomials of $L$ Poisson commute among themselves.

The Lie bracket conditions on the left hand side of
Eq.~(\mref{eq:ansatz}) impose consistence restrains on the
classical $r$-matrices on the right hand side. In the case of
constant $r$-matrices (namely $L$-independent) that we will consider
below, as observed by Bordemann, the space spanned by the component
functions $L^i$ behaves like a finite-dimensional Lie subalgebra of
the Poisson algebra of functions on the phase space $(P,w)$ since
the right-hand side of Eq.~(\ref{eq:ansatz}) is linear in $L$.
Suppose one wants to collectively investigate all nonabelian
generalized Lax pairs that are defined on a given Hamiltonian system
$(P,w,\mathcal{H})$, that have given $\frak{g}$, $\rho$ and
$\frak{a}$, and that satisfy Eq.~(\mref{eq:ansatz}) with given
classical $r$-matrices $r_{+}$ and $r_{-}$. Then one is led to the
following stronger condition than the above mentioned consistence
restrains imposed on Eq.~(\ref{eq:ansatz}):
\begin{condition}
{\rm The quantities
$$
f^{ij}_k\equiv
-\sum_Ar_{+}^{iA}\rho_{Ak}^j+\sum_Ar_{-}^{jA}\rho_{Ak}^i-\theta_ic_{ik}^j
$$
should be the structure constants of a Lie structure on $\frak{a}^*$.
}
\mlabel{con:lie}
\end{condition}
To obtain an index-free form of Condition~\mref{con:lie}, we first give the following lemma.
\begin{lemma}
Let $\frak{g}$ be a Lie algebra and $(\frak{a},\rho)$ be a
$\frak{g}$-Lie algebra. Let $\frak{B}:\frak{a}\otimes
\frak{a}\rightarrow\mathbb{R}$ be a nondegenerate bilinear form on
$\frak{a}$ which can be identified as an invertible linear map
$\varphi:\frak{a}\rightarrow \frak{a}^*$ through
\begin{equation}
\frak{B}(x,y)=\langle \varphi(x),y\rangle, \quad \forall x,y\in \fraka.
\label{eq:definelinearmap}
\end{equation}
Let $(\fraka^*,\rho_\varphi)$ be the $\frakg$-Lie algebra through
$\varphi$ by transporting the $\frakg$-Lie algebra structure on
$\fraka$. More precisely, define the Lie bracket on $\frak{a}^*$ by
\begin{equation}
[a^*,b^*]_{\frak{a}^*}=\varphi([\varphi^{-1}(a^*),
\varphi^{-1}(b^*)]_{\frak{a}}),\;\;\forall a^*,b^*\in
\frak{a}^*\label{eq:dualpro}
\end{equation}
and define a linear map
\begin{equation}
\rho_{\varphi}:
\frak{g}\rightarrow \End_{\mathbb{R}}(\frak{a}^*),\quad \rho_{\varphi}(\xi)a^*\equiv \varphi
\rho(\xi)\varphi^{-1}(a^*),\;\;\forall a^*\in \frak{a}^*,\xi\in
\frak{g}.\label{eq:dualrepre}
\end{equation}
If $\frak{B}$ satisfies Eq.~$($\ref{eq:rhobiform}$)$, then
$\rho_{\varphi}$ is just the dual representation $\rho^*$ of $\rho$
which is defined by
$$
\langle\rho^*(\xi)a^*,x\rangle=-\langle a^*,\rho(\xi)x\rangle,\quad
\forall \xi\in\frak{g},x\in\frak{a},a^*\in\frak{a}^*.
$$
In this case, $(\frak{a}^*,\rho^*)$ is a $\frak{g}$-Lie algebra with
the Lie bracket defined by Eq.~$($\ref{eq:dualpro}$)$.
 \label{le:dualbi}
\end{lemma}

\begin{proof}
If $\frak{B}$ satisfies
Eq.~(\ref{eq:rhobiform}), then for any
$\xi\in\frak{g},x,y\in\frak{a}$,
$$\langle\varphi(\rho(\xi)x),y\rangle=-\langle\varphi(x),\rho(\xi)y\rangle\Rightarrow
\varphi\rho(\xi)=\rho^*(\xi)\varphi,\quad \forall \xi\in\frak{g}.$$
Hence $\rho_{\varphi}=\rho^*$. So the conclusion holds.
\end{proof}

Assume that $\frak{a}$ is equipped with a nondegenerate symmetric
bilinear form $\frak{B}:\frak{a}\otimes\frak{a}\rightarrow\mathbb R$ for which the nonabelian generalized Lax pair $(\mathfrak{g},\rho,\mathfrak{a},L,M)$ is self-dual. Let $\frak{a}^*$ be equipped with the Lie bracket defined by Eq.~(\ref{eq:dualpro}). By Lemma~\mref{le:dualbi}, $(\fraka^*,[\,,\,]_{\fraka^*},\rho^*)$ is a $\frakg$-Lie algebra. Since $\frak{B}$ is nondegenerate and symmetric, we can choose a basis $\{e_i\}_{1\leq i\leq{\rm
dim}\frak{a}}$ of $\fraka$ such that
$$b_{ij}\equiv\frak{B}(e_i,e_j)=\langle\varphi(e_i),e_j\rangle=0,\;\;{\rm
if}\;\; i\neq
j;\;\;b_{ii}\equiv\frak{B}(e_i,e_i)=\langle\varphi(e_i),e_i\rangle\neq
0.$$ Thus, $\varphi(e_i)=b_{ii}e_i^*$, where $\{e_i^*\}_{1\leq
i\leq{\rm dim}\frak{a}}$ is the dual basis of $\{e_i\}_{1\leq
i\leq{\rm dim}\frak{a}}$. Since
$\frak{B}([e_i,e_j]_{\frak{a}},e_k)+\frak{B}(e_j,[e_i,e_k]_{\frak{a}})=0$,
we have $c_{ij}^kb_{kk}+c_{ik}^jb_{jj}=0$. Therefore,
$$[e_i^*,e_j^*]_{\frak{a}^*}=\varphi[\varphi^{-1}(e_i^*),\varphi^{-1}(e_j^*)]_{\frak{a}}=
\varphi([\frac{e_i}{b_{ii}},\frac{e_j}{b_{jj}}]_{\frak{a}})=
\frac{\sum\limits_{k}c_{ij}^kb_{kk}e_k}{b_{ii}b_{jj}}=\frac{-\sum\limits_{k}c_{ik}^je_k}{b_{ii}}.$$
Now we set $\theta_i\equiv\frac{\lambda}{b_{ii}}$ for
$\lambda\in\mathbb R$. On the other hand, since
$\frak{a}\otimes\frak{g}\simeq{\rm Hom}(\frak{a}^*,\frak{g})$,
$r_{+}$ and $r_{-}$ can be considered as linear maps
$\frak{a}^*\rightarrow \frak{g}:x=x_ie_i^*\rightarrow
r_{\pm}(x)\equiv \sum\limits_{i,A}x_ir_{\pm}^{iA}T_A$. Set
$$\frak{k}\equiv\frak{a}^*, \quad \pi\equiv \rho^*,\quad \xi\cdot x\equiv \pi(\xi)x, \quad x\in \frakk, \xi\in \frakg.$$
Then Condition~\mref{con:lie} can be reformulated
as follows:
\begin{condition}
{\rm {\bf (Double Lie algebra structure) } The product
$$
[x,y]_R\equiv r_{+}(x)\cdot y-r_{-}(y)\cdot
x+\lambda[x,y]_{\frak{k}},\quad \forall
x,y\in\frak{k}.
$$
defines a Lie bracket on $\frakk$.
\mlabel{con:lie2}
}
\end{condition}
Define
\begin{equation}
r\equiv (r_{+}+r_{-})/2,\quad
\beta\equiv (r_{+}-r_{-})/2.\label{eq:ralpha}
\end{equation}
Then $r_{\pm}=r\pm\beta$. Moreover, we have the following result:

\begin{prop}
Condition~\mref{con:lie2} holds if and only if for any
$x,y,z\in\frak{k}$,
\begin{enumerate}
\item
$[x,y]_R=r(x)\cdot y-r(y)\cdot
x+\lambda[x,y]_{\frak{k}}\Leftrightarrow \beta(x)\cdot
y+\beta(y)\cdot x=0,$
\mlabel{it:lie}
\item
$([r(x),r(y)]_{\frak{g}}-r([x,y]_R))\cdot
z+([r(y),r(z)]_{\frak{g}}-r([y,z]_R))\cdot x
+([r(z),r(x)]_{\frak{g}}-r([z,x]_R))\cdot y =0.$ \mlabel{it:cyc}
\end{enumerate}
\label{pp:Liestructure}
\end{prop}

To simplify notations, for an expression $\eta(x,y,z)$ in $x,y$ and $z$, we denote
$$
\eta(x,y,z)+\text{cycl.} = \eta(x,y,z) +\eta(y,z,x)+\eta(z,x,y).
$$

\begin{proof}
Obviously, condition (\mref{it:lie}) is equivalent to the fact that
$[,]_R$ is skew-symmetric. Now we prove that condition
(\mref{it:cyc}) is equivalent to the fact that $[,]_R$ satisfies
Jacobi identity. In fact, for all $x,y,z\in\frak{k}$,
\begin{eqnarray*}
[[x,y]_R,z]_R&=&r([x,y]_R)\cdot z-r(z)\cdot(r(x)\cdot
y)+r(z)\cdot(r(y)\cdot x)-\lambda r(z)\cdot[x,y]_{\frak{k}}+\\
\mbox{}& \mbox{}&\lambda[r(x)\cdot y,z]_{\frak{k}}-
\lambda[r(y)\cdot x,z]_{\frak{k}}+\lambda^2[[x,y]_{\frak{k}},z]_{\frak{k}},\\
{[[z,x]_R,y]}_R &=& r([z,x]_R)\cdot y-r(y)\cdot(r(z)\cdot
x)+r(y)\cdot(r(x)\cdot z)-\lambda r(y)\cdot[z,x]_{\frak{k}}+\\
\mbox{}& \mbox{}&\lambda[r(z)\cdot x,y]_{\frak{k}}-\lambda[r(x)\cdot
z,y]_{\frak{k}}+\lambda^2[[z,x]_{\frak{k}},y]_{\frak{k}},\\
{[[y,z]_R,x]}_R &=&r([y,z]_R)\cdot x-r(x)\cdot(r(y)\cdot
z)+r(x)\cdot(r(z)\cdot y)-\lambda r(x)\cdot[y,z]_{\frak{k}}+\\
\mbox{}& \mbox{}&\lambda[r(y)\cdot z,x]_{\frak{k}}-\lambda[r(z)\cdot
y,x]_{\frak{k}}+\lambda^2[[y,z]_{\frak{k}},x]_{\frak{k}}.
\end{eqnarray*}
Then the conclusion follows from the fact that $(\frak{k},\pi)=(\fraka^*,\rho^*)$ is a
$\frak{g}$-Lie algebra.
\end{proof}

\subsection{\Tto $\calo$-operators and double Lie brackets}
\mlabel{ss:oop}
We will next study the conditions in Proposition~\mref{pp:Liestructure} in order to understand double Lie algebra structures and nonabelian generalized Lax pairs. For this purpose, we introduce the following concepts.
\begin{defn}
{\rm Let $(\frakg,[\,,\,]_\frakg)$ be a Lie algebra and let $(\frakk,[\,,\,]_\frakk,\pi)$ be a $\frakg$-Lie algebra. Let $\nu,\kappa,\mu$ and $\lambda$ be constants (in $\RR$).
\begin{enumerate}
\item A linear map $\beta:\frakk\to \frakg$ is called {\bf antisymmetric (of \bwt $\nu$)} if
$\nu\beta(x)\cdot y +\nu \beta(y)\cdot x=0$ for any $x,y\in \frakk$;
\item A linear map $\beta:\frakk\to \frakg$ is called {\bf $\frakg$-invariant (of \bwt $\kappa$)} if
$\kappa\beta(\xi\cdot x)=\kappa[\xi,\beta(x)]_{\frak{g}}$, for any
$\xi\in\frak{g},x\in\frak{k}$;
\item A linear map $\beta:\frakk\to \frakg$ is called {\bf equivalent (of \bwt $\mu$)} if
$\mu \beta([x,y]_{\frak{k}})\cdot z=\mu[\beta(x)\cdot
y,z]_{\frak{k}}$, for any $x,y,z\in\frak{k}$;
\item Let $\beta:\frakk\to \frakg$ be antisymmetric of \bwt $\nu$, $\frakg$-invariant
of \bwt $\kappa$ and equivalent of \bwt $\mu$. Let $r: \frakk \to
\frakg$ be a linear map. The pair $(r,\beta)$ or simply $r$ is
called an {\bf \tto $\calo$-operator of weight $\lambda$ with
extension $\beta$ of \ewt $(\nu,\kappa,\mu)$} if
\begin{equation}
[r(x),r(y)]_\frakg-r(r(x)\cdot v-r(y)\cdot x+\lambda
[x,y]_\frakk)=\kappa[\beta(x),\beta(y)]_\frakg+\mu\beta([x,y]_\frakk),\quad
\forall x,y\in \frakk.\mlabel{eq:gmcybe}
\end{equation}
\item
A linear map $r:\frakk\to \frakg$ is called an {\bf $\calo$-operator of weight $\lambda$} if
\begin{equation}
[r(x),r(y)]_{\frak{g}}=r(r(x)\cdot y-r(y)\cdot
x+\lambda[x,y]_{\frak{k}}),\quad \forall x,y\in\frak{k}.
\mlabel{eq:gcybe}
\end{equation}
\item
Let $(\frakk,[\,,\,]_\frakk,\pi)$ be the $\frakg$-Lie algebra $(\frakg,[\,,\,]_\frakg,{\rm ad})$. Then an $\calo$-operator $r:\frakg\to \frakg$ becomes what is known as a {\bf Rota-Baxter operator of weight $\lambda$} satisfying
\begin{equation}
[r(x),r(y)]_\frakg=r([r(x),y]_\frakg+[x,r(y)]_\frakg+\lambda[x,y]_\frakg),\quad \forall x,y\in
\frak{g}.\label{eq:rotabaxter}
\end{equation}
A Lie algebra equipped with a Rota-Baxter operator is called a {\bf Rota-Baxter Lie algebra.}
\end{enumerate}
\mlabel{de:oop}
}
\end{defn}
\begin{remark}
{\rm
\begin{enumerate}
\item
We include the parameters $\nu,\kappa,\mu,\lambda$ in the definition in order to uniformly treat the different cases when the parameters vary.
\item
Rota-Baxter operators on associative algebras were introduced by the mathematician Glenn Baxter~\cite{Ba} in 1960 and have recently found many applications especially in the algebraic approach of Connes and Kreimer to renormalization of quantum field theory~\mcite{CK1,CK2} . For further details, see the survey articles~\mcite{EG,Gu,R}. See also~\mcite{BGN} for the relationship between Rota-Baxter operators on associative algebras and the associative CYBE motivated by the study of this paper.
\item
If $\lambda\neq 0$, then $r$ is an $\calo$-operator of
weight $\lambda$ if and only if $r/\lambda$ is an $\calo$-operator of
weight 1.
\item
When $\lambda=1$, the difference of the two sides of Eq.~(\ref{eq:gcybe}) has
appeared in the work of Y. Kosmann-Schwarzbach and F. Magri under
the name { Schouten curvature}, which is the algebraic version of
the contravariant analogue of the { Cartan curvature of
Lie-algebra valued one-form on a Lie group} (see~\cite{KoM} for
details).
\end{enumerate}
}
\end{remark}

When $\frakk$ in Definition~\mref{de:oop} is taken to be a vector
space regarded as an abelian Lie algebra, then $(\frakk,\pi)$ is a
$\frakg$-Lie algebra means that $\pi:\frakg\to \frak{gl}(\frakk)$ is
a linear representation of $\frakg$. Thus the above definition has
the following variation (with $\nu=\kappa$).

\begin{defn}
{\rm Let $\frak{g}$ be a Lie algebra and $V$ be a vector space. Let
$\rho:\frak{g}\rightarrow \frak{gl}(V)$ be a linear representation
of $\frak{g}$. Suppose that $\beta:V\rightarrow\frak{g}$ is an
antisymmetric of \bwt $\kappa$, $\frakg$-invariant of \bwt $\kappa$
linear map. Let $r:V\to \frakg$ be a linear map. The pair
$(r,\beta)$ or simply $r$ is called an {\bf \tto $\calo$-operator
with \bop $\beta$ of \bwt $\kappa$} if
\begin{equation}
[r(u),r(v)]-r(r(u)\cdot v-r(v)\cdot
u)=\kappa[\beta(u),\beta(v)],\quad \forall u,v\in
V.\mlabel{eq:kgmcybe}
\end{equation}
} \mlabel{de:oopv}
\end{defn}

When $\kappa=0$, we obtain the $\calo$-operator defined by Kupershmidt~\mcite{Ku} and (the operator form of) the classical Yang-Baxter equation (CYBE) of Bordemann~\mcite{Bo}. When $\kappa=-1$, Eq.~(\ref{eq:kgmcybe}) is called the { modified
classical Yang-Baxter equation} ({\rm MCYBE}) in~\cite{Bo,Ko,Se}.

The following theorem displays the close connection between \tto $\calo$-operators and the double Lie algebra structures on $\frakk$ in Condition~\mref{con:lie2}.

\begin{theorem}
Let $\frak g$ be a Lie algebra and $(\frak k,\pi)$ be a $\frak
g$-Lie algebra. Let $r_\pm: \frak k\rightarrow \frak g$ be two
linear maps, $\lambda,\nu, \kappa,\mu\in \RR$ and $r$ and $\beta$ be
defined by Eq.~$($\ref{eq:ralpha}$)$.
\begin{enumerate} \item Suppose $r$ is an \tto $\calo$-operator of weight $\lambda$ with \bop $\beta$
of \bwt $(\nu, \kappa,\mu)$ for $\nu\neq0$. Then Condition~\mref{con:lie2} holds.
\mlabel{it:oop}
\item
Suppose $\beta$ satisfies $\beta(\xi\cdot
x)=[\xi,\beta(x)]_{\frak{g}}$, for all
$\xi\in\frak{g},x\in\frak{k}$, that is, $\beta$ is
$\frakg$-invariant of \bwt 1 $($or equivalently, a $\frak g$-module
homomorphism$)$. Then $r$ satisfies Eq.~$($\ref{eq:gmcybe}$)$ for
$\kappa=-1$, $\mu=\pm\lambda$ if and only if the following equation
holds:
\begin{equation}
[r_{\pm}(x),r_{\pm}(y)]_{\frak{g}}-r_{\pm}([x,y]_R)=0,\quad \forall
x,y\in\frak{k}.\label{eq:pmcybe}
\end{equation}
\mlabel{it:rpm}
\end{enumerate}
\label{thm:ansatz}
\end{theorem}
\begin{proof}
(\mref{it:oop} )
In order to prove that Eq.~(\ref{eq:gmcybe}) implies the Jacobi
identity for the bracket $[,]_R$ on $\frak{k}$, it is enough to
prove that
$$(k[\beta(x),\beta(y)]_{\frak{g}}+\mu\beta([x,y]_{\frak{k}}))\cdot
z+{\rm cycl}.=0.$$ In fact, we will prove that
\begin{equation}
k[\beta(x),\beta(y)]_{\frak{g}}\cdot z+{\rm cycl}.=0\label{eq:cycl1}
\end{equation}
and
\begin{equation}
\mu\beta([x,y]_{\frak{k}})\cdot z+{\rm cycl}.=0.\label{eq:cycl2}
\end{equation}
Eq.~(\ref{eq:cycl1}) has already been proved by Bordemann
~\cite{Bo}. In order to be self-contained, we give the details. For
any $x,y,z\in\frak{k}$,
\begin{eqnarray*}
k[\beta(x),\beta(y)]_{\frak{g}}\cdot z&=&k\beta(x)(\beta(y)\cdot
z)-k\beta(y)\cdot(\beta(x)\cdot z)\\
&=&-k\beta(\beta(y)\cdot z)\cdot x-k\beta(\beta(z)\cdot x)\cdot
y\quad ({\rm by\; antisymmetry})\\
&=&-k[\beta(y),\beta(z)]_{\frak{g}}\cdot
x-k[\beta(z),\beta(x)]_{\frak{g}}\cdot y\quad ({\rm by}\;
\frakg-{\rm invariance}).
\end{eqnarray*}
So Eq.~(\ref{eq:cycl1}) follows immediately. Moreover,
\begin{eqnarray*}
& &\mu\beta([x,y]_{\frak{k}})\cdot
z=-\mu\beta(z)\cdot[x,y]_{\frak{k}}\quad ({\rm by\;
antisymmetry})\\
&=&-\mu[\beta(z)\cdot x,y]_{\frak{k}}-\mu[x,\beta(z)\cdot
y]_{\frak{k}}\\
&=&\mu[\beta(x)\cdot z,y]_{\frak{k}}+\mu[x,\beta(y)\cdot
z]_{\frak{k}}\quad ({\rm by\; antisymmetry})\\
&=&\mu\beta(x)\cdot[z,y]_{\frak{k}}-\mu[z,\beta(x)\cdot
y]_{\frak{k}}+\mu\beta(y)\cdot[x,z]_{\frak{k}}-\mu[\beta(y)\cdot
x,z]_{\frak{k}}\\
&=&-\mu\beta([z,y]_{\frak{k}})\cdot
x-\mu\beta([x,z]_{\frak{k}})\cdot y+2\mu[\beta(x)\cdot
y,z]_{\frak{k}}\quad ({\rm by\;
antisymmetry})\\
&=&\mu\beta([y,z]_{\frak{k}})\cdot x+\mu\beta([z,x]_{\frak{k}})\cdot
y+2\mu\beta([x,y]_{\frak{k}})\cdot z\quad ({\rm by\; equivalence}).
\end{eqnarray*}
Therefore, Eq.~(\ref{eq:cycl2}) holds. So by
Proposition~\mref{pp:Liestructure}, Condition~\mref{con:lie2} holds.
\smallskip

\noindent
(\mref{it:rpm})
A direct computation gives
\begin{eqnarray*}
& &[(r\pm\beta)(x),(r\pm\beta)(y)]_{\frak{g}}-(r\pm\beta)(r(x)\cdot
y-r(y)\cdot
x+\lambda[x,y]_{\frak{k}})\\
&=&[r(x),r(y)]_{\frak{g}}-r(r(x)\cdot y-r(y)\cdot
x+\lambda[x,y]_{\frak{k}})+[\beta(x),\beta(y)]_{\frak{g}}\mp\lambda\beta([x,y]_{\frak{k}})
\pm([r(x),\beta(y)]_{\frak{g}}\\ & &-\beta(r(x)\cdot
y)+[\beta(x),r(y)]_{\frak{g}}+\beta(r(y)\cdot x))\\
&=&[r(x),r(y)]_{\frak{g}}-r(r(x)\cdot y-r(y)\cdot
x+\lambda[x,y]_{\frak{k}})+[\beta(x),\beta(y)]_{\frak{g}}\mp\lambda\beta([x,y]_{\frak{k}}),
\end{eqnarray*}
where the last equality follows from $\frakg$-invariance of \bwt 1.
So (\mref{it:rpm}) holds.
\end{proof}

\begin{remark} {\rm
When the bracket $[,]_{\frak{k}}$ on $\frak{k}$ is trivial and
$\kappa=-1$, Proposition~\mref{pp:Liestructure} and
Theorem~\ref{thm:ansatz} give Theorem 2.18 in ~\cite{Bo}.}
\end{remark}

The following results give the relations of $\calo$-operators with  Eq.~(\ref{eq:pmcybe}) and \tto $\calo$-operators.

\begin{theorem}
Let $\frak g$ be a Lie algebra and $(\frak k,\pi)$ be a $\frak
g$-Lie algebra. Let $r_\pm: \frak k\rightarrow \frak g$ be two
linear maps and let $\lambda\in \RR$ and $r$ and $\beta$ be defined
by Eq.~$($\ref{eq:ralpha}$)$. Suppose that $\beta$ is antisymmetric
of \bwt $\nu\neq0$, $\frakg$-invariant of \bwt $\kappa\neq0$ and
equivalent of \bwt $\lambda$.
\begin{enumerate}
\item
$(\frak{k}_\pm,[\,,\,]_{\pm},\pi)$ are $\frak{g}$-Lie algebras,
where $(\frak{k}_\pm,[\,,\,]_{\pm})$ are the new Lie algebra
structures on $\frak k$ defined by
\begin{equation}
[x,y]_{\pm}\equiv\lambda[x,y]_{\frak{k}}\pm2\beta(x)\cdot y,\quad
\forall x,y\in\frak{k}.\label{eq:pmbra}
\end{equation}
\mlabel{it:deliebra}
\item
$r$ is an \tto $\calo$-operator of weight $\lambda$ with \bop
$\beta$ of \bwt $(\nu,-1,\pm\lambda)$ for $\nu\neq0$ if and only if
$r_{\pm}:\frak{k}_\mp\rightarrow\frak{g}$ is an $\calo$-operators of
weight $1$, where $\frak{k}_\mp$ is equipped with the Lie bracket
$[,]_{\mp}$ defined by Eq.~$($\ref{eq:pmbra}$)$. \mlabel{it:pmbra}
\end{enumerate}
\mlabel{thm:deliebra}
\end{theorem}
\begin{proof}
(\mref{it:deliebra})
Since $\beta$ is antisymmetric, $[,]_{\pm}$ is antisymmetric.
Moreover, for any $x,y,z\in \frak{k}$, we have
\begin{eqnarray*}
& &[[x,y]_{\pm},z]_{\pm}+{\rm
cycl}.=[\lambda[x,y]_{\frak{k}}\pm2\beta(x)\cdot
y,z]_{\pm}+{\rm cycl}.\\
&=&(\lambda^2[[x,y]_{\frak{k}},z]_{\frak{k}}\pm2\lambda[\beta(x)\cdot
y,z]_{\frak{k}}\pm2\lambda\beta([x,y]_{\frak{k}})\cdot
z+4\beta(\beta(x)\cdot y)\cdot z)+{\rm cycl}.\\
&=&(\lambda^2[[x,y]_{\frak{k}},z]_{\frak{k}}\pm4\lambda\beta([x,y]_{\frak{k}})\cdot
z+4[\beta(x),\beta(y)]_{\frak{g}}\cdot z)+{\rm cycl}.,
\end{eqnarray*}
where the last equality follows from the $\frakg$-invariance of \bwt
$\kappa\neq0$ and equivalence of \bwt $\lambda$. So by
Theorem~\ref{thm:ansatz} the Jacobi identity for the bracket
$[,]_{\pm}$ on $\frak{k}$ holds. Moreover, for any $\xi\in\frak{g}$,
we have
\begin{eqnarray*}
&
&\xi\cdot[x,y]_{\pm}=\lambda\xi\cdot[x,y]_{\frak{k}}\pm2\xi\cdot(\beta(x)\cdot
y)\\
&=&\lambda[\xi\cdot x,y]_{\frak{k}}+\lambda[x,\xi\cdot
y]_{\frak{k}}\pm2\beta(\xi\cdot x)\cdot y\pm2\beta(x)\cdot(\xi\cdot
y)\quad ({\rm by}\; \frakg-{\rm invariance})\\
&=&[\xi\cdot x,y]_{\pm}+[x,\xi\cdot y]_{\pm}.
\end{eqnarray*}
So $(\frak{k}_\pm,\pi)$ equipped with the bracket $[,]_{\pm}$ on
$\frak k$ is a $\frak{g}$-Lie algebra.
\medskip

(\mref{it:pmbra}) The last conclusion follows from
Theorem~\ref{thm:ansatz}, Item~(\ref{it:deliebra}) and the following
computations:
\begin{eqnarray*}
&
&[r_{\pm}(x),r_{\pm}(y)]_{\frak{g}}-r_{\pm}([x,y]_R)\\
&=&[r_{\pm}(x),r_{\pm}(y)]_{\frak{g}}- r_{\pm}(r_{\pm}(x)\cdot
y-r_{\pm}(y)\cdot
x+\lambda[x,y]_{\frak{k}}\mp\beta(x)\cdot y\pm\beta(y)\cdot x)\\
&=&[r_{\pm}(x),r_{\pm}(y)]_{\frak{g}}-r_{\pm}(r_{\pm}(x)\cdot
y-r_{\pm}(y)\cdot x+[x,y]_{\mp})\quad ({\rm by\; antisymmetry}).
\end{eqnarray*}
\end{proof}

When $\frakk$ in Theorem~\mref{thm:deliebra} is taken to be a
vector space regarded as an abelian Lie algebra, we obtain the
following conclusions.

\begin{coro}
Let $\frak{g}$ be a Lie algebra and $V$ be a vector space. Let
$\rho:\frak{g}\rightarrow \frak{gl}(V)$ be a linear representation
of $\frak{g}$. Suppose that $\beta:V\rightarrow\frak{g}$ is
antisymmetric of \bwt $\kappa\neq 0$ and $\frakg$-invariant of \bwt
$\kappa\neq 0$.
\begin{enumerate}
\item
$(V_{\pm},[\,,\,]_{\pm},\rho)$ are $\frak{g}$-Lie algebras, where
$(V_{\pm},[\,,\,]_{\pm})$ are the Lie algebra structures on $V$
defined by
\begin{equation}
[u,v]_{\pm}\equiv\pm2\beta(u)\cdot v,\quad \forall u,v\in
V.\mlabel{eq:moliebrapm}
\end{equation}
\mlabel{it:useconclu1}
\item
Let $r:V\rightarrow\frak{g}$ be a linear map. Then $r$ is an \tto
$\calo$-operator with \bop $\beta$ of \bwt $-1$ if and only if $r\pm
\beta:V_{\mp}\to\frak{g}$ are $\calo$-operators of weight 1, where
$V_{\mp}$ are equipped with the Lie brackets $[\,,\,]_{\mp}$ defined
by Eq.~$($\mref{eq:moliebrapm}$)$.
\end{enumerate}
\mlabel{co:useconclu}
\end{coro}

\subsection{Adjoint representations and Baxter Lie algebras}
\mlabel{sec:ad} \mlabel{ss:idbaxter}
We now consider the case of adjoint representations.
If $\frak{k}=\frak{g}$ with the trivial Lie bracket and $\pi={\rm ad}$,
then by Proposition~\ref{pp:Liestructure}, Theorem~\ref{thm:ansatz}
and Theorem~\ref{thm:deliebra} we have the following conclusion.
\begin{prop}
Let $\frak{g}$ be a Lie algebra and $R,\beta:\frak g\rightarrow
\frak g$ be two linear maps. Let $\beta$ be antisymmetric of \bwt
$\kappa$ and $\frak{g}$-invariant of \bwt $\kappa$, i.e.,
the following equation holds:
\begin{equation}
\kappa\beta([x,y])=\kappa[\beta(x),y]=\kappa[x,\beta(y)],\quad
\forall x,y\in\frak{g}.\label{eq:adk}
\end{equation}
 Suppose that $R$ is an \tto $\calo$-operator with \bop $\beta$ of \bwt $\kappa$, i.e., the
 following equation holds:
\begin{equation}
[R(x),R(y)]-R([R(x),y]+[x,R(y)])=\kappa[\beta(x),\beta(y)],\quad
\forall x,y\in\frak{g}.\label{eq:adkmcybe}
\end{equation}
Then the product
$$
[x,y]_R=[R(x),y]+[x,R(y)],\quad \forall x,y\in\frak{g},
$$
defines a Lie bracket on $\frak{g}$. On the other hand, if $\beta$
satisfies Eq.~$($\ref{eq:adk}$)$ for $\kappa\neq 0$, then
$(\frak{g}_{\pm},[\,,\,]_{\pm},{\rm ad})$ are $\frak{g}$-Lie
algebras, where $(\frak{g}_{\pm},[\,,\,]_{\pm})$ are the new Lie
algebra structures defined by
\begin{equation}
[x,y]_{\pm}\equiv\pm2[\beta(x),y],\quad \forall
x,y\in\frak{g}.\mlabel{eq:admoliebrapm}
\end{equation}
Moreover, $R$ is an \tto $\calo$-operator with \bop $\beta$ of \bwt
$-1$, i.e., Eq.~(\ref{eq:adkmcybe}) holds for $\kappa=-1$, if and
only if $R\pm \beta:\frak{g}_{\mp}\to\frak{g}$ are
$\calo$-operators of weight 1, where $\frak{g}_{\mp}$ are equipped
with the Lie brackets $[\,,\,]_{\mp}$ defined by
Eq.~$($\mref{eq:admoliebrapm}$)$.
\end{prop}
\begin{remark}
{\rm Let $\frak{g}$ be a Lie algebra. A linear endomorphism $\beta$
of $\frak{g}$ satisfying Eq.~(\ref{eq:adk}) for $\kappa\neq 0$ is
called an {\bf intertwining operator} in ~\cite{RS1}, where it is
used to construct {\bf compatible Poisson brackets}. If
$\beta:\frak{g}\rightarrow\frak{g}$ is an intertwining operator on
$\frak{g}$, then it is also an {\bf averaging operator} ~\cite{Ag,R}
in the Lie algebraic context, namely,
$$[\beta(x),\beta(y)]=\beta([x,\beta(y)])=\beta([\beta(x),y]),\quad
\forall x,y\in\frak{g},
$$
 and is a {\bf Nijenhuis tensor}, namely,
\begin{equation}
[\beta(x),\beta(y)]+\beta^2([x,y])=\beta([\beta(x),y]+[x,\beta(y)]),\quad
\forall x,y\in\frak{g}.\label{eq:nijenhuis}
\end{equation}
   }
\end{remark}

Let the $\frak{g}$-Lie algebra
$(\mathfrak k, \pi)$ be $(\frak g, {\rm ad})$. It is obvious that $\beta={\rm
id}:\frak{g}\to \frak{g}$ satisfies the conditions of
Proposition~\ref{pp:Liestructure}, Theorem~\ref{thm:ansatz} and
Theorem~\ref{thm:deliebra} and in this case,
Eq.~(\ref{eq:gmcybe}) takes the following form (set $r=R$):
\begin{equation}
[R(x),R(y)]-R([R(x),y]+[x,R(y)]+\hat{\lambda}[x,y])=\hat{\kappa}[x,y],\quad
\forall x,y\in\frak{g},\label{eq:hatlambdak}
\end{equation}
for $\hat{\lambda}=\lambda$ and $\hat{\kappa}=\kappa+\mu$. When
$\hat{\kappa}=-1\pm\hat{\lambda}$, by Theorem~\ref{thm:deliebra},
$R$ satisfies Eq.~(\ref{eq:hatlambdak}) if and only if $R\pm{\rm
id}$ is a Rota-Baxter operator of weight $\hat{\lambda}\mp2$. Note
that when $\hat{\lambda}=0$, Eq.~(\ref{eq:hatlambdak}) takes the
following form
\begin{equation}
[R(x),R(y)]-R([R(x),y]+[x,R(y)])=\kappa[x,y],\quad \forall
x,y\in\frak{g},\label{eq:kmcyb}
\end{equation}
for $\kappa=\hat{\kappa}$. When $\kappa=-1$, Eq.~(\ref{eq:kmcyb})
becomes
\begin{equation}
[R(x),R(y)]-R([R(x),y]+[x,R(y)])=-[x,y],\quad \forall
x,y\in\frak{g}.\label{eq:-1mcyb}
\end{equation}
A Lie algebra equipped with a linear endomorphism satisfying
Eq.~(\ref{eq:-1mcyb}) is called a {\bf Baxter Lie algebra}
in~\mcite{Bo}. We note the difference between a Baxter Lie algebra
and a Rota-Baxter Lie algebra defined in Definition~\mref{de:oop}.
Moreover, the equivalence of the facts that $R$ satisfies
Eq.~(\ref{eq:-1mcyb}) and $R\pm{\rm id}$ is a Rota-Baxter operator
of weight $\mp2$ was pointed out in ~\cite{E, Se}.

\section{\Tto $\calo$-operators, the \tte CYBE and type II quasitriangular Lie bialgebras}
\label{sec:liebialgebra} In this section, we define the \tte CYBE and apply the study in Section~\ref{sec:lax} to investigate the relationship between \tto $\calo$-operators and the \tte CYBE. We also introduce the concept of type II quasitriangular Lie bialgebras from type II CYBE as a parallel concept of quasitriangular Lie bialgebras from CYBE. We then explicitly describe the Drinfeld's doubles and Manin triples of type II quasitriangular Lie bialgebras.

\subsection{Lie bialgebras and the \tte CYBE}
\label{ss:cybe1}
We recall the following concepts~\mcite{CP}.
\begin{defn}
{\rm
Let $\mathfrak{g}$ be a Lie algebra.
\begin{enumerate}
\item A {\bf Lie bialgebra} structure on $\mathfrak{g}$ is a skew-symmetric $\mathbb R$-linear map
$\delta:\mathfrak{g}\rightarrow\mathfrak{g}\otimes\mathfrak{g}$,
called {\bf cocommutator}, such that $(\mathfrak{g},\delta)$ is a
Lie coalgebra and $\delta$ is a $1$-cocycle of $\mathfrak{g}$ with
coefficients in $\mathfrak{g}\otimes\mathfrak{g}$, that is, it
satisfies the following equation:
$$
\delta([x,y])=({\rm ad}(x)\otimes \id +\id\otimes {\rm
ad}(x))\delta(y)- ({\rm ad}(y)\otimes \id+\id\otimes {\rm
ad}(y))\delta(x),\;\; \forall x,y\in\mathfrak{g}.
$$
\item
A Lie bialgebra $(\mathfrak{g},\delta)$ is called
{\bf coboundary} if $\delta$ is a $1$-coboundary, that is, there
exists an $r\in\mathfrak{g}\otimes\mathfrak{g}$ such that
\begin{equation}
\delta(x)=({\rm ad}(x)\otimes \id+\id\otimes {\rm ad}(x))r,\;\; \forall
x\in\mathfrak{g}.\label{eq:1coboundary}
\end{equation}
We usually denote the coboundary Lie bialgebra by $(\frak{g},r)$ or simply $\frakg$.
\item
A {\bf Manin triple} is a triple $(\fraka,\fraka_+,\fraka_-)$ of Lie algebras together with a nondegenerate symmetric invariant bilinear form $\frakB(\,,\,)$ on $\fraka$, such that
\begin{enumerate}
\item
$\fraka_+$ and $\fraka_-$ are Lie subalgebras of $\fraka$;
\item
$\fraka=\fraka_+\oplus \fraka_-$ as vector spaces;
\item
$\fraka_+$ and $\fraka_-$ are isotropic for $\frakB(\,,\,)$.
\end{enumerate}
\end{enumerate}
}
\mlabel{de:liebi}
\end{defn}

We recall the following basic results on Lie bialgebras and Manin triples.
\begin{prop}{\rm (\cite{Dr})}
Let $(\frak{g},\delta)$ be a Lie bialgebra. Let $\cald(\frakg)\equiv \frakg \oplus \frakg^*$. Then $(\cald(\frakg),\frakg,\frakg^*)$ is a { Manin triple} with respect to the bilinear form
\begin{equation}
\frak{B}_p((x,a^*),(y,b^*))=\langle a^*,y\rangle +\langle
x,b^*\rangle ,\quad \forall x,y\in \frak{g},\; a^*,b^*\in
\frak{g}^*,\label{eq:biform1}
\end{equation}
on $\mathcal{D}(\frak{g})$.
Explicitly, the Lie algebra structure on $\mathcal{D}(\frak{g})$ is
given by
\begin{equation}
[(x,a^*),(y,b^*)]_{\cald(\frakg)}=([x,y]+{\rm ad}^*(a^*)y-{\rm
ad}^*(b^*)x,[a^*,b^*]_{\delta}+{\rm ad}^*(x)b^*-{\rm
ad}^*(y)a^*),\  \forall
x,y\in\frak{g},a^*,b^*\in\frak{g}^*,\mlabel{eq:drindouliestr}
\end{equation}
where the Lie algebra structure $[\,,\,]_{\delta}$ on $\frak{g}^*$
is defined by
\begin{equation}
\langle[a^*,b^*]_{\delta},x\rangle=\langle a^*\otimes
b^*,\delta(x)\rangle,\quad \forall
x\in\frak{g},a^*,b^*\in\frak{g}^*.\mlabel{eq:adddelta}
\end{equation}
\label{pp:doublespace}
\end{prop}
$\cald(\frakg)$ is called the {\bf Drinfeld's double} for the Lie bialgebra $(\frakg,r)$.

\begin{prop}{\rm (\cite{CP})}
Let $\mathfrak{g}$ be a Lie algebra and
$r\in\mathfrak{g}\otimes\mathfrak{g}$. The linear map $\delta$
defined by Eq.~$($\ref{eq:1coboundary}$)$ is the commutator of a Lie
bialgebra structure on $\mathfrak{g}$ if and only if the following
conditions are satisfied for all $x\in\frak{g}$:
\begin{enumerate}
\item
 $({\rm ad}(x)\otimes \id+\id\otimes{\rm ad}(x))(r+\sigma(r))=0$, that is, the symmetric part of $r$ is invariant.
 \label{it:syinvariant}
\item
$({\rm ad}(x)\otimes \id\otimes \id+\id\otimes{\rm ad}(x)\otimes
\id+\id\otimes \id\otimes{\rm
ad}(x))([r_{12},r_{13}]+[r_{12},r_{23}]+[r_{13},r_{23}])=0$.
\end{enumerate}
Here $\sigma: \frak{g}^{\ot 2} \rightarrow \frak{g}^{\ot 2}$ is the
twisting operator defined by
$$
\sigma (x \otimes y) = y\otimes x,\quad \forall x, y\in \frak{g}.
$$
\label{pp:liebialgebra}
\end{prop}

In the following we call $r=\sum\limits_i a_i\ot b_i\in
\frak{g}^{\ot 2}$ {\bf skew-symmetric} (resp. {\bf symmetric}) if
$r=-\sigma(r)$ (resp. $r=\sigma(r)$). Moreover, we use the
notations (in the universal enveloping algebra $U(\frak g)$):
$$
r_{12}=\sum_ia_i\otimes b_i\otimes 1,\quad r_{13}=\sum_{i}a_i\otimes
1\otimes b_i,\quad r_{23}=\sum_i1\otimes a_i\otimes b_i,
$$
and
$$
[r_{12},r_{13}]=\sum_{i,j}[a_i,a_j]\otimes b_i\otimes b_j,\;
[r_{13},r_{23}]=\sum_{i,j}a_i\otimes
a_j\otimes[b_i,b_j],\;[r_{23},r_{12}]=\sum_{i,j}a_j\otimes
[a_i,b_j]\otimes b_i.
$$

The equation
\begin{equation}
\textbf{C}(r)\equiv[r_{12},r_{13}]+[r_{12},r_{23}]+[r_{13},r_{23}]=0\label{eq:cybe}
\end{equation}
is called the (tensor form of) the {\bf classical Yang-Baxter equation}
(CYBE). One should not confuse it with the (operator form of) CYBE of Bordemann~\cite{Bo}, though under certain conditions the former
is equivalent to a particular case of the later that we will elaborate next.

A coboundary Lie bialgebra
$(\frak{g},r)$ arising from a solution of CYBE is said to be {\bf
quasitriangular}, whereas a coboundary Lie bialgebra $(\frak{g},r)$
arising from a skew-symmetric solution of CYBE is said to be {\bf
triangular} ~\cite{BD, CP}. Note that for any coboundary Lie
bialgebra $(\frak{g},r)$, the condition (\ref{it:syinvariant}) in
Proposition~\ref{pp:liebialgebra} holds automatically.

For any $r=\sum\limits_ia_i\otimes b_i\in \frak g\otimes \frak g$,
we set
$$r_{21}=\sum_ib_i\otimes a_i\otimes 1,\quad r_{32}=\sum_i1\otimes
b_i\otimes a_i,\quad r_{31}=\sum_ib_i\otimes 1\otimes a_i.$$
Moreover, we set
\begin{equation}
 [(a_1\ot a_2\ot a_3), (b_1\ot b_2\ot b_3)] =
  [a_1, b_1]\ot [a_2, b_2] \ot [a_3, b_3], \quad \forall\, a_i,b_i\in \frakg, i=1,2,3.
\notag 
\end{equation}

\begin{defn}
{\rm Let $\frak{g}$ be a Lie algebra. Fix $\epsilon\in \mathbb{R}$.
The equation
\begin{equation}
[r_{12},r_{13}]+[r_{12},r_{23}]+[r_{13},r_{23}]=\epsilon
[(r_{13}+r_{31}),(r_{23}+r_{32})] \mlabel{eq:ecybe}
\end{equation}
is called the {\bf \tte classical Yang-Baxter equation of \ewt
$\epsilon$} (or {\bf ECYBE of \ewt $\epsilon$} in short). }
\mlabel{de:ecybe}
\end{defn}

\begin{remark}
{\rm \begin{enumerate} \item
 When $\epsilon=0$ or $r$ is
skew-symmetric, then the ECYBE of \ewt $\epsilon$ is the same as the CYBE in
Eq.~(\mref{eq:cybe}):
\item
If the symmetric part $\beta$ of $r$ is invariant, by the proof of
Theorem~\mref{thm:cybea} below, for any $a^*,b^*,c^*\in\frak{g}^*$, we
have
\begin{eqnarray*}
\langle[r_{13}+r_{31},r_{23}+r_{32}],a^*\otimes b^*\otimes
c^*\rangle&=&\langle
4[\beta(a^*),\beta(b^*)],c^*\rangle=\langle4\beta({\rm
ad}^*(\beta(a^*))b^*),c^*\rangle\\
&=&\langle[r_{23}+r_{32},r_{12}+r_{21}],a^*\otimes b^*\otimes
c^*\rangle\\
\langle[r_{13}+r_{31},r_{23}+r_{32}],a^*\otimes b^*\otimes
c^*\rangle&=&\langle
4[\beta(a^*),\beta(b^*)],c^*\rangle=\langle-4\beta({\rm
ad}^*(\beta(b^*))a^*),c^*\rangle\\
&=&\langle[r_{12}+r_{21},r_{13}+r_{31}],a^*\otimes b^*\otimes
c^*\rangle.
\end{eqnarray*}
So in this case, the ECYBE of \ewt $\epsilon$ is equivalent to either one of
the following two equations:
$$
[r_{12},r_{13}]+[r_{12},r_{23}]+[r_{13},r_{23}]=\epsilon[r_{23}+r_{32},r_{12}+r_{21}],
$$
$$
[r_{12},r_{13}]+[r_{12},r_{23}]+[r_{13},r_{23}]=\epsilon[r_{12}+r_{21},r_{13}+r_{31}].
$$
\end{enumerate}
}
\end{remark}

\subsection{\Tto $\calo$-operators and the ECYBE}
We now study the relationship between \tto $\calo$-operators and
solutions of the ECYBE, generalizing the well-known relationship
between the operator form and tensor form of the CYBE~\mcite{KoM}.

Let $\frak{g}$ be a Lie algebra and $r\in\frak{g}\otimes\frak{g}$.
Since $\frakg$ is assumed to be finite-dimensional, we will be able to identify $r$ with the linear map $r:\frak{g}^*\rightarrow
\frak{g}$ through
\begin{equation}
\langle r(a^*),b^*\rangle =\langle a^*\otimes b^*,r\rangle,\quad
\forall a^*,b^*\in \frak{g}^*.\label{eq:idenrmap}
\end{equation}
We will do this throughout the rest of the paper. Moreover, $r^t:\frak{g}^*\rightarrow \frak{g}$ is defined as
$$
\langle a^*,r^t(b^*)\rangle =\langle a^*\otimes b^*,r\rangle ,\quad
\forall a^*,b^*\in \frak{g}^*.
$$
 Note that $r^t$ is just the linear map (from $\frak{g}^*$ to
$\frak{g}$) induced by $\sigma(r)$. We also use the following
notations:
\begin{equation}
\alpha=(r-\sigma(r))/2=(r-r^t)/2,\quad
\beta=(r+\sigma(r))/2=(r+r^t)/2,\label{eq:alphabeta}
\end{equation}
that is, $\alpha$ and $\beta$ are the {\bf skew-symmetric part} and
{\bf symmetric part} of $r$ respectively, and in this case
$r=\alpha+\beta$ and $r^t=-\alpha+\beta$.
\begin{lemma}
Let $\frak{g}$ be a Lie algebra and $\beta\in
\frak{g}\otimes\frak{g}$ be symmetric. Then the following conditions
are equivalent.
\begin{enumerate}
\item
 $\beta\in\frak{g}\otimes\frak{g}$ is invariant, that is,
$({\rm ad}(x)\otimes \id+\id\otimes{\rm ad}(x))\beta=0$, for any
$x\in\frak{g}$;\label{it:betainvariant}
\item
$\beta:\frak{g}^*\rightarrow\frak{g}$ is antisymmetry, that is,
${\rm ad}^*(\beta(a^*))b^*+{\rm ad}^*(\beta(b^*))a^*=0$, for any
$a^*,b^*\in\frak{g}^*$;\label{it:betaantisymmetry}
\item
$\beta:\frak{g}^*\rightarrow\frak{g}$ is $\frakg$-invariant, that is,
$\beta({\rm ad}^*(x)a^*)=[x,\beta(a^*)]$, for any $x\in\frak{g}$,
$a^*\in\frak{g}^*$.\label{it:betaginvariant}
\end{enumerate} \label{le:symmetry}
\end{lemma}
\begin{proof}
Bordemann in ~\cite{Bo} pointed out the equivalence of
(\ref{it:betaantisymmetry}) and (\ref{it:betaginvariant}). For completeness, we shall prove
(\ref{it:betainvariant})$\Leftrightarrow$(\ref{it:betaantisymmetry})
and
(\ref{it:betainvariant})$\Leftrightarrow$(\ref{it:betaginvariant}). In fact, for any $x\in\frak{g}$,
$a^*,b^*\in\frak{g}^*$,
\begin{eqnarray*}
\langle({\rm ad}(x)\otimes \id+\id\otimes{\rm ad}(x))\beta,a^*\otimes
b^*\rangle&=&\langle\beta,-({\rm ad}^*(x)a^*)\otimes
b^*\rangle+\langle\beta,-a^*\otimes({\rm ad}^*(x)b^*)\rangle\\
&=&\langle a^*,[x,\beta(b^*)]\rangle+\langle
[x,\beta(a^*)],b^*\rangle\quad ({\rm by\; symmetry})\\
&=&\langle{\rm ad}^*(\beta(b^*))a^*+{\rm
ad}^*(\beta(a^*))b^*,x\rangle.
\end{eqnarray*}
So
(\ref{it:betainvariant})$\Leftrightarrow$(\ref{it:betaantisymmetry}).
Moreover,
\begin{eqnarray*}
\langle({\rm ad}(x)\otimes \id+\id\otimes{\rm ad}(x))\beta,a^*\otimes
b^*\rangle&=&\langle\beta,-({\rm ad}^*(x)a^*)\otimes
b^*\rangle+\langle\beta,-a^*\otimes({\rm ad}^*(x)b^*)\rangle\\
&=&\langle-\beta({\rm ad}^*(x)a^*)+[x,\beta(a^*)],b^*\rangle.
\end{eqnarray*}
So
(\ref{it:betainvariant})$\Leftrightarrow$(\ref{it:betaginvariant}).
\end{proof}

Note that the condition (\ref{it:betainvariant}) in
Lemma~\ref{le:symmetry} is exactly the condition
(\ref{it:syinvariant}) of Proposition~\ref{pp:liebialgebra}.

\begin{lemma}~$($\mcite{KoM}$)$\quad
Let $\frak{g}$ be a Lie algebra and $r\in \frak{g}\otimes \frak{g}$.
Let $\alpha,\beta:\frak g^*\rightarrow \frak g$ be the two linear
maps given by Eq.~$($\ref{eq:alphabeta}$)$. Then the bracket
$[,]_{\delta}$ defined by Eq.~$($\ref{eq:adddelta}$)$ satisfies
\begin{equation}
[a^*,b^*]_{\delta}={\rm ad}^*(r(a^*))b^*+{\rm
ad}^*(r^t(b^*))a^*,\quad \forall
a^*,b^*\in\frak{g}^*.\label{eq:deltap}
\end{equation}
Moreover, if the symmetric part $\beta$ of $r$ is invariant, then
\begin{equation}
[a^*,b^*]_{\delta}={\rm ad}^*(\alpha(a^*))b^*-{\rm
ad}^*(\alpha(b^*))a^*,\quad \forall
a^*,b^*\in\frak{g}^*.\label{eq:braalpha}
\end{equation}
\label{le:bradelta}
\end{lemma}
We supply a proof to be self-contained.
\begin{proof}
 Let $\{e_i\}_{1\leq i\leq {\rm dim}\frak{g}}$ be a basis
of $\frak{g}$ and $\{e_i^*\}_{1\leq i\leq {\rm dim}\frak{g}}$ be its
dual basis. Then the first conclusion holds due to the following
equations:
\begin{eqnarray*}
[e_k^*,e_l^*]_{\delta}&=&\sum_s\langle e_k^*\otimes
e_l^*,\delta(e_s)\rangle e_s^*=\sum_s\langle e_k^*\otimes
e_l^*,({\rm ad}(e_s)\otimes
\id+\id\otimes{\rm ad}(e_s))r\rangle e_s^*\\
&=&\sum_{s,t}(a_{tl}c_{st}^k+a_{kt}c_{st}^l)e_s^*={\rm
ad}^*(r(e_k^*))e_l^*+{\rm ad}^*(r^t(e_l^*))e_k^*.
\end{eqnarray*}
The last conclusion follows from Lemma~\ref{le:symmetry}.
\end{proof}

 The above lemma motivates us to apply the study in Section~\ref{sec:lax}.
More precisely, we have the following results.

\begin{prop}
Let $\frak{g}$ be a Lie algebra and $r\in \frak{g}\otimes \frak{g}$.
Let $\alpha,\beta:\frak g^*\rightarrow \frak g$ be two linear maps
given by Eq.~(\ref{eq:alphabeta}). Suppose that $\beta$, regarded as
an element of $\frak{g}\otimes\frak{g}$, is invariant.
\begin{enumerate}
\item
$(\frak{g},r)$ becomes a (coboundary) Lie bialgebra if $\alpha$ is
an \tto $\calo$-operator with \bop $\beta$ of \ewt $\kappa\in \RR$,
namely the following equation holds:
\begin{equation}
[\alpha(a^*),\alpha(b^*)]-\alpha({\rm ad}^*(\alpha(a^*))b^*-{\rm
ad}^*(\alpha(b^*))a^*)=\kappa[\beta(a^*),\beta(b^*)],\quad \forall
a^*,b^*\in\frak{g}^*.\label{eq:kmcybe}
\end{equation}
\mlabel{it:conclusion1}
\item~$($\mcite{KoM}$)$\quad
The following conditions are equivalent:

\begin{enumerate}

\item $\alpha$ is an \tto $\calo$-operator with \bop $\beta$ of \ewt
$-1$, i.e., Eq.~$($\ref{eq:kmcybe}$)$ $($with $\kappa=-1$$)$ holds;

\item $r$ $($resp. $-r^t$$)$ satisfies the following equation:
\begin{equation}
[r(a^*),r(b^*)]=r([a^*,b^*]_{\delta}),\quad \forall a^*,b^*\in
\frak{g}^*\label{eq:rrtho}
\end{equation}
\begin{equation}
({\rm resp.}\quad
[(-r^t)(a^*),(-r^t)(b^*)]=(-r^t)([a^*,b^*]_{\delta}),\quad \forall
a^*,b^*\in \frak{g}^*)\label{eq:trrtho};
\end{equation}

\item $r$ $($resp. $-r^t$$)$ is an $\calo$-operator of
weight 1, that is, $r$ $($resp. $-r^t$$)$ satisfies the following
equation:
\begin{equation}
[r(a^*),r(b^*)]=r({\rm ad}^*(r(a^*))b^*-{\rm
ad}^*(r(b^*))a^*+[a^*,b^*]_{-}),\quad  \forall
a^*,b^*\in\frak{g}^*,\label{eq:coadjointgcybe}
\end{equation}
\begin{equation}
({\rm resp}.\;\; [(-r^t)(a^*),(-r^t)(b^*)]=(-r^t)({\rm
ad}^*((-r^t)(a^*))b^*-{\rm ad}^*((-r^t)(b^*))a^*+[a^*,b^*]_{+}),
\forall a^*,b^*\in\frak{g}^*)\label{eq:tcoadjointgcybe}
\end{equation}
where the brackets $[,]_{\pm}$ on $\frak{g}^*$ are defined by
\begin{equation}
[a^*,b^*]_{\pm}\equiv\pm2{\rm ad}^*(\beta(a^*))b^*,\quad \forall
a^*,b^*\in \frak{g}^*,\label{eq:pmbrac}
\end{equation}
and $(\frak{g}^*,{\rm ad}^*)$ equipped with the brackets $[,]_{\pm}$
on $\frak{g}^*$ are $\frak{g}$-Lie algebras.
\end{enumerate}
\mlabel{it:conclusion2}
\end{enumerate}
\label{pp:conclusion}
\end{prop}
\begin{proof}
(\mref{it:conclusion1})
By Lemma~\mref{le:bradelta}, we
see that $(\frak{g},r)$ becomes a (coboundary) Lie bialgebra if the
bracket $[\,,\,]_{\delta}$ defined by Eq.~(\mref{eq:deltap}) is a
Lie structure on $\frak{g}^*$. Further by Lemma~\mref{le:symmetry},
$\beta$ is antisymmetric of \bwt $\nu\neq0$ and $\frakg$-invariant
of \bwt $\kappa\neq0$. Then the conclusion follows from
Theorem~\ref{thm:ansatz}.(\mref{it:oop}) by setting
$(\frak{k},\pi)=(\frak{g}^*,{\rm ad}^*)$ with trivial
 Lie bracket, $r_{+}=r$ and $r_{-}=-r^t$.
\medskip

\noindent (\mref{it:conclusion2}) It follows from
Theorem~\ref{thm:ansatz} and Theorem~\ref{thm:deliebra} by setting
$(\frak{k},\pi)=(\frak{g}^*,{\rm ad}^*)$ with trivial
 Lie bracket, $r_{+}=r$ and $r_{-}=-r^t$.
\end{proof}

The following theorem establishes a close relationship between \tto
$\calo$-operators on a Lie algebra $\frakg$ and solutions of the
ECYBE in $\frakg$.
\begin{theorem}
Let $\frakg$ be a Lie algebra and let $r\in \frakg\otimes \frakg$
which is identified as a linear map from $\frakg^*$ to $\frakg$.
Define $\alpha$ and $\beta$ by Eq.~$($\mref{eq:alphabeta}$)$.
Suppose that the symmetric part $\beta$ of $r$ is invariant. Then
$r$ is a solution of ECYBE of \ewt $\frac{\kappa+1}{4}$:
\begin{equation}
[r_{12},r_{13}]+[r_{12},r_{23}]+[r_{13},r_{23}]=\frac{\kappa+1}{4}[(r_{13}+r_{31}),
(r_{23}+r_{32})] \notag
\end{equation}
if and only if $\alpha$ is an \tto $\calo$-operator with \bop
$\beta$ of \bwt $\kappa$, i.e., Eq.~$($\mref{eq:kmcybe}$)$ holds.
\mlabel{thm:cybea}
\end{theorem}
\begin{proof}
Let $r=\sum\limits_{i,j}u_i\otimes v_i\in\frak{g}\otimes\frak{g}$
for $u_i,v_i\in\frak{g}$, then
\begin{eqnarray*}
\langle [r_{12},r_{13}],a^*\otimes b^*\otimes
c^*\rangle&=&\sum_{i,j}\langle[u_i,u_j],a^*\rangle\langle
v_i,b^*\rangle\langle v_j,c^*\rangle=\langle-r({\rm
ad}^*(r^t(b^*))a^*),c^*\rangle,\\
\langle [r_{12},r_{23}],a^*\otimes b^*\otimes
c^*\rangle&=&\sum_{i,j}\langle
u_i,a^*\rangle\langle[v_i,u_j],b^*\rangle\langle
v_j,c^*\rangle=\langle-r({\rm ad}^*(r(a^*))b^*),c^*\rangle,\\
\langle[r_{13},r_{23}],a^*\otimes b^*\otimes c^*\rangle&=&
\sum_{i,j}\langle u_i,a^*\rangle \langle u_j,b^*\rangle\langle
[v_i,v_j],c^*\rangle=\langle[r(a^*),r(b^*)],c^*\rangle.
\end{eqnarray*}
Therefore, $r$ is a solution of CYBE in $\frakg$ if and only if
Eq.~(\mref{eq:rrtho}) holds, i.e., $$[r(a^*), r(b^*)]=r({\rm
ad}^*(r(a^*))b^*+{\rm ad}^*(r^t(b^*))a^*),\;\;\forall a^*,b^*\in
\frakg^*.$$ Therefore, by Proposition~\mref{pp:conclusion}, for any
$a^*,b^*,c^*\in\frak{g}^*$, we have that
\begin{eqnarray*}
& &\langle[\alpha(a^*),\alpha(b^*)]-\alpha({\rm
ad}^*(\alpha(a^*))b^*-{\rm
ad}^*(\alpha(b^*))a^*)-\kappa[\beta(a^*),\beta(b^*)],c^*\rangle\\
&=&\langle[\alpha(a^*),\alpha(b^*)]-\alpha({\rm
ad}^*(\alpha(a^*))b^*-{\rm
ad}^*(\alpha(b^*))a^*)+[\beta(a^*),\beta(b^*)]-(\kappa+1)[\beta(a^*),\beta(b^*)],c^*\rangle\\
&=&\langle[r_{12},r_{13}]+[r_{12},r_{23}]+[r_{13},r_{23}],a^*\otimes
b^*\otimes
c^*\rangle-(\kappa+1)\langle[\beta_{13},\beta_{23}],a^*\otimes
b^*\otimes c^*\rangle\\
&=&\langle[r_{12},r_{13}]+[r_{12},r_{23}]+[r_{13},r_{23}]-
(\kappa+1)[\frac{r_{13}+r_{31}}{2},\frac{r_{23}+r_{32}}{2}],a^*\otimes
b^*\otimes c^*\rangle.
\end{eqnarray*}
So $r$ is a solution of the ECYBE of \ewt $(\kappa+1)/4$ if and only
if $\alpha$ is an \tto $\calo$-operator with \bop $\beta$ of \bwt
$\kappa$.
\end{proof}

Therefore by Proposition~\ref{pp:conclusion} and
Theorem~\mref{thm:cybea} (for $\kappa=-1$), we have the following
conclusion:
\begin{coro} ~$($\mcite{KoM}$)$\quad
Let $\frak{g}$ be a Lie algebra and $r\in \frak{g}\otimes \frak{g}$.
Let $\alpha,\beta:\frak g^*\rightarrow \frak g$ be two linear maps
given by Eq.~(\ref{eq:alphabeta}). Suppose that $\beta$, regarded as
an element of $\frak{g}\otimes\frak{g}$, is invariant. Then the
following conditions are equivalent:

\begin{enumerate}

\item $r$ is a solution of the CYBE;

\item $(\frak{g},r)$ is a quasitriangular  Lie bialgebra;

\item $r$ $($resp. $-r^t$$)$ is an $\calo$-operator of weight 1, that is,
$r$ $($resp. $-r^t$$)$ satisfies
Eq.~$($\ref{eq:coadjointgcybe}$)$$)$ $($resp.
Eq.~$($\ref{eq:tcoadjointgcybe}$)$$)$ with $\frak{g}^*$ equipped
with the bracket $[,]_{-}$ $($resp.
 $[,]_{+}$$)$ defined by
Eq.~$($\ref{eq:pmbrac}$)$.

\item $\alpha$ is an \tto $\calo$-operator with \bop $\beta$ of \ewt
$-1$, i.e., $\alpha$ and $\beta$ satisfy Eq.~$($\ref{eq:kmcybe}$)$
with $k=-1$;

\item $r$ $($resp. $-r^t$$)$ satisfies Eq.~$($\ref{eq:rrtho}$)$ $($resp.
Eq.~$($\ref{eq:trrtho}$)$$)$.

\end{enumerate}
\label{co:remark}
\end{coro}

\subsection{\Tto $\calo$-operators (of
\bwt 1) and type II CYBE} \label{ss:cybe2} Proposition~\mref{pp:conclusion} and Theorem~\mref{thm:cybea} reveal close connections of \tto $\calo$-operators $\alpha: \frakg^*\to \frakg$ (defined by Eq.~(\mref{eq:kmcybe})) with coboundary Lie bialgebras and ECYBE. Thus we
would like to study these operators in more detail. Note that, for
$\kappa=\eta^2 \kappa'$ with $\kappa,\kappa'\in \RR$ and $\eta\in
\RR^\times$, $\alpha$ is an \tto $\calo$-operator with \bop $\beta$
of \bwt $\kappa$ if and only if $\alpha$ is an \tto $\calo$-operator
with \bop $\eta\beta$ of \bwt $\kappa'$. Thus we only need to
consider the cases when $\kappa=0,1,-1$.

The case of $\kappa=-1$ is considered in Corollary~\mref{co:remark}.
The case of $\kappa=0$ has been considered by Kupershmidt ~\cite{Ku} as remarked before.
So we will next focus on the case when $\kappa=1$:
\begin{equation}
[\alpha(a^*),\alpha(b^*)]-\alpha({\rm ad}^*(\alpha(a^*))b^*-{\rm
ad}^*(\alpha(b^*))a^*)=[\beta(a^*),\beta(b^*)],\quad \forall
a^*,b^*\in\frak{g}^*.\label{eq:mcybe2}
\end{equation}
Note here $\beta$ regarded as an element of $\frakg\otimes \frakg$ is
invariant (Lemma~\ref{le:symmetry}).

\begin{defn}
{\rm Let $\frak{g}$ be a Lie algebra and
$r\in\frak{g}\otimes\frak{g}$. Then
\begin{equation}
[r_{12},r_{13}]+[r_{12},r_{23}]+[r_{13},r_{23}]=\frac{1}{2}[r_{13}+r_{31},r_{23}+r_{32}]\label{eq:type2cybe}
\end{equation}
is called the {\bf type II Classical Yang-Baxter Equation (type II
CYBE)}.}
\end{defn}

The following conclusion follows directly from
Theorem~\mref{thm:cybea} for $\kappa=1$.
\begin{prop}
Let $\frak{g}$ be a Lie algebra and $r\in\frak{g}\otimes\frak{g}$.
Let $\alpha,\beta:\frak g^*\rightarrow \frak g$ be two linear maps
given by Eq.~$($\ref{eq:alphabeta}$)$. Suppose that $\beta$,
regarded as an element of $\frak{g}\otimes\frak{g}$, is invariant.
Then $r$ is a solution of type II CYBE if and only if $\alpha$ is an
\tto $\calo$-operator with \bop $\beta$ of \bwt 1, i.e.,
Eq.~$($\mref{eq:mcybe2}$)$ holds. In this case, $(\frak{g},r)$
becomes a coboundary Lie bialgebra.\label{pp:cybeeqmcybe2}
\end{prop}

\begin{coro}
Let $\frak{g}$ be a Lie algebra and $r\in \frak{g}\otimes \frak{g}$.
Let $\alpha,\beta:\frak g^*\rightarrow \frak g$ be the two linear
maps given by Eq.~$($\ref{eq:alphabeta}$)$. Suppose that $\beta$,
regarded as an element of $\frak{g}\otimes\frak{g}$, is invariant.
Define $\hat{\frak{g}}=\frakg\ot \CC=\frak{g}\oplus i\frak{g}$,
where $i=\sqrt{-1}$, and regard $\hat{\frak{g}}$ as a real Lie
algebra. The following conditions are equivalent:

\begin{enumerate}
\item $r$ is a solution of the type II CYBE.\mlabel{it:1cybeeqmcybe2}

\item  $\alpha$ is an \tto $\calo$-operator with \bop $\beta$ of \bwt 1.\mlabel{it:2cybeeqmcybe2}

\item Regarding $\alpha$ and $i\beta$ as linear maps from $\hat{\frakg}^*=\frak{g}^*\oplus i\frakg^*$ to $\hat{\frakg}$,
 $\alpha$ is an \tto $\calo$-operator with \bop $i\beta$ of \ewt
$-1$.\mlabel{it:3cybeeqmcybe2}
\item $\alpha\pm i\beta$ are solutions of the CYBE in $\hat{\frakg}$\,.\mlabel{it:4cybeeqmcybe2}

\item $\alpha\pm i\beta$, regarded as linear maps from $\hat{\frakg}^*=\frak{g}^*\oplus i\frakg^*$ to $\hat{\frakg}$,
satisfy
\begin{equation} (\alpha\pm i\beta)([a^*,
b^*]_\delta)=[(\alpha\pm i\beta)(a^*),(\alpha\pm i\beta)(b^*)],\quad
\forall a^*,b^*\in \frak g^*\subset\hat{\frakg}^*=\frak{g}^*\oplus
i\frakg^*,\label{eq:irrtho}
\end{equation}
where the Lie algebra structure $[,]_\delta$ on $\frak g^*$ is given
by Eq.~$($\ref{eq:braalpha}$)$.\mlabel{it:ty2v}
\end{enumerate}
\label{co:2-cybe}

\end{coro}
\begin{proof}
By Proposition~\mref{pp:cybeeqmcybe2}, we have
(\mref{it:1cybeeqmcybe2})$\Leftrightarrow$(\mref{it:2cybeeqmcybe2}).
It follows from the definition of \tto $\calo$-operators that
(\mref{it:2cybeeqmcybe2})$\Leftrightarrow$(\mref{it:3cybeeqmcybe2}).
Moreover, applying Proposition~\mref{pp:conclusion} to $\hat{\frakg}$, we have
(\mref{it:3cybeeqmcybe2})$\Leftrightarrow$(\mref{it:4cybeeqmcybe2}). To prove
(\mref{it:4cybeeqmcybe2})$\Leftrightarrow$(\mref{it:ty2v}), we note that Proposition~\mref{pp:conclusion} also gives the equivalence of (\mref{it:4cybeeqmcybe2}) with the equation
\begin{equation}
(\alpha\pm i\beta)([u, v]_\delta)=[(\alpha\pm i\beta)(u),(\alpha\pm
i\beta)(v)],\quad \forall u,v\in \hat{\frakg}^*=\frak{g}^*\oplus
i\frakg^*,\mlabel{eq:iproeq}
\end{equation}
 where
$$[u,v]_{\delta}={\rm ad}^*(\alpha(u))v-{\rm ad}^*(\alpha(v))u,\quad
\forall u,v\in \hat{\frakg}^*=\frak{g}^*\oplus i\frakg^*.$$
Then (\mref{it:4cybeeqmcybe2})$\Leftrightarrow$(\mref{it:ty2v}) follows since
Eq.~(\mref{eq:iproeq}) $\Leftrightarrow$ Eq.~(\mref{eq:irrtho}) by the definition of \tto $\calo$-operators.
\end{proof}

\subsection{Type II quasitriangular Lie
bialgebras}\label{ss:type2double}
Considering the important role played by the Manin triple and Drinfeld's double from a Lie bialgebra in the
classification of the Poisson homogeneous spaces and symplectic
leaves of the corresponding Poisson-Lie groups~\cite{Dr1, HY,
Se1, Y}, it is important to investigate such Manin triple, as in~\mcite{HY,LS,St}. However, explicit structures for Manin triples have been obtained only in special cases, such as for quasitriangular Lie bialgebras in~\mcite{HY}.  Making use of the relationship between type II CYBE and \tto $\calo$-operators as displayed in Proposition~\ref{pp:cybeeqmcybe2}, we consider the following class of Lie bialgebras and obtain a similar explicit constructions of their Manin triples.

\begin{defn}
{\rm A coboundary Lie bialgebra $(\frak{g},r)$ is said to be {\bf
type II quasitriangular} if it arises from a solution $r$ of type II
CYBE given by Eq.~(\ref{eq:type2cybe}). }
\end{defn}

Our strategy is to express the Drinfeld's double $\cald(\frakg)$ as an extension of a Lie algebra by an abelian Lie algebra, both derived from the \tto $\calo$-operator associated to the solution $r$ of the type II CYBE.
We then obtain the structure of the Manin triple explicitly in terms of this extension.

\subsubsection{An Lie algebra extension associated to a type II
quasitriangular Lie bialgebra} We obtain the Lie algebra extension
from a type II quasitriangular Lie bialgebra by an exact sequence.
Let $\frak{g}$ be a Lie algebra and $r\in\frak{g}\otimes\frak{g}$.
Define the symmetric and skew-symmetric parts $\alpha$ and $\beta$
by Eq.~(\ref{eq:alphabeta}).
\begin{lemma}
With the same conditions as above, suppose that $(\frakg,r)$ is a
Lie bialgebra and $\beta$ is invariant.
\begin{enumerate}
\item
For any
$x\in\frak{g},a^*\in\frak{g}^*$, we have
$$
{\rm ad}^*(a^*)x=-[x,\alpha(a^*)]+\alpha({\rm ad}^*(x)a^*).
$$
\label{it:dualbraex}
\item
If $r$ is a solution of type II CYBE,
then
$$
[(-\alpha(a^*),a^*),(-\alpha(b^*),b^*)]_{\mathcal{D}(\frak{g})}=(-[\beta(a^*),\beta(b^*)],0),\;\;\forall
a^*,b^*\in \frak{g}^*.
$$
\mlabel{it:bradouble}
\end{enumerate}
\label{le:bradouble}
\end{lemma}

\begin{proof}
(\mref{it:dualbraex})
By Lemma~\ref{le:bradelta}, for any
$x\in\frak{g},a^*,b^*\in\frak{g}^*$, we have
\begin{eqnarray*}
\langle{\rm ad}^*(a^*)x,b^*\rangle&=&\langle
x,[b^*,a^*]_{\delta}\rangle=\langle x,-{\rm ad}^*(\alpha(a^*))b^*+
{\rm ad}^*(\alpha(b^*))a^*\rangle\\
&=&\langle-[x,\alpha(a^*)]+\alpha({\rm ad}^*(x)a^*),b^*\rangle,
\end{eqnarray*}
where the last equality follows from the fact that $\alpha$ is
skew-symmetric.
\medskip

\noindent
(\mref{it:bradouble})
Since $r$ is a solution of type II CYBE and $\beta$ is invariant, by
Proposition~\ref{pp:cybeeqmcybe2}, $\alpha$ and $\beta$ satisfy
Eq.~(\ref{eq:mcybe2}). So by Lemma~\ref{le:bradelta} and
Item~(\ref{it:dualbraex}), for any $a^*,b^*\in\frak{g}^*$ we have
\begin{eqnarray*}
& &[(-\alpha(a^*),a^*),(-\alpha(b^*),b^*)]_{\mathcal{D}(\frak{g})}\\
&=&([\alpha(a^*),\alpha(b^*)]-{\rm ad}^*(a^*)\alpha(b^*)+{\rm
ad}^*(b^*)\alpha(a^*),[a^*,b^*]_{\delta}-{\rm
ad}^*(\alpha(a^*))b^*+{\rm ad}^*(\alpha(b^*))a^*)\\
&=&([\alpha(a^*),\alpha(b^*)]+[\alpha(b^*),\alpha(a^*)]-\alpha({\rm
ad}^*(\alpha(b^*))a^*)-[\alpha(a^*),\alpha(b^*)]+\alpha({\rm
ad}^*(\alpha(a^*))b^*),0)\\
&=&(-[\alpha(a^*),\alpha(b^*)]+\alpha({\rm ad}^*(\alpha(a^*))b^*-
{\rm ad}^*(\alpha(b^*))a^*),0)=(-[\beta(a^*),\beta(b^*)],0).
\end{eqnarray*}
\end{proof}

Now let $(\frak{g},r)$ be a type II quasitriangular Lie bialgebra.
By Proposition~\mref{pp:liebialgebra}, $\beta\in \frak{g}\otimes \frak{g}$ is invariant.
Regarding $\beta$ as a linear map from $\frak{g}^*$ to $\frak{g}$,
we define
$$
\frak{f}={\rm Im}\beta,\quad \frak{f}^{\bot}={\rm Ker}\beta.
$$
Then by Lemma~\ref{le:symmetry}, $\frak{f}$ is an ideal of
$\frak{g}$. On the other hand, define $\hat{\frak{g}}=\frakg\ot
\CC=\frak{g}\oplus i\frak{g}$, where $i=\sqrt{-1}$, and regard
$\hat{\frak{g}}$ as a real Lie algebra. Let $\cald(\frakg)\equiv
\frakg\oplus \frakg^*$ be the Drinfeld's double defined in
Proposition~\mref{pp:doublespace}.

\begin{prop}
With the notations explained above, define two linear maps
$\Theta_{\pm}:\mathcal{D}(\frak{g})\rightarrow \hat{\frak{g}}$ by
\begin{equation}
\Theta_{\pm}(x,a^*)=x+\alpha(a^*)\pm i\beta(a^*),\quad \forall
x\in\frak{g},a^*\in\frak{g}^*.\label{eq:Theta1}
\end{equation}
Then $\Theta_{\pm}$ are homomorphisms of Lie algebras. Moreover,
${\rm Ker}\Theta_{+}={\rm Ker}\Theta_{-}$ is an abelian Lie
subalgebra of $\mathcal{D}(\frak{g})$. \label{pp:thetaho}
\end{prop}

\begin{proof} First, it is obvious
that for any $x,y\in \frak{g}$,
$$\Theta_{\pm}([x,y]_{\mathcal{D}(\frak{g})})=[\Theta_{\pm}(x),\Theta_{\pm}(y)]_{\hat{\frak{g}}}.$$
On the other hand, by Corollary~\ref{co:2-cybe}.(\mref{it:ty2v}),
Eq.~(\ref{eq:irrtho}) holds, that is, for any
$a^*,b^*\in\frak{g}^*$, we have
$$\Theta_{\pm}([a^*,b^*]_{\mathcal{D}(\frak{g})})=[\Theta_{\pm}(a^*),\Theta_{\pm}(b^*)]_{\hat{\frak{g}}}.$$
Furthermore, by Lemma~\ref{le:symmetry} and
Lemma~\mref{le:bradouble}.(\ref{it:dualbraex}), we have
\begin{eqnarray*}
\Theta_{\pm}([x,a^*]_{\mathcal{D}(\frak{g})})&=&\Theta_{\pm}({\rm
ad}^*(x)a^*-{\rm ad}^*(a^*)x)=\alpha({\rm ad}^*(x)a^*)-{\rm
ad}^*(a^*)x\pm i\beta({\rm
ad}^*(x)a^*)\\
&=&[x,\alpha(a^*)]\pm i[x,\beta(a^*)]=[x,(\alpha\pm
i\beta)(a^*)]_{\hat{\frak{g}}}=[\Theta_{\pm}(x),\Theta_{\pm}(a^*)]_{\hat{\frak{g}}}.
\end{eqnarray*}
So $\Theta_{\pm}$ are homomorphisms of Lie algebras.

Moreover, it is easy to show that
$${\rm Ker}\Theta_{+}={\rm Ker}\Theta_{-}=\{(-\alpha(a^*),a^*)|a^*\in\frak{f}^{\bot}={\rm
Ker}\beta\}.$$ By Lemma~\ref{le:bradouble}.(\mref{it:bradouble}), for any $a^*,b^*\in \frak{f}^{\bot}={\rm Ker}\beta$, we have
$$[(-\alpha(a^*),a^*),(-\alpha(b^*),b^*)]_{\mathcal{D}(\frak{g})}=(-[\beta(a^*),\beta(b^*)],0)=(0,0).$$
So ${\rm Ker}\Theta_{+}={\rm Ker}\Theta_{-}$ is an abelian Lie
subalgebra of $\mathcal{D}(\frak{g})$.
\end{proof}

Equip the space $\frak{f}^{\bot}={\rm Ker}\beta$ with the structure
of an abelian Lie algebra. Define a linear map
$\iota:\frak{f}^{\bot}\to \mathcal{D}(\frak{g})$ by
$$
\iota(a^*)=(-\alpha(a^*),a^*),\quad \forall a^*\in\frak{f}^{\bot}.
$$
Then $\iota$ is in fact an embedding of Lie algebras whose image
coincides with ${\rm Ker}\Theta_{+}={\rm Ker}\Theta_{-}$. On the
other hand, the images of $\Theta_{\pm}$ in
$\hat{\frak{g}}=\frak{g}\oplus i\frak{g}$ are $\frak{g}\oplus i{\rm
Im}\beta=\frak{g}\oplus i\frak{f}$, which is a Lie subalgebra of
$\hat{\frak{g}}$. Thus we have
\begin{prop}
The sequences
\begin{equation}
\xymatrix{ 0\ar@{->}[r]& \frak{f}^{\bot} \ar@{->}^{\iota}[r]&
\mathcal{D}(\frak{g})\ar@{->}^{\Theta_{\pm}}[r] & \frak{g}\oplus
i\frak{f} \ar@{->}[r] & 0}\label{eq:exactse}
\end{equation}
are exact.\label{pp:exactse}
\end{prop}

As a special case, we have
\begin{coro}{\rm (\cite{LQ})}\quad Let $\frak{g}$ be a Lie algebra and $r\in\frak{g}\otimes\frak{g}$.
Define $\alpha$ and $\beta$ by Eq.~(\ref{eq:alphabeta}). Suppose
that $\beta$ is invariant and invertible (regarded as a linear map from
$\frak{g}^*$ to $\frak{g}$). If $(\frak{g},r)$ is a type II
quasitriangular Lie bialgebra,
then $\Theta_{\pm}:\mathcal{D}(\frak{g})\rightarrow \frak{g}\oplus
i\frak{g}$ are isomorphisms of Lie algebras.
\end{coro}

\begin{proof}
In this case, ${\rm Ker}\Theta_{+}={\rm Ker}\Theta_{-}=0$ and ${\rm
Im}\Theta_{+}={\rm Im}\Theta_{-}=\frak{g}\oplus i\frak{g}$.
\end{proof}

\subsubsection{Description of the extension}

According to Proposition~\ref{pp:exactse}, $\mathcal{D}(\frak{g})$
is an extension of $\frak{g}\oplus i\frak{f}$ by the abelian Lie
algebra $\frak{f}^{\bot}$. So
there is an induced representation of $\frak{g}\oplus i\frak{f}$ on
$\frak{f}^{\bot}$ and the extension is uniquely defined by an
element of $H^2(\frak{g}\oplus i\frak{f},\frak{f}^{\bot})$. To
describe these structures explicitly, we need to fix two splittings
$S_{\pm}:\frak{g}\oplus i\frak{f}\to\mathcal{D}(\frak{g})$ of
Eq.~(\ref{eq:exactse}) in the category of vector spaces, that is,
$\Theta_{\pm}\circ S_{\pm}={\rm id}_{\frak{g}\oplus i\frak{f}}$ such
that $S(0)=0$. In fact, suppose that $s:\frak{f}\to \frak{g}^*$ is a
right inverse of $\beta:\frak{g}^*\to \frak{f}\subset\frak{g}$, that
is, $\beta\circ s={\rm id}_{\frak{f}}$, then the desired splittings
$S_{\pm}:\frak{g}\oplus i\frak{f}\to\mathcal{D}(\frak{g})$ are
defined by
$$
S_{\pm}(x+iy)=x\mp\alpha s(y)\pm s(y),\quad \forall
x\in\frak{g},y\in\frak{f}.
$$

Recall that the construction of a Lie algebra $\frak{h}$ by a
$\frak{h}$-module $V$ associated to a cohomology class $[\tau]\in
H^2(\frak{h},V)$ is the vector space $\frak{h}\oplus V$ equipped
with the bracket
$
[(x,u),(y,v)]=([x,y],x\cdot v-y\cdot u+\tau(x,y)),\quad \forall
x,y\in\frak{h},u,v\in V.
$
We denote such extension by $\frak{h}\ltimes_{\tau}V$.

Returning to $\mathcal{D}(\frak{g})$, we shall write down the
actions of $\frak{g}\oplus i\frak{f}$ on $\frak{f}^{\bot}$ and the
cohomology classes $\tau_{\pm}$ explicitly.

\begin{lemma}
The actions of $\frak{g}\oplus i\frak{f}$ on $\frak{f}^{\bot}$
induced from the extensions defined by Eq.~$($\ref{eq:exactse}$)$
are given by $(x+iy)\cdot_{\pm}a^*={\rm ad}^*(x)a^*$, for any
$x\in\frak{g},y\in\frak{f},a^*\in\frak{f}^{\bot}$.
\end{lemma}

\begin{proof}
According to Lemma~\mref{le:bradouble},
for any $x\in\frak{g},y\in\frak{f},a^*\in\frak{f}^{\bot}$, we have
\begin{eqnarray*}
[S_{\pm}(x+iy),\iota(a^*)]&=&[x\mp\alpha(s(y))\pm
s(y),-\alpha(a^*)+a^*]\\
&=&[x,-\alpha(a^*)+a^*]\pm[\beta(s(y)),\beta(a^*)]\\
&=&-[x,\alpha(a^*)]-{\rm ad}^*(a^*)x+{\rm ad}(x)a^*=\iota({\rm
ad}(x)a^*).
\end{eqnarray*}
So the actions are given by $(x+iy)\cdot_{\pm}
a^*=\iota^{-1}([S(x+iy),\iota(a^*)])={\rm ad}^*(x)a^*$.
\end{proof}

\begin{theorem}
Define two forms $\tau_{\pm}:(\frak{g}\oplus
i\frak{f})\otimes(\frak{g}\oplus i\frak{f})\to\frak{f}^{\bot}$ by
$$
\tau_{\pm}(x_1+iy_1,x_2+iy_2)=\pm({\rm ad}^*(x_1)s(y_2)-{\rm
ad}^*(x_2)s(y_1)-s([x_1,y_2])+s([x_2,y_1])),
$$
for any $x_1,x_2\in\frak{g},y_1,y_2\in\frak{f}^{\bot}$. Then the
forms $\tau_{\pm}$ are 2-cocycles and
\begin{equation}
\mathcal{D}(\frak{g})\cong(\frak{g}\oplus
i\frak{f})\ltimes_{\tau_{\pm}}\frak{f}^{\bot}.\label{eq:identification}
\end{equation}
\end{theorem}

\begin{proof}
The cohomology classes associated to the extensions defined by
Eq.~(\ref{eq:exactse}) are the classes of the 2-cocycles
($x_1,x_2\in\frak{g},y_1,y_2\in\frak{f}^{\bot}$)

{\small \begin{eqnarray*}
& &\iota^{-1}([S_{\pm}(x_1+iy_1),S_{\pm}(x_2+iy_2)]-S_{\pm}([x_1+iy_1,x_2+iy_2]))\\
&=&\iota^{-1}([x_1\mp\alpha(s(y_1))\pm
s(y_1),x_2\mp\alpha(s(y_2))\pm
s(y_2)]-S_{\pm}([x_1,x_2]-[y_1,y_2]+i([x_1,y_2]+[y_1,x_2]))\\
&=&\iota^{-1}([x_1,x_2]+[x_1,\pm(-\alpha(s(y_2))+s(y_2))]+[\pm(-\alpha(s(y_1))+s(y_1)),x_2]+
[-\alpha(s(y_1))+\\
&
&s(y_1),-\alpha(s(y_2))+s(y_2)]-[x_1,x_2]+[y_1,y_2]\pm\alpha(s([x_1,y_2]+[y_1,x_2]))
\mp s([x_1,y_2]+[y_1,x_2]))\\
&=&\iota^{-1}(\pm \iota({\rm ad}^*(x_1)(s(y_2)))\mp\iota({\rm
ad}^*(x_2)(s(y_1)))-[\beta(s(y_1)),\beta(s(y_2))]+[y_1,y_2]\\
&&\mp
\iota(s([x_1,y_2]+[y_1,x_2])))\\
&=&\pm({\rm ad}^*(x_1)s(y_2)-{\rm
ad}^*(x_2)s(y_1)-s([x_1,y_2])+s([x_2,y_1])),
\end{eqnarray*}}
where the third equality follows from Lemma~\ref{le:bradouble}.
\end{proof}

\subsubsection{The embeddings of $\frak{g}$ and $\frak{g}^*$ in
$\mathcal{D}(\frak{g})$ and the description of the Manin triple} We
now apply the isomorphisms in Eq.~(\mref{eq:identification}) to
describe the structure of the Manin triple
$(\cald(\frakg),\frakg,\frakg^*)$ explicitly in terms of
$(\frak{g}\oplus i\frak{f})\ltimes_{\tau_{\pm}}\frak{f}^{\bot}$.


It is clear that from the identifications defined by
Eq.~(\ref{eq:identification}), $\frak{g}$ is embedded in
$\mathcal{D}(\frak{g})$ by
\begin{equation}
x\mapsto (x,0)\ltimes_{\tau_{\pm}}0\in(\frak{g}\oplus
i\frak{f})\ltimes_{\tau_{\pm}}\frak{f}^{\bot},\quad \forall
x\in\frak{g}.\label{eq:gem}
\end{equation}
Moreover, for any $a^*\in\frak{g}^*$, we have
$a^*-s(\beta(a^*))\in\frak{f}^{\bot}$ and
$a^*=S_{\pm}(\alpha(a^*)\pm i\beta(a^*))+\iota(a^*-s(\beta(a^*)))$.
So the embeddings of $\frak{g}^*$ in $(\frak{g}\oplus
i\frak{f})\ltimes_{\tau_{\pm}}\frak{f}^{\bot}\cong\mathcal{D}(\frak{g})$
are given by
\begin{equation}
a^*\mapsto (\alpha(a^*)\pm
i\beta(a^*))\ltimes_{\tau_{\pm}}(a^*-s(\beta(a^*))).\label{eq:gstarem}
\end{equation}
To describe the embeddings of $\frak{g}^*$ more explicitly, we first
recall some results in~\cite{HY} about classification of subalgebras
of extensions of the form $\frak{h}\ltimes_{\tau}V$, where
$\frak{h}$ is a Lie algebra, $V$ is an $\frak{h}$-module and
$\tau\in H^2(\frak{h},V)$. Let $p:\frak{h}\ltimes_{\tau}V\to
\frak{h}$ and $q:\frak{h}\ltimes_{\tau}V\to V$ be the projections
$p(h,u)=h$ and $q(h,u)=u$ for any $h\in\frak{h},u\in V$.

\begin{theorem}{\rm (\cite{HY})}\quad
Let $\frak{b}$ be a Lie subalgebra of $\frak{h}$ and $W$ be a
$\frak{b}$-submodule of $V$. Let $\phi:\frak{b}\to V/W$ be a
1-cochain whose coboundary is $-\epsilon\circ \tau|_{\frak{b}}$,
where $\epsilon$ denotes the projection $V\to V/W$. Define
$$
\frak{b}_{W}^{\phi}=\{(x,u)|x\in\frak{b},u+W=\phi(x)\}.
$$
Then $\frak{b}_{W}^{\phi}$ is a Lie subalgebra of
$\frak{h}\ltimes_{\tau}V$. Conversely, if $\frak{k}$ is a Lie
subalgebra of $\frak{h}\ltimes_{\tau}V$, then $\frak{k}$ is of the
form $\frak{b}_{W}^{\phi}$, where $\frak{b}=p(\frak{k}),
W=\frak{k}\cap V$ and $\phi:\frak{b}\to V/W$ is given by
$\phi(x)=q(p^{-1}(x))+W$, for any
$x\in\frak{b}$.\label{thm:classify}
\end{theorem}

We now identify $\frakg^*$ with its embedded images inside $(\frak{g}\oplus
i\frak{f})\ltimes_{\tau_{\pm}}\frak{f}^{\bot}$. It follows from
Eq.~(\ref{eq:gstarem}) that  $W={\rm Ker}\alpha\cap{\rm Ker}\beta$
and $\frak{b}_{\pm}=\Theta_{\pm}(\frak{g}^*)=\{\alpha(a^*)\pm
i\beta(a^*)|a^*\in\frak{g}^*\}$, where $\Theta_{\pm}$ are defined by
Eq.~(\ref{eq:Theta1}). Furthermore the projections
$p_{\pm}|_{\frak{g}^*}:\frak{g}^*\to\Theta_{\pm}(\frak{g}^*)$ factor
through the isomorphisms
$\bar{p}_{\pm}:\frak{g}^*/W\to\frak{b}_{\pm}$ given by
$$
\bar{p}_{\pm}(a^*+W)=\alpha(a^*)\pm i\beta(a^*),\quad \forall
a^*\in\frak{g}^*,
$$
respectively. Hence the 1-cochains
$\phi_{\pm}:\frak{b}_{\pm}\to\bar{\frak{f}}^{\bot}=\frak{f}^{\bot}/W$
of Theorem~\ref{thm:classify} in this situation are given by
\begin{eqnarray}
\phi_{\pm}(x+iy)&=&\bar{p}_{\pm}^{-1}(x+iy)-\epsilon
s\beta\bar{p}_{\pm}^{-1}(x+iy)\notag\\
&=&\bar{p}_{\pm}^{-1}(x+iy)\mp\epsilon s(y).\label{eq:1cochains}
\end{eqnarray}

Thus we have

\begin{theorem}
The images of $\frak{g}^*$ inside $\mathcal{D}(\frak{g})$ under the
isomorphisms $\mathcal{D}(\frak{g})\cong(\frak{g}\oplus
i\frak{f})\ltimes_{\tau_{\pm}}\frak{f}^{\bot}$ coincide with the
subalgebras ${\frak{b}_{\pm}}_{W}^{\phi_{\pm}}$ respectively, where
$\frak{b}_{\pm}=\Theta_{\pm}(\frak{g}^*)$, $W={\rm
Ker}\alpha\cap{\rm Ker}\beta$ and
$\phi_{\pm}:\frak{b}_{\pm}\to\bar{\frak{f}}^{\bot}$ are described by
Eq.~$($\ref{eq:1cochains}$)$. \mlabel{thm:triple}
\end{theorem}

\begin{remark}
{\rm One can define a {\bf type II quasitriangular Poisson-Lie
group} as a simply connected Poisson-Lie group whose tangent Lie
bialgebra is a type II quasitriangular Lie bialgebra. Moreover, one
can investigate the above descriptions of the structure of
$\mathcal{D}(\frak{g})$ and the embeddings of $\frak{g}$ and
$\frak{g}^*$ in $\mathcal{D}(\frak{g})$ in the context of (type II
quasitriangular) Poisson-Lie groups. For the corresponding
discussion of quasitriangular Lie bialgebras and quasitriangular
Poisson-Lie groups, see the study in~\cite{HY}. }
\end{remark}

We end our explicit description of the Manin triple
$(\cald(\frakg),\frakg,\frakg^*)$ in terms of the isomorphisms in
Eq.~(\mref{eq:identification}) by expressing the bilinear form
$\frakB_p$ in Eq.~(\mref{eq:biform1}).
For any
$$
d=x+iy\ltimes_{\tau_{\pm}}\eta\in(\frak{g}\oplus
i\frak{f})\ltimes_{\tau_{\pm}}\frak{f}^{\bot},\quad
x\in\frak{g},y\in\frak{f},\eta\in\frak{f}^{\bot},
$$
define
$$
\Xi_{\pm}(d)\equiv x-\alpha(\eta)\mp\alpha(s(y))\in\frak{g},\quad
\Lambda_{\pm}(d)\equiv\eta\pm s(y)\in\frak{g}^*.
$$
Using Eq.~(\ref{eq:gem}) and Eq.~(\ref{eq:gstarem}), it is obvious
that the compositions of the isomorphisms $(\frak{g}\oplus
i\frak{f})\ltimes_{\tau_{\pm}}\frak{f}^{\bot}\cong\mathcal{D}(\frak{g})\cong
\frak{g}\oplus\frak{g}^*$ are given by
$d\mapsto(\Xi_{\pm}(d),\Lambda_{\pm}(d))$ respectively. Therefore,
the bilinear forms given by Eq.~(\ref{eq:biform1}) on
$(\frak{g}\oplus
i\frak{f})\ltimes_{\tau_{\pm}}\frak{f}^{\bot}\cong\mathcal{D}(\frak{g})$
satisfy
$$
\frak{B}_{\pm}(d_1,d_2)\equiv\langle\Lambda_{\pm}(d_1),\Xi_{\pm}(d_2)\rangle+
\langle\Lambda_{\pm}(d_2),\Xi_{\pm}(d_1)\rangle.
$$

\section{Self-dual Lie algebras and factorizable (type II)
quasitriangular Lie bialgebras} \label{se:self}
We will focus on \tto $\calo$-operators on self-dual Lie algebras and the related (type II) factorizable quasitriangular Lie bialgebras in this section. We first obtain finer properties of the various \tto $\calo$-operators (in Eq.~(\ref{eq:adkmcybe}) and Eq.~(\ref{eq:kmcybe})) and the ECYBE in this context. We then apply these properties to provide new examples of (type II) factorizable quasitriangular Lie bialgebras.

\subsection{\Tto $\calo$-operators and the ECYBE on self-dual Lie algebras}
\begin{defn}
{\rm Let $\frak{g}$ be a Lie algebra and $\frak{B}:\frak{g}\otimes\frak{g}\to \mathbb{R}$ be a bilinear form. Suppose that $R:\frakg\to\frakg$ is a linear endomorphism of $\frakg$. Then $R$ is called {\bf self-adjoint} (resp. {\bf skew-adjoint}) with respect to $\frak{B}$ if
$$
\frak{B}(R(x),y)=\frak{B}(x,R(y))\quad ({\rm resp}.\, \frak{B}(R(x),y)=-\frak{B}(x,R(y)))
$$
for any $x,y\in\frakg$.}
\end{defn}

\begin{lemma}
Let $\frak{g}$ be a Lie algebra and
$\frak{B}:\frak{g}\otimes\frak{g}\to \mathbb{R}$ be a nondegenerate
symmetric invariant bilinear form. Let
$\varphi:\frak{g}\to\frak{g}^*$ be defined from $\frakB$ by
Eq.~$($\ref{eq:definelinearmap}$)$. Suppose that $\beta:\frak{g}\to
\frak{g}$ is an endomorphism that is self-adjoint with respect to
$\frak{B}$. Then for a given $\kappa\in\mathbb{R}$, $\beta$ is
antisymmetric of \bwt $\kappa$ and $\frak{g}$-invariant of \bwt
$\kappa$, i.e., it satisfies Eq.~$($\ref{eq:adk}$)$, if and only if
$\tilde{\beta}=\beta\varphi^{-1}:\frak{g}^*\to \frak{g}$ is
antisymmetric of \bwt $\kappa$ and $\frak{g}$-invariant of \bwt
$\kappa$, i.e.,
\begin{equation}
\kappa\tilde{\beta}({\rm
ad}^*(x)a^*)=\kappa[x,\tilde{\beta}(a^*)],\quad \forall
x\in\frak{g},a^*\in\frak{g}^*,\label{eq:kcogequ}
\end{equation}
\begin{equation}
\kappa{\rm ad}^*(\tilde{\beta}(a^*))b^*+\kappa{\rm
ad}^*(\tilde{\beta}(b^*))a^*=0,\quad \forall
a^*,b^*\in\frak{g}^*.\label{eq:kcoanti}
\end{equation}
\label{le:frosy}
\end{lemma}
\begin{proof}
When $\kappa=0$, the conclusion is obvious. Now we assume
$\kappa\neq 0$. Since $\frak{B}$ is symmetric and $\beta$ is
self-adjoint with respect to $\frak{B}$, for any
$a^*,b^*\in\frak{g}^*$ and
$x=\varphi^{-1}(a^*),y=\varphi^{-1}(b^*)\in\frak{g}$, we have
$\langle\beta(x),\varphi(y)\rangle=\langle\varphi(x),\beta(y)\rangle$.
Hence $\langle\tilde{\beta}(a^*),b^*\rangle=\langle
a^*,\tilde{\beta}(b^*)\rangle$, that is, $\tilde{\beta}$ as an
element of $\frak{g}\otimes\frak{g}$ is symmetric. So by
Lemma~\ref{le:symmetry}, Eq.~(\ref{eq:kcogequ}) and
Eq.~(\ref{eq:kcoanti}) are equivalent. On the other hand, since
$\frak{B}$ is symmetric and invariant and $\beta$ is self-adjoint
with respect to $\frak{B}$, for any $z\in\frak{g}$, we have
\begin{eqnarray*}
\langle {\rm ad}^*(\tilde{\beta}(a^*))b^*,z\rangle&=&\langle
b^*,[z,\beta(x)]\rangle=\frak{B}(y,[z,\beta(x)])=\frak{B}([y,z],\beta(x))=\frak{B}(\beta([y,z]),x),\\
\langle{\rm
ad}^*(\tilde{\beta}(b^*))a^*,z\rangle&=&\frak{B}(x,[z,\beta(y)]).
\end{eqnarray*}
Since $\frak{B}$ is nondegenerate, ${\rm
ad}^*(\tilde{\beta}(a^*))b^*+{\rm ad}^*(\tilde{\beta}(b^*))a^*=0$ if
and only if $\beta([y,z])=[\beta(y),z]$, which is equivalent to the
fact that $\beta$ satisfies Eq.~(\ref{eq:adk}) for $k\neq 0$. So the
conclusion follows.
\end{proof}

\begin{prop}
Let $\frak{g}$ be a Lie algebra and
$\frak{B}:\frak{g}\otimes\frak{g}\to \mathbb{R}$ be a nondegenerate
symmetric invariant bilinear form. Let
$\varphi:\frak{g}\to\frak{g}^*$ be defined from $\frakB$ by
Eq.~$($\ref{eq:definelinearmap}$)$. Suppose that $R$ and $\beta$ are
two linear endomorphisms of $\frak{g}$ and $\beta$ is self-adjoint
with respect to $\frak{B}$. Let $\kappa\in\RR$ be given.
\begin{enumerate}
\item
$R$ is an \tto $\calo$-operator with \bop $\beta$ of \bwt $\kappa$,
i.e., $\beta$ satisfies Eq.~$($\ref{eq:adk}$)$ and $R$ and $\beta$
satisfy Eq.~$($\ref{eq:adkmcybe}$)$, if and only if
$\tilde{R}=R\varphi^{-1}:\frak{g}^*\to\frak{g}$ is an \tto
$\calo$-operator with \bop
$\tilde{\beta}=\beta\varphi^{-1}:\frak{g}^*\to\frak{g}$ of \bwt
$\kappa$, i.e., $\tilde{\beta}$ satisfies Eq.~(\ref{eq:kcogequ}) and
Eq.~$($\ref{eq:kcoanti}$)$ and $\tilde{R}$ and $\tilde{\beta}$
satisfy Eq.~$($\ref{eq:kmcybe}$)$ for $\alpha=\tilde{R}$ and
$\beta=\tilde{\beta}$, where the linear map
$\varphi:\frak{g}\to\frak{g}^*$ is defined by
Eq.~$($\ref{eq:definelinearmap}$)$. \mlabel{it:extequiv}
\item
Suppose in addition that $R$ is
skew-adjoint with respect to $\frak{B}$.
Then $r_{\pm}=\tilde{R}\pm\tilde{\beta}$ regarded as an element of
$\frak{g}\otimes\frak{g}$ is a solution of ECYBE of \ewt
$\frac{\kappa+1}{4}$ if and only if $R$ is an \tto $\calo$-operator
with \bop $\beta$ of \bwt $\kappa$.\mlabel{it:fk0}
\end{enumerate}
\label{pp:equivalence}
\end{prop}

\begin{proof}
(\mref{it:extequiv})
First, by Lemma~\ref{le:frosy} we know that $\beta$ is
antisymmetric of \bwt $\kappa$ and $\frak{g}$-invariant of \bwt
$\kappa$ if and only if $\tilde{\beta}=\beta\varphi^{-1}$
is antisymmetric of \bwt $\kappa$ and $\frak{g}$-invariant of
\bwt $\kappa$. On the other hand, since $\frak{B}$ is
symmetric and invariant, for any $x,y,z\in\frak{g}$, we have
\begin{equation}
\frak{B}([x,y],z)=\frak{B}(x,[y,z])\Leftrightarrow\langle\varphi([x,y]),z\rangle=\langle\varphi(x),[y,z]\rangle
\Leftrightarrow \varphi({\rm ad}(y)x)={\rm
ad}^*(y)\varphi(x).\label{eq:biinvariant}
\end{equation}
For any $x,y\in\frak{g}$, put $a^*=\varphi(x),b^*=\varphi(y)$. Since
$\varphi$ is invertible, Eq.~(\ref{eq:adkmcybe}) can be written as
$$[\tilde{R}(a^*),\tilde{R}(b^*)]-\tilde{R}(\varphi([\tilde{R}(a^*),\varphi^{-1}(b^*)]+
[\varphi^{-1}(a^*),\tilde{R}(b^*)]))=k[\tilde{\beta}(a^*),\tilde{\beta}(b^*)].$$
By Eq.~(\ref{eq:biinvariant}), the above equation is equivalent to
$$[\tilde{R}(a^*),\tilde{R}(b^*)]-\tilde{R}({\rm
ad}^*(\tilde{R}(a^*))b^*-{\rm
ad}^*(\tilde{R}(b^*))a^*)=\kappa[\tilde{\beta}(a^*),\tilde{\beta}(b^*)].$$
 So $R$ is an
\tto $\calo$-operator with \bop $\beta$ of \bwt $\kappa$ if and only
if $\tilde{R}=R\varphi^{-1}:\frak{g}^*\to\frak{g}$ is an \tto
$\calo$-operator with \bop $\tilde{\beta}$ of \bwt $\kappa$.
\medskip

\noindent (\mref{it:fk0}) Furthermore, if $R$ is skew-adjoint with
respect to $\frak{B}$, then $\langle
R(x),\varphi(y)\rangle+\langle\varphi(x),R(y)\rangle=0$. Hence
$\langle\tilde{R}(a^*),b^*\rangle+\langle
a^*,\tilde{R}(b^*)\rangle=0$, that is, $\tilde{R}$ regarded as an
element of $\frak{g}\otimes \frak{g}$ is skew-symmetric. Therefore,
the conclusion (\ref{it:fk0}) follows from Item~(\mref{it:extequiv})
and Theorem~\mref{thm:cybea}.
\end{proof}

As special cases of Proposition~\mref{pp:equivalence}.(\mref{it:fk0}), we have
\begin{coro}
Under the same assumptions as in Proposition~\mref{pp:equivalence}.(\mref{it:fk0}), we have
\begin{enumerate}
\item
If $\kappa=-1$, then $r_{\pm}=\tilde{R}\pm\tilde{\beta}$ as an
element of $\frak{g}\otimes\frak{g}$ is a solution of the CYBE
$($Eq.~$($\ref{eq:cybe}$)$$)$, namely $(\frak{g},r_{\pm})$ is a
quasitriangular Lie bialgebra, if and only if $R$ is an \tto
$\calo$-operator with \bop $\beta$ of \bwt $-1$, that is, $\beta$
satisfies Eq.~$($\ref{eq:adk}$)$ for $\kappa\neq0$ and $R$ and
$\beta$ satisfy Eq.~$($\ref{eq:adkmcybe}$)$ for $\kappa=-1$.
\mlabel{it:fk1}
\item
If $\kappa=1$, then $r_{\pm}=\tilde{R}\pm\tilde{\beta}$ as an
element of $\frak{g}\otimes\frak{g}$ is a solution of type II CYBE
$($Eq.~$($\ref{eq:type2cybe}$)$$)$, namely $(\frak{g},r_{\pm})$ is a
type II quasitriangular Lie bialgebra, if and only if $R$ is an \tto
$\calo$-operator with \bop $\beta$ of \bwt 1, that is, $\beta$
satisfies Eq.~$($\ref{eq:adk}$)$ for $\kappa\neq0$ and $R$ and
$\beta$ satisfy Eq.~$($\ref{eq:adkmcybe}$)$ for $\kappa=1$.
\mlabel{it:fk2}
\end{enumerate}
\mlabel{co:kappa}
\end{coro}

\begin{remark} {\rm
Conclusion (\ref{it:fk1}) in the above corollary in the special case when $\beta={\rm id}_{\frak{g}}$ can also be found in ~\cite{Ko}.}
\end{remark}

\subsection{Factorizable quasitriangular Lie bialgebras}

Recall that a quasitriangular Lie bialgebra $(\frak{g},r)$ is said
to be {\bf factorizable} if the symmetric part of $r$ regarded as a
linear map from $\frak{g}^*$ to $\frak{g}$ is invertible.
Factorizable quasitriangular Lie bialgebras are related to the
factorization problem in integrable systems ~\cite{RS}. Next we will
provide some new examples of factorizable quasitriangular Lie
bialgebras.

\begin{lemma}
Let $G$ be a simply connected Lie group whose Lie algebra is
$\frak{g}$. Let $N$ be a linear transformation of $\frak{g}$ which
induces a left invariant $(1,1)$ tensor field on $G$. If there
exists a left invariant torsion-free connection $\nabla$ on $G$ such
that $N$ is parallel with respect to $\nabla$, then $N$ is a
Nijenhuis tensor, that is, it satisfies
Eq.~$($\ref{eq:nijenhuis}$)$.\mlabel{le:useful}
\end{lemma}

\begin{proof}
Since $N$ is parallel with respect to $\nabla$, for any
$x,y\in\frak{g}$, we have that $N(\nabla_{\hat{x}}\hat{y}(e))=\nabla_{\hat{x}}N(y)^{\wedge}(e)$, where $\hat{x},\hat{y}$ are the left invariant vector fields
generated by $x,y\in\frak{g}$ respectively and $e$ is the identity element of
$G$.
Moreover, since $\nabla$ is  torsion-free, for any $x,y\in\frak{g}$,
we show that
\begin{eqnarray*}
[N(x),N(y)]+N^2([x,y])&=&\nabla_{N(x)^{\wedge}}N(y)^{\wedge}(e)-
\nabla_{N(y)^{\wedge}}N(x)^{\wedge}(e)+N^2(\nabla_{\hat{x}}\hat{y}(e))-
N^2(\nabla_{\hat{y}}\hat{x}(e))\\
&=&N(\nabla_{N(x)^{\wedge}}\hat{y}(e))-N(\nabla_{\hat{y}}N(x)^{\wedge}(e))+
N(\nabla_{\hat{x}}N(y)^{\wedge}(e))-N(\nabla_{N(y)^{\wedge}}\hat{x}(e))\\
&=&N([N(x),y]+[x,N(y)]).
\end{eqnarray*}
\end{proof}

\begin{lemma}
Let $(\frak{g},r)$ be a triangular Lie bialgebra, that is, $r$ is a
skew-symmetric solution of CYBE. Suppose that $r$ regarded as a
linear map from $\frak{g}^*$ to $\frak{g}$ is invertible. Define a family of
linear maps
$N_{\lambda_1,\lambda_2,\lambda_3,\lambda_4}:\mathcal{D}(\frak{g})=\frak{g}\oplus\frak{g}^*\to
\mathcal{D}(\frak{g})=\frak{g}\oplus\frak{g}^*$ by
\begin{equation}
N_{\lambda_1,\lambda_2,\lambda_3,\lambda_4}(x,a^*)=(\lambda_1
r(a^*)+\lambda_2x,\lambda_3r^{-1}(x)+\lambda_4a^*),\quad \forall
x\in\frak{g},a^*\in\frak{g}^*,\lambda_i\in\mathbb{R},i=1,2,3,4.\label{eq:lambda4}
\end{equation}
Then $N_{\lambda_1,\lambda_2,\lambda_3,\lambda_4}$ is skew-adjoint
with respect to the bilinear form $\frakB_p$ defined by
Eq.~$($\ref{eq:biform1}$)$ if and only if $\lambda_2+\lambda_4=0$.
\label{le:skewadjoint1}
\end{lemma}
The lemma is interesting on its own right since the simply connected Lie group corresponding to the Lie algebra in the lemma  is a symplectic Lie group~{\rm (\cite{CP, DM, Dr})}.
\begin{proof}
In fact, for any $x,y\in\frak{g},a^*,b^*\in\frak{g}$, we have

{\small \begin{eqnarray*} &
&\frak{B}_p(N_{\lambda_1,\lambda_2,\lambda_3,\lambda_4}(x,a^*),(y,b^*))+
\frak{B}_p((x,a^*),N_{\lambda_1,\lambda_2,\lambda_3,\lambda_4}(y,b^*))\\
&=&\frak{B}_p((\lambda_1r(a^*)+\lambda_2x,\lambda_3r^{-1}(x)+\lambda_4a^*),(y,b^*))+
\frak{B}_p((x,a^*),(\lambda_1r(b^*)+\lambda_2y,\lambda_3r^{-1}(y)+\lambda_4b^*))\\
&=&\lambda_1\langle r(a^*),b^*\rangle+\lambda_2\langle
x,b^*\rangle+\lambda_3\langle r^{-1}(x),y\rangle+\lambda_4\langle
a^*,y\rangle+\lambda_3\langle x,r^{-1}(y)\rangle +\lambda_4\langle
x,b^*\rangle+\lambda_1\langle a^*,r(b^*)\rangle\\
& &+\lambda_2\langle a^*,y\rangle =(\lambda_2+\lambda_4)(\langle
x,b^*\rangle+\langle a^*,y\rangle),
\end{eqnarray*}}
where the last equality follows from $r$ being skew-symmetric. So
the conclusion follows.
\end{proof}

\begin{lemma}
With the conditions and notations in Lemma~\ref{le:skewadjoint1},
the linear operator $N_{\lambda_1,\lambda_2,\lambda_3,\lambda_4}$
defined by Eq.~$($\ref{eq:lambda4}$)$ is a Nijenhuis tensor on
$\mathcal{D}(\frak{g})$,  that is, it satisfies
Eq.~$($\ref{eq:nijenhuis}$)$ on
$\mathcal{D}(\frak{g})$.\label{le:4nijenhuis}
\end{lemma}
\begin{proof}
Let $\mathcal{D}(G)$ be the corresponding simply connected double
Lie group of the Drinfeld's double $\cald(g)$, where $G$ denotes the
simply connected Poisson-Lie group of the Lie bialgebra
$(\frak{g},r)$. Then it is easy to see that the following equation
defines a left invariant torsion-free connection (in fact, according
to~\cite{DM}, it is also flat) on $\mathcal{D}(G)$:
$$
\nabla_{(x,a^*)^{\wedge}}(y,b^*)^{\wedge}(e)=(r({\rm
ad}^*(x)r^{-1}(y))+{\rm ad}^*(a^*)y,{\rm ad}^*(r(a^*))b^*+{\rm
ad}^*(x)b^*),\quad \forall x,y\in\frak{g},a^*,b^*\in\frak{g}^*,
$$
where $(x,a^*)^{\wedge},(y,b^*)^{\wedge}$ are the left invariant vector fields
generated by $(x,a^*),(y,b^*)\in\cald(g)$ respectively and $e$ is the identity element of $\mathcal{D}(G)$. We
only need to prove that the tensor
$N_{\lambda_1,\lambda_2,\lambda_3,\lambda_4}$ defined by
Eq.~(\ref{eq:lambda4}) is parallel with respect to the above
connection, since then
$N_{\lambda_1,\lambda_2,\lambda_3,\lambda_4}$ satisfies
Eq.~(\ref{eq:nijenhuis}) on $\mathcal{D}(\frak{g})$ by Lemma~\mref{le:useful}. Now by
Lemma~\ref{le:bradelta}, Corollary~\ref{co:remark} and
Lemma~\ref{le:bradouble}.(\ref{it:dualbraex}), for any $a^*,b^*\in\frak{g}^*$,
\begin{equation}
{\rm ad}^*(a^*)r(b^*)=-[r(b^*),r(a^*)]+r({\rm
ad}^*(r(b^*))a^*)=-r([b^*,a^*]_{\delta})+r({\rm
ad}^*(r(b^*))a^*)=r({\rm ad}^*(r(a^*))b^*).\label{eq:skewuseful}
\end{equation}
Moreover, for any $x,y\in\frak{g}$,
\begin{eqnarray*}
&
&\nabla_{(x,a^*)^{\wedge}}
N_{\lambda_1,\lambda_2,\lambda_3,\lambda_4}(y,b^*)^{\wedge}(e)=
\nabla_{(x,a^*)^{\wedge}}(\lambda_1 r(b^*)+\lambda_2
y,\lambda_3r^{-1}(y)+\lambda_4 b^*)^{\wedge}(e)\\&=&(\lambda_1 r({\rm
ad}^*(x)b^*)+\lambda_2 r({\rm ad}^*(x) r^{-1}(y))+\lambda_1{\rm
ad}^*(a^*)r(b^*)+\lambda_2{\rm ad}^*(a^*)y, \lambda_3{\rm
ad}^*(r(a^*))r^{-1}(y)+\\ & &\lambda_4{\rm
ad}^*(r(a^*))b^*+\lambda_3{\rm ad}^*(x)r^{-1}(y)+\lambda_4{\rm
ad}^*(x)b^*),\\
&
&N_{\lambda_1,\lambda_2,\lambda_3,\lambda_4}(\nabla_{(x,a^*)^{\wedge}}(y,b^*)^{\wedge}(e))\\
&&= N_{\lambda_1,\lambda_2,\lambda_3,\lambda_4}(r({\rm
ad}^*(x)r^{-1}(y))+{\rm ad}^*(a^*)y,{\rm ad}^*(r(a^*))b^*+{\rm
ad}^*(x)b^*)\\
&&=(\lambda_1r({\rm ad}^*(r(a^*))b^*)+\lambda_1r({\rm
ad}^*(x)b^*)+\lambda_2r({\rm ad}^*(x)r^{-1}(y))+\lambda_2{\rm
ad}^*(a^*)y, \lambda_3{\rm ad}^*(x)r^{-1}(y)+\\ &
&\mbox{}\hspace{0.5cm}\lambda_3 r^{-1}({\rm
ad}^*(a^*)y)+\lambda_4{\rm ad}^*(r(a^*))b^*+\lambda_4{\rm
ad}^*(x)b^*).
\end{eqnarray*}
Therefore by Eq.~(\ref{eq:skewuseful}), we get
$$\nabla_{(x,a^*)^{\wedge}}N_{\lambda_1,\lambda_2,\lambda_3,\lambda_4}(y,b^*)^{\wedge}(e)=
N_{\lambda_1,\lambda_2,\lambda_3,\lambda_4}(\nabla_{(x,a^*)^{\wedge}}(y,b^*)^{\wedge}(e)).$$
Thus, $N_{\lambda_1,\lambda_2,\lambda_3,\lambda_4}$ is parallel with respect to $\nabla$, as needed.
\end{proof}

\begin{prop}
Let $(\frak{g},r)$ be a triangular Lie bialgebra. Let $\frak{B}_p$
be the bilinear form on $\cald(g)=\frakg\oplus\frakg^*$ given by
Eq.~(\ref{eq:biform1}) and let
$\varphi:\mathcal{D}(\frak{g})=\frak{g}\oplus\frak{g}^*\rightarrow\mathcal{D}(\frak{g})^*=
\frak{g}\oplus\frak{g}^*$ be the linear map induced by $\frakB_p$
through Eq.~$($\ref{eq:definelinearmap}$)$ for
$\frak{B}=\frak{B}_p$. Define a family of linear endomorphisms of
$\mathcal{D}(\frak{g})$ by
$$
R_{\mu}(x,a^*)\equiv(\mu r(a^*)+x,-a^*),\quad \forall
x\in\frak{g},a^*\in\frak{g}^*,\mu\in\mathbb{R}.
$$
Define $\tilde{r}_{\pm,\mu}\equiv
R_{\mu}\varphi^{-1}\pm\varphi^{-1}$ and regard $\tilde{r}_{\pm,\mu}$
as elements of $\mathcal{D}(\frak{g})\otimes\mathcal{D}(\frak{g})$.
Then $(\mathcal{D}(\frak{g}),\tilde{r}_{\pm,\mu})$ are factorizable
quasitriangular Lie bialgebras.\label{pp:factorizable}
\end{prop}

\begin{proof}
First we prove that, for any $\mu\in\RR$, $R_{\mu}$ is an \tto
$\calo$-operator with \bop ${\rm
id}:\mathcal{D}(\frak{g})\to\mathcal{D}(\frak{g})$ of \bwt $-1$,
that is, it satisfies Eq.~(\ref{eq:-1mcyb}) on
$\mathcal{D}(\frak{g})$. Recall the Lie algebra structure of
$\mathcal{D}(\frak{g})$ is given by Eq.~(\mref{eq:drindouliestr}).
Then, for any $x,y\in\frak{g},a^*,b^*\in\frak{g}^*$, we have
\begin{eqnarray*}
& &[R_{\mu}(x,a^*),R_{\mu}(y,b^*)]_{\mathcal{D}(\frak{g})}=[(\mu
r(a^*)+x,-a^*),(\mu
r(b^*)+y,-b^*)]_{\mathcal{D}(\frak{g})}\\
&=&([\mu r(a^*)+x,\mu r(b^*)+y]+{\rm ad}^*(-a^*)(\mu r(b^*)+y)-{\rm
ad}^*(-b^*)(\mu r(a^*)+x),[a^*,b^*]_{\delta}-\\ & &{\rm ad}^*(\mu
r(a^*)+x)b^*+{\rm ad}^*(\mu r(b^*)+y)a^*).
\end{eqnarray*}
On the other hand,
\begin{eqnarray*}
[(x,a^*),(y,b^*)]_{\mathcal{D}(\frak{g})}&=&([x,y]+{\rm
ad}^*(a^*)y-{\rm
ad}^*(b^*)x,[a^*,b^*]_{\delta}+{\rm ad}^*(x)b^*-{\rm ad}^*(y)a^*)\\
R_{\mu}([R_{\mu}(x,a^*),(y,b^*)]_{\mathcal{D}(\frak{g})})&=&(-\mu
r([a^*,b^*]_{\delta})+\mu^2 r({\rm ad}^*(r(a^*))b^*)+\mu r({\rm
ad}^*(x)b^*)+\mu r({\rm ad}^*(y)a^*)\\ & &+\mu[r(a^*),y]+[x,y]-{\rm
ad}^*(a^*)y-\mu{\rm ad}^*(b^*)r(a^*)-{\rm ad}^*(b^*)x,\\ &
&[a^*,b^*]_{\delta}-\mu{\rm
ad}^*(r(a^*))b^*-{\rm ad}^*(x)b^*-{\rm ad}^*(y)a^*)\\
R_{\mu}([(x,a^*),R_{\mu}(y,b^*)]_{\mathcal{D}(\frak{g})})&=&(-\mu
r([a^*,b^*]_{\delta})-\mu r({\rm ad}^*(x)b^*)-\mu^2 r({\rm
ad}^*(r(b^*))a^*)-\mu r({\rm ad}^*(y)a^*)\\ &
&+\mu[x,r(b^*)]+[x,y]+\mu{\rm
ad}^*(a^*)(r(b^*))+{\rm ad}^*(a^*)y+{\rm ad}^*(b^*)x, \\
& &[a^*,b^*]_{\delta}+{\rm ad}^*(x)b^*+\mu{\rm ad}^*(r(b^*))a^*+{\rm
ad}^*(y)a^*).
\end{eqnarray*}
Therefore, by the fact that $r$ is a homomorphism of Lie algebras
(see Corollary~\ref{co:remark}), we get
$$[R_{\mu}(x,a^*),R_{\mu}(y,b^*)]_{\mathcal{D}(\frak{g})}+[(x,a^*),(y,b^*)]_{\mathcal{D}(\frak{g})}=
R_{\mu}([R_{\mu}(x,a^*),(y,b^*)]_{\mathcal{D}(\frak{g})})+R_{\mu}([(x,a^*),R_{\mu}(y,b^*)]_{\mathcal{D}(\frak{g})}).$$
On the other hand, from the proof of Lemma~\ref{le:skewadjoint1}, we
know that $R_{\mu}$ is skew-adjoint with respect to the
nondegenerate symmetric
 invariant bilinear form $\frak{B}_p$. So the conclusion follows from
Corollary~\ref{co:kappa}.(\ref{it:fk1}) by setting
$\frak{g}=\mathcal{D}(\frak{g})$, $R=R_{\mu}$, $\beta={\rm
id}_{\mathcal{D}(\frak{g})}$ and $\frak{B}=\frak{B}_p$.
\end{proof}

Note that when $\mu=0$, then Proposition~\mref{pp:factorizable} gives a special case of the
famous ``Drinfeld's double construction"~\cite{Ko} (in the original construction there is no restriction that $\frak{g}$ is triangular, or even
coboundary).

\begin{prop}
Let $(\frak{g},r)$ be a triangular Lie bialgebra such that $r$
regarded as a linear map from $\frak{g}^*$ to $\frak{g}$ is
invertible. Define two families of linear endomorphisms on
$\mathcal{D}(\frak{g})$ by
\begin{eqnarray*} & & N_{\mu}(x,a^*)=(
x,\mu r^{-1}(x)-a^*),\quad \mu\in\mathbb R;\\
& &N_{\kappa_1,\kappa_2}(x,a^*)=(\kappa_1
r(a^*)+\kappa_2x,{{1-\kappa_2^2}\over{\kappa_1}} r^{-1}(x)-\kappa_2
a^*),\quad \kappa_1,\kappa_2\in\mathbb R,\kappa_2^2\neq 1,
\kappa_1\neq 0,
\end{eqnarray*}
for any $x\in\frak{g},a^*\in\frak{g}^*$. Let
$\varphi:\mathcal{D}(\frak{g})=\frak{g}\oplus\frak{g}^*\rightarrow\mathcal{D}(\frak{g})^*=
\frak{g}\oplus\frak{g}^*$ be the linear map induced by the bilinear
form $\frak{B}_p$ given by Eq.~$($\ref{eq:biform1}$)$ through
Eq.~$($\ref{eq:definelinearmap}$)$ for $\frak{B}=\frak{B}_p$. Define
$\tilde{N}_{\pm,\mu}\equiv N_{\mu}\varphi^{-1}\pm\varphi^{-1}$,
$\tilde{N}_{\pm,\kappa_1,\kappa_2}\equiv
N_{\kappa_1,\kappa_2}\varphi^{-1}\pm\varphi^{-1}$ and regard
$\tilde{N}_{\pm,\mu}$ and $\tilde{N}_{\pm,\kappa_1,\kappa_2}$ as
elements of $\mathcal{D}(\frak{g})\otimes\mathcal{D}(\frak{g})$.
Then $(\mathcal{D}(\frak{g}),\tilde{N}_{\pm,\mu})$ and
$(\mathcal{D}(\frak{g}),\tilde{N}_{\pm,\kappa_1,\kappa_2})$ are
factorizable quasitriangular Lie bialgebras. \label{pp:inverfactor}
\end{prop}

\begin{proof}
In fact, according to Lemma~\ref{le:4nijenhuis}, $N_{\mu}$ and
$N_{\kappa_1,\kappa_2}$ satisfy Eq.~(\ref{eq:nijenhuis}) on
$\mathcal{D}(\frak{g})$. Moreover, it is straightforward to check
that $N_{\mu}^2={\rm id}$ and $N_{\kappa_1,\kappa_2}^2={\rm id}$. So
both of them satisfy Eq.~(\ref{eq:-1mcyb}) on
$\mathcal{D}(\frak{g})$. On the other hand, by
Lemma~\ref{le:skewadjoint1}, they are skew-adjoint with respect to
the  nondegenerate symmetric invariant bilinear form $\frak{B}_p$.
So the conclusion follows from Corollary~\ref{co:kappa}.(\ref{it:fk1}) by
setting $\frak{g}=\mathcal{D}(\frak{g})$, $R=N_{\mu}$ or
$N_{\kappa_1,\kappa_2}$, $\beta={\rm id}_{\mathcal{D}(\frak{g})}$
and $\frak{B}=\frak{B}_p$.
\end{proof}

\subsection{Factorizable type II quasitriangular Lie bialgebras}
\label{ss:factorizable2}

We now consider the ``factorizable" case of type II quasitriangular Lie bialgebras.

\begin{defn}
{\rm A type II quasitriangular Lie bialgebra $(\frak{g},r)$ is called {\bf factorizable} if the symmetric part $\beta$ of $r$
regarded as a linear map from $\frak{g}^*$ to $\frak{g}$ is
invertible.}\label{de:factor2}
\end{defn}

The following conclusion is the type II analogue of the ``factorizable" property of quasitriangular Lie bialgebras~\cite{RS}.

\begin{prop}
Let $(\frak{g},r)$ be a factorizable type II quasitriangular Lie
bialgebra. Put $\tilde{r}=\alpha+i\beta:\frakg\oplus
i\frakg\to\frakg\oplus i\frakg$, where $\alpha$ and $\beta$ are
defined by Eq.~$($\ref{eq:alphabeta}$)$. Then any element
$x\in\frak{g}$ admits a unique decomposition:
$$
x=x_{+}+x_{-},
$$
with $(x_{+},x_{-})\in{\rm
Im}(\tilde{r}\oplus\tilde{r}^t)\subset\frak{g}\oplus i\frak{g}$,
where $\tilde{r}$ and $\tilde{r}^t$ are restricted to linear maps from
$i\frak{g}^*\subset \frakg\oplus i\frakg$ to $\frak{g}\oplus i\frak{g}$.
\end{prop}

\begin{proof}
Since $\tilde{r}+\tilde{r}^t=2i\beta$ and
$\beta:\frak{g}^*\to\frak{g}$ is invertible, we have
$$x=\tilde{r}(\frac{\beta^{-1}(x)}{2i})+\tilde{r}^t(\frac{\beta^{-1}(x)}{2i})\in{\rm
Im}(\tilde{r}\oplus\tilde{r}^t)\subset\frak{g}\oplus
i\frak{g},\;\;\forall x\in \frak g.$$ On the other hand, if there
exist $a^*,b^*\in \frak{g}^*$ such that
$x=\tilde{r}(ia^*)+\tilde{r}^t(ia^*)=\tilde{r}(ib^*)+\tilde{r}^t(ib^*)$.
Then
$0=\tilde{r}(ia^*-ib^*)+\tilde{r}^t(ia^*-ib^*)=-2\beta(a^*-b^*)$.
Since $\beta:\frak{g}^*\to\frak{g}$ is invertible, we obtain
$a^*=b^*$. So the conclusion follows.
\end{proof}

The following result provides a class of factorizable type II
quasitriangular Lie bialgebras (hence a new class of (coboundary)
Lie bialgebras).

\begin{prop}
Let $(\frak{g},r)$ be a triangular Lie bialgebra such that $r$
regarded as a linear map from $\frak{g}^*$ to $\frak{g}$ is
invertible. Let $\frak{B}_p$ be the bilinear form on
$\cald(g)=\frakg\oplus\frakg^*$ given by Eq.~$($\ref{eq:biform1}$)$
and let
$\varphi:\mathcal{D}(\frak{g})=\frak{g}\oplus\frak{g}^*\rightarrow\mathcal{D}(\frak{g})^*=
\frak{g}\oplus\frak{g}^*$ be the linear map induced by $\frakB_p$
through Eq.~$($\ref{eq:definelinearmap}$)$ for
$\frak{B}=\frak{B}_p$. Define a family of linear endomorphisms on
$\mathcal{D}(\frak{g})$ by
$$
J_{\lambda,\mu}(x,a^*)=(\lambda r(a^*)+\mu x,
{{-1-\mu^2}\over{\lambda}}r^{-1}(x)-\mu a^*),\quad
\lambda,\mu\in\mathbb R,\lambda\neq 0.
$$
Set $\tilde{r}_{\pm,\lambda,\mu}\equiv
J_{\lambda,\mu}\varphi^{-1}\pm\varphi^{-1}$ and regard
$\tilde{r}_{\pm,\lambda,\mu}$ as elements of
$\cald(\frak{g})\otimes\cald(\frak{g})$. Then
$(\cald(\frak{g}),\tilde{r}_{\pm,\lambda,\mu})$ are factorizable
type II quasitriangular Lie bialgebras. \label{pp:complexfactor}
\end{prop}
\begin{proof}
In fact, according to Lemma~\ref{le:4nijenhuis}, for any
$\lambda,\mu\in\mathbb{R}$, $J_{\lambda,\mu}$ satisfies
Eq.~(\ref{eq:nijenhuis}) on $\mathcal{D}(\frak{g})$. Moreover, it is
straightforward to check that $J_{\lambda,\mu}^2=-{\rm id}$. So
$J_{\lambda,\mu}$ satisfy Eq.~(\ref{eq:kmcyb}) for $\kappa=1$ on
$\mathcal{D}(\frak{g})$. On the other hand, by
Lemma~\ref{le:skewadjoint1}, $J_{\lambda,\mu}$ is skew-adjoint with
respect to the  nondegenerate symmetric invariant bilinear form
$\frak{B}_p$. So the conclusion follows from
Corollary~\ref{co:kappa}.(\ref{it:fk2}) by setting
$\frak{g}=\mathcal{D}(\frak{g})$, $R=J_{\lambda,\mu}$,  $\beta={\rm
id}_{\mathcal{D}(\frak{g})}$ and $\frak{B}=\frak{B}_p$.
\end{proof}
\begin{remark}
{\rm
\begin{enumerate}
\item
A linear transformation on a Lie algebra $\frak{g}$ satisfying
Eq.~(\ref{eq:nijenhuis}) and $J^2=-{\rm id}$ is called a {\bf complex structure} on $\frak{g}$. Suppose a
Lie algebra is self-dual with respect to a nondegenerate symmetric invariant bilinear form. According to
Corollary~\ref{co:kappa}.(\ref{it:fk2}), a complex structure on this
Lie algebra that is self adjoint with respect to the bilinear form gives rise to a coboundary Lie bialgebra structure on this Lie algebra.
This idea was pursued further in~\cite{LQ} in the study of
Poisson-Lie groups.
\item
The complex structure $J_{-1,0}$ has already been found in
~\cite{DM}.
\end{enumerate}
}
\end{remark}

\section{$\calo$-operators, PostLie algebras and dendriform trialgebras}\label{sec:postlie}
In this section, we reveal a PostLie algebra structure underneath the $\calo$-operators. We then show that there is a close relationship between PostLie algebras and dendriform trialgebras of Loday and Ronco~\cite{LR} in parallel to the relationship~\cite{C} between Pre-Lie algebras and dendriform bialgebras.

\subsection{$\calo$-operators and PostLie
algebras} We begin with recalling the concept of a PostLie algebra from an operad study~\cite{Va}.
\begin{defn}{\rm (\cite{Va})}\quad
 {\rm A {\bf (left) PostLie algebra} is a
$\mathbb{R}$-vector space $L$ with two bilinear operations $\circ$
and $[,]$ which satisfy the relations:
\begin{equation}
[x,y]=-[y,x],\label{eq:polie1}
\end{equation}
\begin{equation}
[[x,y],z]+[[z,x],y]+[[y,z],x]=0,\label{eq:polie2}
\end{equation}
\begin{equation}
z\circ(y\circ x)-y\circ(z\circ x)+(y\circ z)\circ x-(z\circ y)\circ
x+[y,z]\circ x=0,\label{eq:polie3}
\end{equation}
\begin{equation}
z\circ[x,y]-[z\circ x,y]-[x,z\circ y]=0,\label{eq:polie4}
\end{equation}
for all $x,y\in L$. Eq.~(\ref{eq:polie1}) and Eq.~(\ref{eq:polie2})
mean that $L$ is a Lie algebra for the bracket $[,]$, and we denote
it by $(\mathfrak{G}(L),[,])$. Moreover, we say that {\bf $(L,[,],\circ)$ is a PostLie algebra structure on
$(\mathfrak{G}(L),[,])$}. On the other hand, it is straightforward
to check that $L$ is also a Lie algebra for the operation:
\begin{equation}
\{x,y\}\equiv x\circ y-y\circ x+[x,y],\;\;\forall x,y\in L.\mlabel{eq:postlielie}
\end{equation}
 We shall denote it by
$(\mathcal{G}(L),\{,\})$ and say that {\bf $(\mathcal{G}(L),\{,\})$
has a compatible PostLie algebra structure given by
$(L,[,],\circ)$}. A {\bf homomorphism between two PostLie algebras}
is defined as a linear map between the two PostLie algebras that
preserves the corresponding operations.}\label{de:postlie}
\end{defn}

\begin{remark}
{\rm
\begin{enumerate}
\item
 The notion of PostLie
algebra was introduced in~\cite{Va} (in its ``right version"), where
it is pointed out that {\bf PostLie}, the operad of PostLie
algebras, is the Koszul dual of {\bf ComTrias}, the operad of {\bf commutative trialgebras}.
\item
If the bracket $[,]$ in the definition of PostLie algebra happens to
be trivial, then a PostLie algebra is a {\bf pre-Lie algebra}
~\cite{Bu}.
\end{enumerate}
}
\end{remark}

\begin{lemma}
Let $(L,[,],\circ)$ be a PostLie algebra. Define
$\rho:L\to\frak{gl}(L)$ by $\rho(x)y=x\circ y$ for any $x,y\in L$.
Then $(\frak{G}(L),[\,,\,],\rho)$ is a $(\calg(L),\{\,,\,\})$-Lie
algebra.

\mlabel{lem:postglie}
\end{lemma}

\begin{proof}
By Eq.~(\ref{eq:polie3}), $\rho$ is a representation of $(\calg(L),\{\,,\,\})$. Then by Eq.~(\ref{eq:polie4}),  $\rho$ is a Lie algebra homomorphism from
$(\mathcal{G}(L),\{,\})$ to ${\rm Der}_{\mathbb{R}}(\mathfrak{G}(L))$.
\end{proof}

\begin{theorem}
 Let $\mathfrak{g}$ be a Lie algebra and $(\frak{k},\pi)$ be a $\frak{g}$-Lie algebra. Let
 $r:\frak{k}\to\frak{g}$ be
  an $\calo$-operator of
weight $\lambda$.
\begin{enumerate}
\item
The following operations define a PostLie algebra structure on the
underlying vector space of $\mathfrak{k}$:
\begin{equation}
[x,y]\equiv\lambda[x,y]_{\frak{k}},\quad x\circ y\equiv r(x)\cdot
y,\quad x,y\in\mathfrak{k},\mlabel{eq:bracirc}
\end{equation}
where $[,]_{\frak{k}}$ is the original Lie bracket of
$\mathfrak{k}$.\label{it:depostlie1}
\item
$r$ is a Lie algebra homomorphism from $\mathcal{G}(\mathfrak{k})$
to $\mathfrak{g}$, where $\mathfrak{k}$ is taken as a PostLie
algebra with the operations $([,],\circ)$ defined in
Eq.~$($\mref{eq:bracirc}$)$.\label{it:depostlie2}
\item
If ${\rm Ker}(r)$ is an ideal of $(\frak{k},[,]_{\frak{k}})$, then
there exists an induced PostLie algebra structure on $r(\frak{k})$
given by
\begin{equation}
[r(x),r(y)]_r\equiv\lambda r([x,y]_{\frak{k}}),\quad r(x)\circ_r
r(y)\equiv r(r(x)\cdot y),\quad \forall
x,y\in\frak{k}.\label{eq:trpolie}
\end{equation}
Further, $r$ is a homomorphism of PostLie
algebras.\label{it:depostlie3}
\end{enumerate}
\label{thm:depostlie}
\end{theorem}

\begin{proof}
(\ref{it:depostlie1}) Since $\frak{k}$ is a Lie algebra,
Eq.~(\ref{eq:polie1}) and Eq.~(\ref{eq:polie2}) obviously hold.
Furthermore, for any $x,y,z\in\frak{k}$, we have
\begin{eqnarray*}
& &z\circ(y\circ x)-y\circ(z\circ x)+(y\circ z)\circ x-(z\circ
y)\circ x+[y,z]\circ x\\
&=&r(z)\cdot(r(y)\cdot x)-r(y)\cdot(r(z)\cdot x)+r(r(y)\cdot z)\cdot
x-r(r(z)\cdot y)\cdot x+\lambda r([y,z]_{\frak{k}})\cdot x\\
&=&([r(z),r(y)]_{\frak{g}}-r(r(z)\cdot y-r(y)\cdot
z+\lambda[z,y]_{\frak{k}}))\cdot x=0
\end{eqnarray*}
So Eq.~(\ref{eq:polie3}) holds. Similarly, Eq.~(\ref{eq:polie4})
holds, too.

\medskip

\noindent (\ref{it:depostlie2}) By Definition~\ref{de:postlie}, for any
$x,y\in\frak{k}$ we have
$$r(\{x,y\})=r(x\circ y-y\circ x+[x,y])=r(r(x)\cdot
y-r(y)\cdot x+\lambda [x,y]_{\frak{k}})=[r(x),r(y)]_{\frak{g}}.$$

\medskip

\noindent (\ref{it:depostlie3}) We first prove that the multiplications given
by Eq.~(\ref{eq:trpolie}) are well-defined. In fact, let
$x_1,y_1,x_2,y_2\in\frak{k}$ such that $r(x_1)=r(x_2)$ and
$r(y_1)=r(y_2)$. Since $x_1-x_2,y_1-y_2\in{\rm Ker}(r)$ and ${\rm
Ker}(r)$ is an ideal of $(\frak{k},[,]_{\frak{k}})$, we have
\begin{eqnarray*}
r(x_1)\circ_rr(y_1)&=&r(r(x_1)\cdot
y_1)=r(r(x_2+(x_1-x_2))\cdot(y_2+(y_1-y_2)))\\
&=&r(r(x_2)\cdot y_2+r(x_2)\cdot(y_1-y_2))\\
&=&r(r(x_2)\cdot
y_2)+[r(x_2),r(y_1-y_2)]_{\frak{g}}+r(r(y_1-y_2)\cdot x_2)-\lambda
r([x_2,y_1-y_2]_{\frak{k}})\\
&=&r(r(x_2)\cdot y_2)=r(x_2)\circ_rr(y_2).
\end{eqnarray*}
Also, $[r(x_1),r(y_1)]_r=[r(x_2)+r(x_1-x_2), r(y_1)+r(y_1-y_2)]_r=[r(x_2),r(y_2)]_r$.
Furthermore, we have $r([x,y])=[r(x),r(y)]_{r}$ and $r(x\circ y)=r(x)\circ_rr(y)$ for any $x,y\in\frak{k}$. Thus, $(r(\frak{k}),[,]_r$, $\circ_r)$ is a PostLie algebra
since  applying $r$ to the PostLie algebra axioms of $(\frak{k},[,],\circ)$ gives the PostLie algebra axioms of $(r(\frak{k}),[,]_r,\circ_r)$. Finally, the last statement in Item (\ref{it:depostlie3}) is clear.
\end{proof}

\begin{coro}
  Let $\frak{g}$ be a Lie algebra. Then there is a compatible PostLie algebra
  structure on $\frak{g}$ if and only if there exists a $\frak{g}$-Lie algebra
  $(\frak{k},\pi)$ and an invertible $\calo$-operator $r:\frak{k}\to
\frak{g}$ of weight $1$.\label{co:ifonlyif}
\end{coro}

\begin{proof}
Suppose that $\frak{g}$ has a compatible PostLie algebra structure
given by $(L,[,],\circ)$, that is, $\mathcal{G}(L)=\frak{g}$.  By
Lemma~\mref{lem:postglie}, $(\frak{G}(L),\rho,[,])$ is a
$\frak{g}$-Lie algebra, where $\rho:L\to\frak{gl}(L)$ is defined as
$\rho(x)y=x\circ y$ for any $x,y\in L$. Moreover, the equation
$\{x,y\}=x\circ y-y\circ x+[x,y]$ means that ${\rm
id}:\frak{G}(L)\to\mathcal{G}(L)=\frak{g}$ is an $\calo$-operator of
weight $1$. Furthermore, ${\rm id}$ is obviously invertible.

Conversely, suppose that $(\frak{k},\pi)$ is a $\frak{g}$-Lie algebra and
$r:\frak{k}\to \frak{g}$ is an invertible $\calo$-operator weight 1. Since ${\rm Ker}(r)=\{0\}$, by
Theorem~\ref{thm:depostlie}, there is a PostLie algebra structure on
$r(\frak{k})=\frak{g}$ given by Eq.~(\ref{eq:trpolie}) for
$\lambda=1$. Moreover, it is obvious that
$(r(\frak{k})=\frak{g},[,]_r,\circ_r)$ (for $\lambda=1$) is a
compatible PostLie algebra structure on $(\frak{g},[,]_{\frak{g}})$.
\end{proof}

\begin{coro}
Let $\frak{g}$ be a Lie algebra and $R:\frak{g}\to\frak{g}$ be a
Rota-Baxter operator of weight $\lambda\in\mathbb{R}$, that is, it
satisfies Eq.~$($\ref{eq:rotabaxter}$)$. Then there is a PostLie
algebra structure on $\frak{g}$ given by
\begin{equation}
[x,y]\equiv\lambda[x,y]_{\frak{g}},\quad x\circ
y\equiv[R(x),y]_{\frak{g}},\quad \forall
x,y\in\frak{g}.\label{eq:rotaconpost}
\end{equation}
If in addition, $R$ is invertible, then there is a compatible
PostLie algebra structure on $\frak{g}$ given by
$$
[x,y]\equiv\lambda R([R^{-1}(x),R^{-1}(y)]_{\frak{g}}),\quad x\circ
y\equiv R([x,R^{-1}(y)]_{\frak{g}}),\quad \forall x,y\in\frak{g}.
$$
\label{co:rotapo}
\end{coro}
\begin{proof}
The conclusion follows from Theorem~\ref{thm:depostlie}.
\end{proof}

We next give examples of PostLie algebras by applying Corollary~\mref{co:rotapo}.

\begin{exam}
{\rm Let $\frak{g}$ be a complex simple Lie algebra, $\frak{h}$ be
its Cartan subalgebra, $\Delta$ be its root system and
$\Delta_{+}\subset\Delta$ be the set of positive roots (with respect
to some fixed order). For $\alpha\in\Delta$, let
$\frak{g}_{\alpha}\subset\frak{g}$ be the corresponding root space.
Put
$\frak{n}_{\pm}=\oplus_{\alpha\in\Delta_{+}}\frak{g}_{\pm\alpha}$,
$\frak{b}_{\pm}=\frak{h}+\frak{n}_{\pm}$. Then we have
$\frak{g}=\frak{b}_{+}+\frak{n}_{-}$ as decomposition of two
subalgebras. Let $P_{\frak{b_{+}}}:\frakg \to \frakb_+
\hookrightarrow \frakg$ and $P_{\frak{n}_{-}}:\frakg \to \frakn_-
\hookrightarrow \frakg$ be the projections onto the subalgebras
$\frak{b}_{+}$ and $\frak{n}_{-}$ respectively. Then
$-P_{\frak{b}_{+}}$ and $-P_{\frak{n}_{-}}$ are Rota-Baxter
operators of weight 1. Define new operations on $\frak{g}$ as
follows:
\begin{equation}
[x,y]\equiv[x,y]_{\frak{g}},\quad
x\circ_{\frak{b}_{+}}y\equiv-[P_{\frak{b}_{+}}(x),y]_{\frak{g}},\quad
\forall x,y\in\frak{g}.\label{eq:poop}
\end{equation}
By Corollary~\ref{co:rotapo}, $([,],\circ_{\frak{b}_{+}})$ defines a PostLie algebra structure on
$\frak{g}$. If
$$
\{H_i\}_{i=1,...,n}\cup\{X_{\alpha}\}_{\alpha\in\Delta_{+}}\cup
\{X_{-\alpha}\}_{\alpha\in\Delta_{+}}
$$
is a basis of $\frak{g}$, then the PostLie operations defined by
Eq.~(\ref{eq:poop}) can be computed as follows:
\begin{eqnarray*}
& &[x,y]=[x,y]_{\frak{g}},\quad
X_{-\alpha}\circ_{\frak{b}_{+}}y=0,\quad
H_i\circ_{\frak{b}_{+}}H_j=0, \quad H_i\circ_{\frak{b}_{+}}X_{\beta}=-\langle\beta,\alpha_i\rangle
X_{\beta},\\
& &
X_{\alpha}\circ_{\frak{b}_{+}}H_i=\langle\alpha,\alpha_i\rangle
X_{\alpha},\quad
X_{\alpha}\circ_{\frak{b}_{+}}X_{\beta}=-N_{\alpha,\beta}X_{\alpha+\beta}, \quad \forall x,
y\in\frak{g}, \ \alpha\in\Delta_{+},\beta\in\Delta.
\end{eqnarray*}
Similarly, with the same bracket $[\,,\,]$ and with
$x\circ_{\frak{n}_{-}}y\equiv-[P_{\frak{n}_{-}}(x),y]_{\frak{g}}$, we obtain another PostLie algebra structure $([,],\circ_{\frak{n}_{-}})$ on $\frakg$.
}
\end{exam}

The following result is interesting considering the importance of Baxter Lie algebra in integrable systems~\cite{Bo,Se}.

\begin{coro}
Let $(\frak{g},R)$ be a Baxter Lie algebra, that is,
$R:\frak{g}\to\frak{g}$ satisfies Eq.~$($\ref{eq:-1mcyb}$)$. Define
the following operations on the underlying vector space of
$\frak{g}$ by
$$
[x,y]\equiv[x,y]_{\frak{g}},\quad
x\circ_{\pm}y\equiv \Big[\Big(\frac{R\pm1}{\mp2}\Big)(x),y\Big]_{\frak{g}},\quad
\forall x,y\in\frak{g}.
$$
Then $(\frak{g},[,],\circ_{\pm})$ are PostLie
algebras.\label{co:baxterpo}
\end{coro}

\begin{proof}
From the discussion at the end of Section~\ref{ss:idbaxter}, we show that
$(R\pm1)/(\mp2)$ both are Rota-Baxter operators of weight 1. So
the conclusion follows from Corollary~\ref{co:rotapo}.
\end{proof}

By Corollary~\ref{co:remark} and
Theorem~\ref{thm:depostlie}, we also obtain the following close relation between quasitriangular Lie bialgebras and PostLie algebras.
\begin{coro}
Let $(\frak{g},r)$ be a quasitriangular Lie bialgebra. Define
$\beta\in\frak{g}\otimes\frak{g}$ by Eq.~$($\ref{eq:alphabeta}$)$.
Then
$$
[a^*,b^*]\equiv-2{\rm ad}^*(\beta(a^*))b^*,\quad a^*\circ
b^*\equiv{\rm ad}^*(r(a^*))b^*,\quad \forall a^*,b^*\in\frak{g}^*,
$$
defines a PostLie algebra structure on $\frak{g}^*$. If in addition,
$r$ regarded as a linear map from $\frak{g}^*$ to $\frak{g}$ is
invertible, then the following operations define a compatible
PostLie algebra structure on $\frak{g}$:
$$
[x,y]\equiv-2r({\rm ad}^*(\beta(r^{-1}(x)))r^{-1}(y)),\quad x\circ
y\equiv r({\rm ad}^*(x)r^{-1}(y)),\quad \forall x,y\in\frak{g}.
$$
\label{co:xia1}
\end{coro}

It is obvious that for any Lie algebra $(\frak g,[, ])$, $(\frak
g,[,],-[,])$ is a PostLie algebra. Moreover, we have the following
conclusion.

\begin{theorem}
Let $(\frak{g},[,])$ be a semisimple Lie algebra. Then any PostLie
algebra structure $(\frak g, [,], \circ)$ (on $\frakg,[,])$) is given by
$$
x\circ y=[f(x), y],\;\;\forall x,y\in \frak g,
$$
where $f:\frak g\rightarrow \frak g$ is a Rota-Baxter operator of weight 1.
\end{theorem}

\begin{proof}
Let $L_\circ$ be the left multiplication operator with respect to $\circ$, that is,
$L_\circ(x) y=x\circ y$ for any $x,y\in \frak g$. Then by
Eq.~(\ref{eq:polie4}), $L_\circ$ is a derivation of the Lie algebra
$\frak g$. Since $\frak g$ is semisimple, every derivation of $\frak
g$ is inner. Therefore, there exists a linear map $f:\frak
g\rightarrow \frak g$ such that
$$x\circ y=L_\circ(x)y={\rm ad}f(x) y=[f(x),y],\;\;\forall x,y\in
\frak g.$$ Moreover, by Eq.~(\ref{eq:polie3}), we see that
$$[[f(y),f(z)],x]=[f([f(y),z]+[y,f(z)]+[y,z]),x],\;\;\forall x,y,z\in \frak
g.$$ Since the center of $\frak g$ is zero, $f$ is a Rota-Baxter operator of weight 1.
\end{proof}

\begin{remark}
{\rm In fact, the above conclusion can be extended to a Lie algebra
$\frak g$ satisfying that the center of $\frak g$ is zero and every
derivation of $\frak g$ is inner (such a Lie algebra is called {\bf
complete} ~\cite{M}). On the other hand, note that $f$ is a
Rota-Baxter operator of weight 1 if and only if $R=2f+1$ is an \tto
$\calo$-operator with \bop ${\rm id}:\frak{g}\to\frak{g}$ of \bwt
$-1$, i.e., $R$ satisfies Eq.~(\mref{eq:-1mcyb}). In particular, the
classification of the linear maps satisfy Eq.~(\mref{eq:-1mcyb}) for
every complex semisimple Lie algebra was given in \cite{Se}.}
\end{remark}

\subsection{Dendriform trialgebras and PostLie algebras: a
commutative diagram}\label{subsec:tri}
Dendriform dialgebras~\mcite{Lo} and trialgebras~\mcite{LR} are introduced with motivation from algebraic $K$-theory and topology.
Dendriform dialgebras are known to give pre-Lie algebras. We will show that a more general correspondence holds between dendriform trialgebras and PostLie algebras.
\begin{defn}
{\rm  (\cite{LR})\quad A {\bf dendriform trialgebra}
$(A,\prec,\succ,\cdot)$ is a vector space $A$ equipped with three
bilinear operations $\{\prec,\succ,\cdot\}$ satisfying the following
equations:
$$
(x\prec y)\prec z=x\prec(y\star z),\quad (x\succ y)\prec
z=x\succ(y\prec z),
$$
$$
(x\star y)\succ z=x\succ(y\succ z),\quad (x\succ y)\cdot
z=x\succ(y\cdot z),
$$
$$
(x\prec y)\cdot z=x\cdot(y\succ z),\quad (x\cdot y)\prec
z=x\cdot(y\prec z),\quad (x\cdot y)\cdot z=x\cdot(y\cdot z),
$$
 for
$x,y,z\in A$. Here $\star\equiv\prec+\succ+\cdot$. }
\end{defn}

According to~\cite{LR}, the product given by $x\star y=x\prec
y+x\succ y+x\cdot y$ defines an associative product on $A$.
Moreover, if the operation $\cdot$ is trivial, then a dendriform
trialgebra is a {\bf dendriform dialgebra} ~\cite{Lo}.

\begin{prop}
 Let $(A,\prec,\succ,\cdot)$ be a dendriform
trialgebra. Then  the products
\begin{equation}
 [x,y]\equiv x\cdot y-y\cdot
x,\quad x\circ y\equiv x\succ y-y\prec x,\;\;\forall x,y\in
A,\mlabel{eq:dentripost}
\end{equation}
make $(A,[,],\circ)$ into a PostLie algebra.
\end{prop}

\begin{proof}
We will only prove Axiom~(\ref{eq:polie4}).
The other axioms are similarly proved. For any $x,y,z\in A$, we have
\begin{eqnarray*}
& &z\circ[x,y]-[z\circ x,y]-[x,z\circ y]\\
&=&z\succ(x\cdot y-y\cdot x)-(x\cdot y-y\cdot x)\prec z-(z\succ
x-x\prec z)\cdot y+y\cdot(z\succ x-x\prec z)-\\
& &x\cdot(z\succ y-y\prec
z)+(z\succ y-y\prec z)\cdot x\\
&=&z\succ(x\cdot y)-(z\succ x)\cdot y-z\succ(y\cdot x)+(z\succ
y)\cdot x-(x\cdot y)\prec z+x\cdot(y\prec z)+\\ & &(y\cdot x)\prec
z-y\cdot(x\prec z)+(x\prec z)\cdot y-x\cdot(z\succ y)+y\cdot(z\succ
x)-(y\prec z)\cdot x=0.
\end{eqnarray*}
\end{proof}

It is easy to see that  Eq.~(\mref{eq:postlielie}) and Eq.~(\mref{eq:dentripost}) fit into the commutative
diagram:
$$
\xymatrix{
\text{Dendriform trialgebra} \ar[rr]^{x\prec y + x\succ y + x\cdot y}
\ar[dd]^{x\circ y =x\succ y - y\prec x }_{[x,y]=x\cdot y -y\cdot x} &&
\mbox{Associative algebra} \ar[dd]^{x\star y - y\star x} \\
&&\\
\mbox{PostLie algebra} \ar[rr]^{x\circ y-y\circ x+[x,y]} && \text{Lie algebra}
}
$$
When the operation $\cdot$ of the dendriform trialgebra and the bracket $[,]$ of the PostLie algebra are trivial, we obtain the following commutative diagram introduced in~\cite{C}
(see also~\cite{Ag, Ag1}):
$$
\xymatrix{ \mbox{Dendriform dialgebra} \ar[rr]^{x\prec y + x\succ y} \ar[dd]^{x\succ y - y\prec x} && \mbox{Associative algebra} \ar[dd]^{x\star y - y\star x}\\
&& \\
\mbox{Pre-Lie algebra} \ar[rr]^{x\circ y - y\circ x} && \mbox{Lie algebra}
}
$$

\section{\Triples and examples of non-abelian generalized Lax pairs}
\mlabel{sec:triple}
Our primary goal in this section is to apply our study of PostLie algebras in Section~\mref{sec:postlie} to study integrable systems.
To construct non-abelian generalized Lax pairs, we formulate the setup of a \triple that is consistent with the classical $r$-matrix approach to integrable systems~\cite{CP,Ko,Se}. We then show that new situations where this setup applies are provided by PostLie algebras from Rota-Baxter operators on complex simple Lie algebras.

\subsection{\Triple and a typical example of non-abelian generalized Lax pairs}
\label{ss:examlax}

We introduce the following concept to obtain self-dual nonabelian generalized Lax pairs.
\begin{defn}
{\rm
A {\bf \triple} consists of the following data $(\frakg,[\,,\,]_0,\rho,[\,,\,],\frakB,r,\lambda)$ where
\begin{enumerate}
\item
$(\frakg,[\,,\,]_0)$ is a Lie algebra;
\item
$[\,,\,]$ is another Lie bracket on the underlying vector space of
$\frak{g}$ and $\rho:\frakg\to \frak{gl}(\frakg)$ is a
representation of $(\frakg,[\,,\,]_0)$ such that
$(\frak{g},[\,,\,],\rho)$ is a $(\frak{g},[,]_0)$-Lie algebra.
Denote $x\cdot y\equiv\rho(x)y$, for any  $x,y\in \frakg$;
\item
$\mathfrak{B}:\frak g\otimes \frak g\rightarrow \mathbb R$ is a nondegenerate symmetric  bilinear form such that
Eq.~(\ref{eq:biform}) and Eq.~(\ref{eq:rhobiform}) hold for
$(\frak{a},[,]_{\frak{a}})=(\frak{g},[,])$.
\item
$r$ is in $\frak{g}\otimes\frak{g}$ such that the corresponding
linear map $r:\frak{g}^*\to\frak{g}$ through
Eq.~(\ref{eq:idenrmap}) has the property that the following bilinear operation defines a
Lie bracket on $\frak g$:
\begin{equation}
[x,y]_r\equiv\tilde{r}(x)\cdot y-\tilde{r}(y)\cdot
x+\lambda[x,y],\quad \forall x,y\in\mathfrak{g},
\mlabel{eq:tildebr}
\end{equation}
for certain $\lambda\in\mathbb{R}$ and for $\tilde{r}\equiv
r\varphi:\frak{g}\to\frak{g}$ where $\varphi$ is defined by
Eq.~(\ref{eq:definelinearmap}). \label{it:tilde}
\end{enumerate}
}
\mlabel{de:as1}
\end{defn}
A \triple is so named because of the three Lie algebra structures $[\,,\,]_0$, $[\,,\,]$ and $[\,,\,]_r$ on the same underlying vector space $\frakg$.
It often happens that the invariant condition in Eq.~(\mref{eq:biform}) implies Eq.~(\mref{eq:rhobiform}), so Eq.~(\mref{eq:biform}) is enough in a \triple. This is the case in the following classical example. This is also the case of PostLie algebras considered in Section~\mref{ss:postlielax}.
\begin{exam}
{\rm An example of \triple is the following well-known setting considered by
Semonov-Tian-Shansky~\mcite{CP,Ko,Se} in integrable systems. Let
$(\frakg,[\,,\,]_0)$ be a semisimple Lie algebra. Let $\rho={\rm
ad}$ be the adjoint representation. Let $(\frakg,[\,,\,])$ be
$(\frakg,[\,,\,]_0)$ and let $\frakB(\,,\,)$ be its Killing form.
Let $r$ be a skew-symmetric solution of the {\bf generalized
classical Yang-Baxter equation (GCYBE)}:
$$
({\rm ad}(x)\otimes \id\otimes \id+\id\otimes{\rm ad}(x)\otimes \id+\id\otimes
\id\otimes{\rm
ad}(x))([r_{12},r_{13}]+[r_{12},r_{23}]+[r_{13},r_{23}])=0, \quad
\forall x\in \frakg.$$ Then Eq.~(\mref{eq:tildebr}) with $\lambda=0$
defines a Lie bracket on the underlying vector space of $\frakg$. }
\mlabel{ex:se}
\end{exam}

\begin{remark}
{\rm
\begin{enumerate}
\item
Let $G$ be a simply connected Lie group whose Lie algebra is $\frakg$.
Then any representation $\rho:\frakg\to \frak{gl}(\frakg)$ is
determined by a left invariant flat connection $\nabla$ on $G$
through
$$
\rho(x)y\equiv\nabla_{\hat{x}}\hat{y}(e),\quad \forall
x,y\in\frak{g}.
$$
Here $\hat{x},\hat{y}$ are the left invariant vector fields
generated by $x,y\in\frak{g}$ and $e$ is the identity element of
$G$.
Moreover, a bilinear form $\frakB$ satisfying Eq.~(\mref{eq:rhobiform}) for $(\frak{a},[,]_{\frak{a}})=(\frak{g},[,])$ corresponds to a left
invariant pseudo-Riemannian metric which is compatible with the
connection $\nabla$ ~\cite{Mi}.
\item
By the study in Section~\ref{sec:lax}, an obvious ansatz satisfies
condition~(\ref{it:tilde}) in Definition~\mref{de:as1} is that $\tilde{r}$ is an \tto $\calo$-operator of weight $\lambda$ with \bop $\beta$ of \bwt $(\nu,\kappa,\mu)$ for $\nu\neq 0$.
\end{enumerate}
}
\end{remark}
For a \triple,  there exists a {\bf Lie-Poisson structure}~\mcite{Vai} on
$\mathfrak{g}^*$, defined by
\begin{equation}
\{f,g\}_r(a^*)\equiv\langle [df(a^*),dg(a^*)]_r,a^*\rangle ,\quad
\forall f,g\in C^{\infty}(\mathfrak{g}^*),
a^*\in\mathfrak{g}^*.\label{eq:kks}
\end{equation}


\begin{prop}
Given a \triple $(\frakg,[\,,\,]_0,\rho,[\,,\,],\frakB,r,\lambda)$ in Definition~\mref{de:as1}, any two smooth functions on $\frakg^*$ that are invariant under the dual representation of $\rho$
and the coadjoint representation of $(\frakg,[\,,\,])$ are in involution with
respect to the Lie-Poisson structure.
\end{prop}
\begin{proof}
If $f$ and $g$ are two smooth functions on $\mathfrak{g}^*$ that
are invariant under the dual representation of $\rho$ and
the coadjoint representation of $\frakg$, then
\begin{eqnarray*}
\{f,g\}_r(a^*)&=&\langle \rho(\tilde{r}(df(a^*)))dg(a^*),a^*\rangle
-\langle \rho(\tilde{r}(dg(a^*)))df(a^*),a^*\rangle+\lambda\langle [df(a^*),dg(a^*)],a^*\rangle \\
&=&-\langle dg(a^*),\rho^*(\tilde{r}(df(a^*)))a^*\rangle +\langle
df(a^*),\rho^*(\tilde{r}(dg(a^*)))a^*\rangle+\lambda\langle
df(a^*),{\rm ad}^*(dg(a^*))a^*\rangle\\
&=&0,
\end{eqnarray*}
as needed.
\end{proof}

 The above proposition motivates us to consider Hamiltonian systems on $\mathfrak{g}^*$ with
 the Lie-Poisson structure $\{,\}_r$.

\begin{theorem}
Let a \triple $(\frakg,[\,,\,]_0,\rho,[\,,\,],\frakB,r,\lambda)$ be given. Let $\mathcal{H}$ (the Hamiltonian) be a
smooth function on $\mathfrak{g}^*$ which is invariant under the
dual representation of $\rho$ and the coadjoint
representation of $(\frakg,[\,,\,])$. Let $\{e_i\}_{1\leq i\leq{\rm dim}\frak{g}}$ be a basis of $\mathfrak{g}$
with dual basis
$\{e^i\}_{1\leq i\leq{\rm dim}\frak{g}}$ with respect to $\mathfrak{B}$. Let
\begin{equation}\Omega\equiv\sum_ie_i\otimes
e^i \in \mathfrak{g}\otimes\mathfrak{g}\label{eq:casimir123}
\end{equation}
be the Casimir element. Let
$L,M:\mathfrak{g}^*\rightarrow\mathfrak{g}$ be smooth maps
defined by $L(a^*)=(a^*\otimes 1)(\Omega)$ and
$M(a^*)=\tilde{r}(d\mathcal{H}(a^*))$, $a^*\in \frakg^*$. Then
$(\mathfrak{g},\rho,\mathfrak{g},L,M)$ is a self-dual nonabelian
generalized Lax pair for the Hamiltonian system
$(\frakg^*,\{\,,\,\}_r,\calh)$ in the sense of
Definition~\ref{de:nonlax}. \mlabel{thm:laxex}
\end{theorem}
\begin{proof}
For any $f\in
C^{\infty}(\mathfrak{g}^*)$, we have
\begin{eqnarray*}
\frac{d}{dt}f(a^*)&=&\{\mathcal{H},f\}_r\\
&=&\langle \rho(\tilde{r}(d\mathcal{H}(a^*)))df(a^*),a^*\rangle -
\langle \rho(\tilde{r}(df(a^*)))d\mathcal{H}(a^*),a^*\rangle+
\lambda\langle[d\mathcal{H}(a^*),df(a^*)],a^*\rangle \\
&=&-\langle df(a^*),\rho^*(\tilde{r}(d\mathcal{H}(a^*)))a^*\rangle
,\;\;\forall a^*\in\mathfrak{g}^*.
\end{eqnarray*}Since $\mathfrak{B}$ satisfies Eq.~(\ref{eq:rhobiform})
for $(\frak{a},[,]_{\frak{a}})=(\frak{g},[,])$, it is easy to show
that (cf. Lemma~\ref{le:frosy})
\begin{equation}(\rho(x)\otimes \id+\id\otimes\rho(x))\Omega=0,\quad \forall
x,y\in\frak{g}.\label{eq:casimir}
\end{equation}
 Then
{\small \begin{equation}
\frac{d}{dt}L(a^*)=-((\rho^*(\tilde{r}(d\mathcal{H}(a^*)))a^*)\otimes
\id)(\Omega)=(a^*\otimes 1)((\rho(M(a^*))\otimes
\id)\Omega)=-(a^*\otimes 1)((\id\otimes\rho(M(a^*)))\Omega).
\end{equation}}
Hence
$$
 \frac{d}{dt}L(a^*)=-\rho(M(a^*))((a^*\otimes
1)(\Omega))=-\rho(M(a^*))L(a^*).
$$
Therefore $(\mathfrak{g},\rho,\mathfrak{g},L,M)$ is a self-dual
nonabelian generalized Lax pair.
\end{proof}
The invariant condition under the dual representation of $\rho$ holds automatically in some interesting cases, such as in Example~\mref{ex:se} and Section~\mref{ss:postlielax}. This is also true for Corollary~\mref{co:naansatz}.
\begin{remark}
{\rm
Consider the \triple in Example~\mref{ex:se} and take $\mathcal{H}$ to be a smooth function on $\mathfrak{g}^*$ which is invariant under the coadjoint representation of $(\frakg,[\,,\,])$.
Applying Theorem~\mref{thm:laxex}, we have
$$
\frac{d}{dt}L(a^*)=[L(a^*),M(a^*)],\quad \forall a^*\in\frak{g}^*,
$$
that is, $(L,M)$ is a {\bf Lax pair} in the ordinary sense
~\cite{CP}. }
\end{remark}

We next show that $(\mathfrak{g},\rho,\mathfrak{g},L,M)$ admits certain ``nonabelian generalized $r$-matrix
ansatz". First, the Poisson bracket of smooth functions on
$\mathfrak{g}^*$ defined by Eq.~(\ref{eq:kks}) can be extended to
$\mathfrak{g}$-valued functions in an obvious way: with the
notations as above, let $E$ and $F$ be two $\mathfrak{g}$-valued
smooth functions on $\mathfrak{g}$ such that
$$E=\sum_sE_se^s,\quad F=\sum_sF_se^s,
$$
where $E_s,F_s\in
C^{\infty}(\mathfrak{g}^*)$, then
$$\{E,F\}_r=\sum_{s,t}\{E_s,F_t\}_re^s\otimes e^t.$$
Suppose
that $r$ is skew-symmetric (resp. symmetric) and
$$
r=\sum_{s,t}a_{st}e_s\otimes e^t=-\sum_{s,t}a_{ts}e^s\otimes
e_t\;\;(\text{resp. } r=\sum_{s,t}a_{st}e_s\otimes e^t=\sum_{s,t}a_{ts}e^s\otimes
e_t). $$
Then
$\tilde{r}(e_s)=r(\varphi(e_s))=-\sum_ta_{ts}e_t$
(resp. $\tilde{r}(e_s)=r(\varphi(e_s))=\sum_ta_{ts}e_t$). Set
$[e_s,e_t]=\sum_{k}d_{st}^ke_k$,
$[e_s,e^t]=\sum_{k}\tilde{d}_{st}^ke^k$ and $e_l\cdot
e^s=\sum_{t}c_{ls}^te^t$. Since $L(a^*)=\sum_sL_s(a^*)e^s$, where
$L_s(a^*)=\langle e_s,a^*\rangle $, we have
\begin{eqnarray*}
\{L,L\}_r(a^*)&=&\sum_{s,t}\{L_s,L_t\}_r(a^*)e^s\otimes e^t=
\sum_{s,t}\langle [dL_s(a^*),dL_t(a^*)]_r,a^*\rangle e^s\otimes e^t\\
&=&\sum_{s,t}\langle [e_s,e_t]_r,a^*\rangle e^s\otimes
e^t=\sum_{s,t}\langle \tilde{r}(e_s)\cdot e_t-\tilde{r}(e_t)\cdot
e_s+\lambda[e_s,e_t],a^*\rangle e^s\otimes e^t\\
&=&\sum_{s,t,l}\langle -a_{ls}e_l\cdot e_t+a_{lt}e_l\cdot
e_s,a^*\rangle e^s\otimes e^t+\lambda\sum_{s,t,k} d_{st}^k\langle
e_k,a^*\rangle e^s\otimes e^t.
\end{eqnarray*}
$$
(\text{resp. }\{L,L\}_r(a^*)=\sum_{s,t,l}\langle a_{ls}e_l\cdot
e_t-a_{lt}e_l\cdot e_s,a^*\rangle e^s\otimes
e^t+\lambda\sum_{s,t,k}d_{st}^k\langle e_k,a^*\rangle e^s\otimes
e^t)
$$
However, by Eq.~(\ref{eq:casimir}) we have
$$\sum_se_l\cdot e_s\otimes e^s=-\sum_{s}e_s\otimes e_l\cdot e^s.
$$
Letting $a^*\otimes 1$
act on both sides of the above equation, we see that
$$
\sum_{s}\langle a^*,e_l\cdot e_s\rangle e^s=-\sum_s\langle
a^*,e_s\rangle e_l\cdot e^s.
$$
Therefore
\begin{eqnarray*}
\langle -a_{ls}e_l\cdot e_t,a^*\rangle e^s\otimes e^t&=&\langle
a_{ls}e_t,a^*\rangle e^s\otimes e_l\cdot e^t,\\
\langle a_{lt} e_l\cdot e_s,a^*\rangle e^s\otimes e^t&=&-\langle
a_{lt} e_s,a^*\rangle e_l\cdot e^s\otimes e_t.
\end{eqnarray*}
Furthermore, since
$\frak{B}([e_s,e_t],e^k)=-\frak{B}(e_t,[e_s,e^k])$, we have
$d_{st}^k=-\tilde{d}_{sk}^t$.
In conclusion, we obtain the ``nonabelian generalized $r$-matrix ansatz" that we are looking for (Eq.~(\ref{eq:ansatz})).
\begin{theorem}
When $r$ is skew-symmetric (resp. symmetric), the self-dual
nonabelian generalized Lax pair in Theorem~\mref{thm:laxex} satisfies
$$
\{L,L\}_r=\sum_{s,t,l,k}\{a_{ls}c_{lk}^t\langle e_k,a^*\rangle
-a_{lt}c_{lk}^s\langle e_k,a^*\rangle
-\lambda\tilde{d}_{sk}^t\langle e_k,a^*\rangle \}e^s\otimes
e^t.
$$
$$
(\text{resp. }\{L,L\}_r=\sum_{s,t,l,k}\{-a_{ls}c_{lk}^t\langle e_k,a^*\rangle
+a_{lt}c_{lk}^s\langle e_k,a^*\rangle
-\lambda\tilde{d}_{sk}^t\langle e_k,a^*\rangle \}e^s\otimes e^t)
$$
\mlabel{thm:naansatz}
\end{theorem}

Thus by Proposition~\ref{pp:pansatz}, we have
\begin{coro}
With the conditions in Theorem~\mref{thm:naansatz}, for any two smooth functions $f$ and $g$ on $\frakg$ that are invariant under the representation $\rho$ and the adjoint representation of $(\frakg,[\,,\,])$, we have
$\{f\circ L,g\circ L\}_r=0$. \mlabel{co:naansatz}
\end{coro}

\subsection{The case of PostLie algebras}
\mlabel{ss:postlielax}
We now apply Rota-Baxter operators and PostLie algebras to give an example of \triple.

\begin{theorem}
Let $(\frak{g},[,]_{\frak{g}})$ be a complex simple Lie algebra. Let
$R:\frak{g}\to\frak{g}$ be a Rota-Baxter operator of weight $1$. Let
$([\,,\,],\circ)$ denote the PostLie algebra structure on $\frak{g}$
given by Eq.~$($\ref{eq:rotaconpost}$)$ for $\lambda=1$. Let
$(\frak{g},\rho,[\,,\,])$ denote the $(\frak{g},\{,\})$-Lie algebra
given by Lemma~\mref{lem:postglie}. Let $\frak{B}$ denote the
Killing form on $\frak{g}$. Suppose there exists an
$r\in\frak{g}\otimes\frak{g}$ such that
\begin{equation}
[x,y]_r\equiv\rho(\tilde{r}(x))y-\rho(\tilde{r}(y))x+\tilde{\lambda}[x,y]=[R(\tilde{r}(x)),y]_{\frak{g}}+
[x,R(\tilde{r}(y))]_{\frak{g}}+\tilde{\lambda}[x,y]_{\frak{g}},\quad \forall
x,y\in\frak{g},\label{eq:postlbr123}
\end{equation}
defines a Lie bracket on the underlying vector space of $\frak{g}$,
where $\tilde{\lambda}\in\mathbb{R}$ and $\tilde{r}\equiv
r\varphi:\frak{g}\to\frak{g}$ and $\varphi$ is defined by
Eq.~$($\ref{eq:definelinearmap}$)$. Then
\begin{enumerate}
 \item
 $(\frakg,\{\,,\,\},\rho,[\,,\,],\frakB,r,\tilde{\lambda})$ is a \triple.
\mlabel{it:postri}
\item
Let $\mathcal{H}$ (the Hamiltonian) be a smooth function on
$\mathfrak{g}^*$ which is invariant under the coadjoint
representation of $(\frakg,[\,,\,])$. Let $\Omega$ be the Casimir
element in Eq.~$($\mref{eq:casimir123}$)$. Let
$L,M:\mathfrak{g}^*\rightarrow\mathfrak{g}$ be smooth maps
defined by $L(a^*)=(a^*\otimes 1)(\Omega)$ and
$M(a^*)=\tilde{r}(d\mathcal{H}(a^*))$, $a^*\in \frakg^*$. Then
$(\mathfrak{g},\rho,\mathfrak{g},L,M)$ is a self-dual nonabelian
generalized Lax pair for the Hamiltonian system
$(\frakg^*,\{\,,\,\}_r,\calh)$ where $\{\,,\,\}_r$ is the
Lie-Poisson structure defined in Eq.~$($\mref{eq:kks}$)$.
\mlabel{it:postlax}
\item
If $r$ is symmetric or skew-symmetric, then for any two smooth functions $f$ and $g$ on $\frakg$ that are invariant under the adjoint representation of $(\frakg,[\,,\,])$, we have $\{f\circ L,g\circ L\}_r=0$.
\mlabel{it:postinv}
\end{enumerate}
\mlabel{thm:posttri}
\end{theorem}
\begin{proof}
(\mref{it:postri}) Since $\frak{B}$ is the Killing form, it satisfies Eq.~(\ref{eq:biform}) for
$(\frak{a},[,]_{\frak{a}})=(\frak{g},[,])$. Moreover, we have
$$\frak{B}([R(x),y],z)+\frak{B}(y,[R(x),z])=0\Leftrightarrow\frak{B}(\rho(x)y,z)+\frak{B}(y,\rho(x)z)=0,\quad
\forall x,y,z\in\frak{g},$$
that is, $\frak{B}$ also satisfies
Eq.~(\ref{eq:rhobiform}) for $(\frak{a},[,]_{\frak{a}})=(\frak{g},[,])$.
\medskip

\noindent
(\mref{it:postlax})
If $\mathcal{H}$ is a smooth function which is invariant under the
coadjoint action of $G$, then $\mathcal{H}$ is also invariant under the
dual representation of $\rho$ since for any $x\in\frak{g},a^*\in\frak{g}^*$,
$$\langle d\mathcal{H}(a^*),\rho^*(x)(a^*)\rangle=-\langle
[R(x),d\mathcal{H}(a^*)],a^*\rangle=\langle d\mathcal{H}(a^*),{\rm
ad}^*(R(x))a^*\rangle=0.$$
By Theorem~\mref{thm:laxex},
$(\mathfrak{g},\rho,\mathfrak{g},L,M)$ is a self-dual nonabelian
generalized Lax pair.
\medskip

\noindent
(\mref{it:postinv})
In this case, $f$ and $g$ are also invariant under the representation $\rho$ since by definition $\rho(x)y=[R(x),y]$, for any $x,y\in\frakg$. Then the conclusion follows from Corollary~\mref{co:naansatz}.
\end{proof}

\section*{Appendix: \Tto $\calo$-operators and affine geometry on Lie groups}

In this appendix, motivated by~\cite{Bo}, we provide a geometric
explanation of the \tto $\calo$-operators. Let $K$ be a simply connected Lie group whose Lie
algebra is $\frak{k}$. Let $\nabla$ be a left invariant connection
on $K$, which, according to~\cite{KN}, is specified by a linear map
$\tilde{r}:\frak{k}\to\frak{gl}(\frak{k})$ through
$$
\tilde{r}(x)\cdot y\equiv\nabla_{\hat{x}}\hat{y}(e),\quad \forall
x,y\in\frak{k},
$$
where $\hat{x},\hat{y}$ are the left invariant vector fields
generated by $x,y\in\frak{k}$ respectively and $e$ is the identity
element of $K$. Define a linear map
$r:\frak{k}\to\frak{gl}(\frak{k})$ by
$$
r(x)\cdot
y\equiv\nabla_{\hat{x}}\hat{y}(e)-\frac{\lambda}{2}[x,y]_{\frak{k}}=\tilde{r}(x)\cdot
y- \frac{\lambda}{2}[x,y]_{\frak{k}},\quad \forall x,y\in\frak{k}.
$$
Let $\frak{g}$ be the Lie subalgebra of $\frak{gl}(\frak{k})$
generated by all $r(x)$. Then $r$ is a linear map from $\frak{k}$ to
$\frak g$. Furthermore, for any $x,y\in\frak{k}$, we have
\begin{eqnarray*}
[x,y]_R&\equiv& r(x)\cdot y-r(y)\cdot x+\lambda [x,y]_{\frak{k}}\\
&=&\tilde{r}(x)\cdot
y-\frac{\lambda}{2}[x,y]_{\frak{k}}-\tilde{r}(y)\cdot
x+\frac{\lambda}{2}[y,x]_{\frak{k}}+\lambda[x,y]_{\frak{k}}\\
&=&\tilde{r}(x)\cdot y-\tilde{r}(y)\cdot
x=\nabla_{\hat{x}}\hat{y}(e)-\nabla_{\hat{y}}\hat{x}(e).
\end{eqnarray*}
So if $[,]_R$ defines a Lie bracket on the underlying vector space
of $\frak{k}$ and $K_R$ denotes the corresponding simply
connected Lie group, then the left invariant connection determined
by
$$
\nabla_{\hat{x}}\hat{y}(e)=r(x)\cdot
y+\frac{\lambda}{2}[x,y]_{\frak{k}}
$$
is torsion-free, where $x,y\in\frak{k}$ and $e$ is the identity
element of $K_R$. Now we assume that $\frak{k}$ is a $\frak{g}$-Lie
algebra, that is, the image of $r$ belongs to ${\rm
Der}_{\mathbb{R}}(\frak{k})$, the Lie subalgebra consisting of the
derivations of $\frak{k}$. This is equivalent to
$$
\nabla_{\hat{x}}([y,z]_{\frak{k}})^{\wedge}(e)=[\nabla_{\hat{x}}\hat{y}(e),z]_{\frak{k}}+
[y,\nabla_{\hat{x}}\hat{z}(e)]_{\frak{k}},\quad \forall
x,y,z\in\frak{k}.
$$
Next we compute the curvature tensor $R(\;,\;)$ of $\nabla$:
\begin{eqnarray*}
R(\hat{x},\hat{y})\hat{z}(e)&=&(\nabla_{\hat{x}}\nabla_{\hat{y}}-
\nabla_{\hat{y}}\nabla_{\hat{x}}-\nabla_{[x,y]_{R}^{\wedge}})\hat{z}(e)\\
&=&r(x)\cdot(r(y)\cdot z)+\frac{\lambda}{2}[x,r(y)\cdot
z]_{\frak{k}}+ \frac{\lambda}{2}r(x)\cdot[y,z]_{\frak{k}}
+\frac{\lambda^2}{4}[x,[y,z]_{\frak{k}}]_{\frak{k}}
-r(y)\cdot(r(x)\cdot z)\\ & &-\frac{\lambda}{2}r(y)\cdot
[x,z]_{\frak{k}} -\frac{\lambda}{2}[y,r(x)\cdot z]_{\frak{k}}
-\frac{\lambda^2}{4}[y,[x,z]_{\frak{k}}]_{\frak{k}}-r([x,y]_R)\cdot
z-\frac{\lambda}{2}[r(x)\cdot y,z]_{\frak{k}}\\
& & +\frac{\lambda}{2}[r(y)\cdot x,z]_{\frak{k}}
-\frac{\lambda^2}{2}[[x,y]_{\frak{k}},z]_{\frak{k}}\\
&=&([r(x),r(y)]_{\frak{g}}-r([x,y]_R))\cdot
z-\frac{\lambda^2}{4}[[x,y]_{\frak{k}},z]_{\frak{k}},
\end{eqnarray*}
where the Lie bracket $[,]_{\frak{g}}$ on $\frak{g}$ is the
commutator bracket of linear transformations. Since $[,]_{\frak{k}}$
satisfies the Jacobi identity, we can re-interpret the ``Jacobi
identity condition'' in
Proposition~\ref{pp:Liestructure}.(\mref{it:cyc}) as the {\bf first
Bianchi's identity} for the curvature tensor of a torsion-free
connection.
\medskip

\noindent
{\bf Theorem. } {\em
With the same notations as above, suppose that $\frak{k}$ is a
$\frak{g}$-Lie algebra and $[,]_R$ defines a Lie bracket on the
underlying vector space of\, $\frak{k}$. Denote $K_R$ for the
corresponding simply connected Lie group. Let
$\beta:\frak{k}\to\frak{g}$ be a linear map such that $\beta$ is
$\frak{g}$-invariant of \bwt $\kappa$ and also of \bwt $\mu$, i.e., the
following equations hold
$$
\kappa\beta(\xi\cdot x)=\kappa[\xi,\beta(x)]_{\frak{g}},\quad
\mu\beta(\xi\cdot x)=\mu[\xi,\beta(x)]_{\frak{g}},\quad  \forall
\xi\in\frak{g},x\in\frak{k}.
$$
Let $r$ and $\beta$ satisfy Eq.~$($\ref{eq:gmcybe}$)$. Then the
corresponding curvature tensor (of the left invariant torsion-free
connection $\nabla$)
$$
R_e(x,y)z\equiv \kappa [\beta(x),\beta(y)]_{\frak{g}}\cdot
z+\mu\beta([x,y]_{\frak{k}})\cdot
z-\frac{\lambda^2}{4}[[x,y]_{\frak{k}},z]_{\frak{k}},\quad \forall
x,y,z\in\frak{k},
$$
is $\frak{g}$-invariant, that is,
$$
\xi\cdot R_e(x,y)z-R_e(x,y)\xi\cdot z-R_e(\xi\cdot
x,y)z-R_e(x,\xi\cdot y)z=0, \quad \forall x,y,z\in\frak{k}, \xi\in\frak{g}.
$$
In particular, setting $\xi=r(w)$, $w\in\frak{k}$, then the
curvature tensor is covariantly constant which in turn is equivalent
to the Lie group $K_R$ being an affine locally symmetric
space.
}
\medskip

\begin{proof}
The first statement depends on a direct computation. Moreover,
combining with the fact that $\nabla$ is torsion-free, we see that
$K_R$ is affine locally symmetric (cf.~\cite{KN}).
\end{proof}

\noindent
{\bf Remark. } {\rm The above conclusion is a generalization of
Theorem 3.7 in \cite{Bo}.}

\end{document}